\documentclass[11pt]{article}

\usepackage{natbib}
\usepackage{kotex}
\usepackage{amsmath}
\usepackage{amssymb}
\usepackage{amsthm}
\usepackage{bbm}
\usepackage{bm}
\usepackage{accents}

\usepackage{graphicx}

\usepackage{apalike}
\usepackage{booktabs}
\usepackage{caption}
\usepackage{subcaption}
\usepackage[left=25.4mm,right=25.4mm,top=30mm,bottom=25.4mm,a4paper]{geometry}
\usepackage{multirow}
\usepackage{booktabs}
\usepackage{siunitx}

\usepackage{pdflscape}
\usepackage{afterpage}
\usepackage{capt-of}

\theoremstyle{definition}
\newtheorem{theorem}{Theorem}
\newtheorem{lemma}[theorem]{Lemma}
\newtheorem{corollary}[theorem]{Corollary}
\newtheorem{proposition}[theorem]{Proposition}
\newtheorem{remark}[theorem]{Remark}
\newtheorem{assumption}{Assumption}
\newtheorem{definition}{Definition}
\newtheorem{notation}[definition]{Notation}

\newcommand{\E}{\mathbb E}
\newcommand{\ee}{\mathrm e}  %

\newcommand{\D}{\mathrm{d}}
\newcommand{\F}{\mathcal F}

\newcommand{\PP}{\mathbb P}
\newcommand{\Var}{\mathrm{Var}}
\newcommand{\argmax}{\mathrm{arg max}}
\newcommand{\vc}{\mathrm{vec}}

\newcommand{\Aa}{\bm{\mathrm{A}}} %
\newcommand{\BB}{\bm{\mathrm{B}}} %
\newcommand{\CC}{\bm{\mathrm{C}}} %
\newcommand{\G}{\bm{\mathrm{G}}}
\newcommand{\M}{\bm{\mathrm{M}}}
\newcommand{\N}{\bm{\mathrm{N}}}
\newcommand{\X}{\bm{\mathrm{X}}}

\newcommand{\DD}{\bm{\mathrm{D}}}
\newcommand{\Dg}{\mathrm{Dg}}
\newcommand{\V}{\bm{\mathrm{V}}}
\newcommand{\U}{\bm{\mathrm{u}}}
\newcommand{\Z}{\bm{\mathrm{Z}}} %
\newcommand{\ZZ}{\bm{\mathrm{Z}}} %

\begin{document}

\title{Application of Hawkes volatility in the observation of filtered high-frequency price process in tick structures}

\title{Application of Hawkes volatility in the observation of filtered high-frequency price process in tick structures}
\author{Kyungsub Lee\footnote{Department of Statistics, Yeungnam University, Gyeongsan, Gyeongbuk 38541, Korea}}



\maketitle

\begin{abstract}
	The Hawkes model is suitable for describing self and mutually exciting random events.
	In addition, the exponential decay in the Hawkes process allows us to calculate the moment properties in the model.
	However, due to the complexity of the model and formula, few studies have been conducted on the performance of Hawkes volatility.
	In this study, we derived a variance formula that is directly applicable under the general settings of both unmarked and marked Hawkes models for tick-level price dynamics.
	In the marked model, the linear impact function and possible dependency between the marks and underlying processes are considered.
	The Hawkes volatility is applied to the mid-price process filtered at 0.1-second intervals to show reliable results;
	furthermore, intraday estimation is expected to have high utilization in real-time risk management.	
	We also note the increasing predictive power of intraday Hawkes volatility over time and examine the relationship between futures and stock volatilities.
\end{abstract}




\section{Introduction}\label{sec1}

The importance of volatility as a risk measure of financial assets was first observed by  \cite{Markowitz}.
Since then, research on the historical volatility of financial asset prices has increased. 
For example, seminal studies on volatility were conducted via autoregressive heteroscedastic models, 
proposed by  \cite{Engle} and \cite{Bollerslev}, 
which typically estimates the volatilities based on daily data.
Daily data as well as intraday data have become useful for volatility estimation
with developments in computing, recording, and storage technology, 
thereby advancing the theory and practice of volatility estimation.

Based on the theory of quadratic variation of semimartingales, 
the realized volatility studied by \cite{ABDL}, which is equivalent to the integrated volatility \citep{Barndorff2002a,Barndorff2002b}
for the It\^{o} process, is a reliable volatility estimator.
It uses intraday data observed at appropriate frequencies.
These theories and practical application methodologies were further developed in subsequent studies \citep{barndorff2004power,Ait2005,Zhang2005,Hansen,fan2008spot,Ait2011}.

Another useful method for estimating volatility or risk using intraday data is based on point processes \citep{Bacryetal2013,Blanc2017,lee2017modeling,lee2017marked,herrera2020marked}.
The Hawkes process is a point process introduced by \cite{Hawkes1,Hawkes2}, 
and was initially recognized for its importance in seismology.
However, it has since been used in various fields such as finance, biology, and computer science.

Early studies that apply the Hawkes model in the field of finance include \cite{Hewlett2006}, \cite{Bowsher2007}, and \cite{Large2007}.
The Hawkes point process and similar intensity-based models have been studied in terms of price dynamics, especially for microstructure
\citep{Bauwens2009,Bacry2014,HAINAUT2020124330},
bid-ask price \citep{Lee2023}, and limit order book modeling \citep{Zheng2014,hainaut2019switching,morariu2021state},
with various applications such as optimal execution \citep{choi2021,da2021simple}.
It has also been used in other financial studies on credit risk \citep{Errais2010,Dassios2012,AITSAHALIA2015585}, 
extreme risk \citep{herrera2018point}, systemic risk \citep{Jang2020}, and derivatives pricing \citep{MA2020122645}.
Theoretical and applied studies on the Hawkes process in finance 
have been reviewed in the literature \citep{Law2015,Bacry2015hawkes,Hawkes2018}.

We compute the daily and intraday volatilities of financial asset prices using Hawkes models.
Under the mutually and self-excited unmarked Hawkes model framework,  
\cite{Bacryetal2013} derive formulas for the signature plot and the correlation between two price increments.
\cite{Fonseca2014} achieve a similar result under an unmarked Hawkes model setting, with derivations in matrix forms and linear differential equations.
Furthermore, \cite{lee2017modeling} focus on volatility estimation in a symmetric setting in the Hawkes kernel.  
A similar step-by-step calculation for the general moments of the Hawkes model is demonstrated by \cite{cui_hawkes_yi_2020}.

An extension of the basic Hawkes model for financial data can be achieved through by considering the event size.
The marked Hawkes model is extended by including the magnitude of the event or size of the jump in the model.
An early study on the application of the marked model to financial data can be found in \cite{Embrechts2011}.
\cite{lee2017marked} and \cite{Lee_Recurrent_2023} attempt to estimate volatility using the symmetric marked Hawkes process. 
Our study can be viewed as an extended version, with more relaxed parameter conditions.
Furthermore, marked Hawkes processes have been actively studied in the field of finance \citep{Chavez2012,ji2020combining,cai2020hawkes}.

Our first goal is to develop a closed-form volatility formula that can be used directly in computer programs with minimal model assumptions and parameter constraints.
The volatility formula is derived for all unmarked, dependent marked, and independent marked models, and provides a closed-form solution under minimal assumptions.
Volatility can be computed using estimates based on high-frequency stock-price data.
Our empirical study shows stable results.

Second, we test the performance of Hawkes volatility through a model estimation.
For the model estimation, we do not use raw data that records all the activities; 
instead, we use filtered data at 0.1 second intervals.
The filtered data removes meaningless information at the ultra-high frequency level,
making it easier to compute long-term volatility such as daily volatility.
The results show that Hawkes volatility exhibits high performance
with an increasing predictive power over time.
We also provide a comparative study of realized volatility.

In addition, the abundance of tick data makes it possible to measure Hawkes volatility in real time, 
which is significant advantage.
Real-time Hawkes variability enables detailed empirical analysis.
For example, 
we visualize the pattern by which inflow information contributes to the prediction of volatility as new information flows into the market.
We also examine the relationship between futures and stock market volatilities, as well as their predictability, for denser time intervals.


The remainder of this paper is organized as follows.
Section~\ref{Sec:motive} describes the motivation behind study.
Section~\ref{Sec:simple} introduces the basic unmarked Hawkes model and derives a volatility formula.
Section~\ref{Sec:mark} extends these results to the marked Hawkes model.
Section~\ref{Sec:application} presents the various empirical results.
Section~\ref{Sect:concl} concludes the paper.
The proofs are provided in Appendix~\ref{Sect:proof} and
an example estimation result can be found in Appendix~\ref{Sect:estimate}.

\newpage

\section{Motivation}~\label{Sec:motive}

As mentioned in the previous section, 
traditional volatility estimation methods that use high-frequency data, such as realized variance \citep{ABDL,Barndorff2002b}, typically assume that price changes follow a semimartingale process, such as a diffusion process \citep{mcaleer2008realized} or a diffusion process with jumps \citep{Barndorff2007}.
These methods are known to be reliable;
however, stock prices move in discrete steps because of their minimum tick size \citep{huang2001tick}.
To better address this discreteness, point process models, such as Hawkes processes, may offer a more natural approach. 
These models capture the time and clustering of trade arrivals and price movements, 
and provide an alternative perspective that enriches our understanding of market dynamics \citep{chen2022}.

The Hawkes process with an exponential kernel is a suitable fundamental model for representing tick movements. 
Studies on the microstructure of asset price movements suggest that new events such as price changes influence the intensity of future events. 
This influence dissipates over time and can be represented by exponential decay. 
Although the decay rate may not exactly follow an exponential function, 
and can align with a power law or other decreasing trends  \citep{bacry2012non,hardiman2013critical,bacry2016first}, 
the complexity of such kernels complicates volatility calculations and falls outside the scope of this study. 

Nevertheless, we explore the robustness of our method against different kernel properties through filtering effects using a simulation study.
Although the exponential Hawkes model is mathematically tractable for deriving various equations,
this becomes complex when numerous parameters are incorporated by relaxing the parameter constraints. 
Consequently, several studies focus on simpler symmetric models \citep{bacry2012non, Bacry2014,DaFonseca&Zaatour2014,lee2017modeling}.  
However, our research derives tractable formulas while relaxing the parameter constraints as much as possible.

An unmarked Hawkes process, 
which typically incorporates only the temporal sequence of events, may not sufficiently capture the complex dependencies between the size and time of each event in the tick structure.
The marked Hawkes process extends the standard Hawkes model by including marks that represent additional information such as the jump size in price changes. 
This allows for more detailed modeling of how specific characteristics of one event influence the intensity and nature of subsequent events. 
The additional information represented by the marks can include volumes, transaction sizes, bid-ask spreads, or other relevant market variables; however, in this study, we focus on the size of price changes.

In the marked Hawkes model setting, we aim to provide conditions that minimize the constraints while yielding a tractable volatility formula.
However, the incorporation of marks and their dependence structures complicates the mathematical derivations.
Therefore, to maintain mathematical tractability, we assume a linear function for the impact of jumps exceeding the basic tick size.
This assumption implies that the larger the jump size of an event, the greater its influence on future intensities.
We consider this assumption reasonably realistic.
Throughout the derivation process, we attempt to simplify the formulas as much as possible.
We present the variance formula derived under these conditions in Section~\ref{Sec:mark}.
  
With the advancements in technology, high-frequency financial data have evolved.
For example, for NYSE-listed stocks, the data were recorded with millisecond precision until July 2015. 
The recording frequency then increased to microseconds until around September 2018, and subsequently increased to nanoseconds.
Similarly, for Nasdaq-listed securities, data recording had a millisecond resolution up until July 2015, shifted to microseconds by around October 2016, and has been logged in nanoseconds since then \citep{Lee2023}.
This progression in the data recording frequency reflects an improvement in technological capabilities, allowing for a more detailed and precise analysis of market dynamics at increasingly finer resolutions.
Moreover, the sheer number of transactions has also increased over time.

However, data recorded at ultra-high frequencies inevitably include market microstructure noise, which can lead to unstable estimations if raw data are used directly for analysis. 
A simple exponential Hawkes model may not fit the raw data accurately~\citep{Lee2024}.
To address this issue, Subsection~\ref{Subsec:filter} describes a filtering method. 
Subsampling itself is a well-known concept for handling microstructure noise~\citep{AitSahalia2005}.
This technique is designed to refine data by reducing noise and enhancing the reliability of estimations, thus providing more stable and accurate results from high-frequency trading data.
This will be discussed later; however although further research is necessary to determine the optimal data-filtering method, 
the filtering approach presented in this study works sufficiently well with the proposed method.

Our model and its assumptions are summarized as follows: 
We apply a marked Hawkes model with an exponential kernel to sparsely observed high-frequency data and do not make any special assumptions such as symmetry, except for the decay parameter $\beta$. The mark represents the magnitude of price changes, and its future impact is modeled as a linear function. 
Mark distributions and intensities may or may not be dependent on one another.

\section{Basic Hawkes model}~\label{Sec:simple}

The microstructure of stock price dynamics at the intraday level can be described by two-point processes $N_i$ for $i=1,2$,
representing the times of up ($i=1$) and down ($i=2$) price movements, respectively.
In this section, we derive the formula of 
\begin{equation}
	\Var(N_1(t)- N_2(t)) \label{Eq:var}
\end{equation}
which measures the variability of price changes under the condition 
that $N_1$ and $N_2$ follow Hawkes processes with $N_1(0) = N_2(0) = 0$.
Accordingly, we derive the formulas for the moments of the intensity and counting processes of the basic bivariate Hawkes model in terms of the parameters.
The basic Hawkes model for intraday price dynamics does not include marks, which can be used to describe jump size or trade volume in microstructure price movements.
Thus, the basic Hawkes model only describes the distribution of the size of the intervals between events.

The derivation of moments in the basic Hawkes model is not entirely a new approach;
similar findings can be found in previous studies \citep{DaFonseca&Zaatour2014,lee2017modeling,cui_hawkes_yi_2020} as mentioned in Section~\ref{sec1}.
We review them in an organized manner to derive Eq.~\eqref{Eq:var}.
Hence, this formula can be used practically; 
for example, it can be applied directly to a computer programming code.
Deriving the variance formula of the basic Hawkes model is important 
because it serves as a building block for deriving a volatility formula in more complicated models.

The basic Hawkes process is defined in a probability space $(\Omega, \F, \PP)$.
There is a strictly increasing sequence of real-valued random  variables $\{ \tau_{i,n} \}$ that represent event times on the space with countable index $n$.
Assuming simple counting processes, the probability of two or more events occurring simultaneously is zero.
For practical purposes, only positive $\tau$ values are meaningful. 
However, to define the intensity processes later, 
nonpositive $\tau$ values are conceptually assumed to exist.

We use $N_i$ to denote the counting process as a stochastic process and a random measure
\[
N_i(t) = N_i((0,t]) = \sum_{n} \mathbbm{1}_{ \{ 0 < \tau_{i,n} \leq t \} } = \# \textrm{ of } \tau_{i,n} \in (0, t], \quad \textrm{for } i=1,2.
\]
The stochastic intensity $\lambda_i$ of the given counting process $N_i$ is a nonnegative and $\F$-predictable process such that
\begin{equation*}
	\E  \left[  \left. \int_s^t \lambda_i(u) \D u \right\rvert \F_s \right] =  \E[N_i(t) - N_i(u) \rvert \F_s ], \quad \textrm{for } s < t. 
\end{equation*}

\begin{assumption}[Unmarked model]~\label{Assum1}
The bivariate Hawkes processes and their corresponding intensity processes are defined as follows:
\[
\bm{N}_t = \begin{bmatrix} N_1(t) \\ N_2(t) \end{bmatrix}, \quad \bm{\lambda}_t = \begin{bmatrix} \lambda_1(t) \\ \lambda_2(t) \end{bmatrix},
\]
where
\begin{equation}
	\bm{\lambda}_t = \bm{\mu} + \int_{-\infty}^{t} \bm{h}(t-u) \D \bm{N}_u \label{Eq:lambda}
\end{equation}
with the matrix-vector multiplication in the integral;
$\bm{\mu} = [\mu_1, \mu_2]^{\top}$ is a positive constant vector representing the exogenous intensities; 
$\bm{h}$ is often called a kernel, such that
$$ \bm{h}(t) = \bm{\alpha} \circ
\begin{bmatrix}
	\ee^{-\beta_1 t} & \ee^{-\beta_1 t} \\
	\ee^{-\beta_2 t} & \ee^{-\beta_2 t}
\end{bmatrix}
$$ 
with the Hadamard (element-wise) product $\circ$;
$\bm{\alpha}$ is a positive $2\times 2$ constant matrix.
\end{assumption}

The order of the Hadamard product $\circ$ and matrix product in Eq.~\eqref{Eq:lambda} is from left to right.
Note that $\bm{\mu}$ and $\bm\alpha$ have no parameter restrictions on the requirement of positivity;
however, the decay rates $\beta$ have the same value for each $i$ to ensure the Markov property of the joint process $(\bm{N}, \bm{\lambda})$.

In a differential form, the intensity vector process is represented by
\begin{align*}
	\D \bm{\lambda}_t &= \bm{\beta} (\bm{\mu} - \bm{\lambda}_t) \D t + \bm{\alpha} \D \bm{N}_t \\
	&= \{ \bm{\beta}\bm{\mu} + (\bm{\alpha} - \bm{\beta}) \bm{\lambda}_t) \}\D t + \bm{\alpha} (\D \bm{N}_t - \bm{\lambda}_t \D t),
\end{align*}
where
$$ \bm{\beta} = \begin{bmatrix} \beta_1 & 0 \\ 0 & \beta_2 \end{bmatrix},$$
is a diagonal matrix.
In integral form,
\begin{align}
	\bm{\lambda}_t 
	={}& \bm{\lambda}_0 + \int_0^t \bm{\beta} (\bm{\mu} - \bm{\lambda}_s)\D s + \int_0^t \bm{\alpha} \D \bm{N}_s  \\
	={}& \bm{\lambda}_0 + \int_0^t \{ \bm{\beta}\bm{\mu} + (\bm{\alpha} - \bm{\beta}) \bm{\lambda}_s) \}\D s + \int_0^t \bm{\alpha} (\D \bm{N}_s - \bm{\lambda}_s \D s), \label{eq:lambda}
\end{align}
where the last term of Eq.~\eqref{eq:lambda} is a martingale.
The counting processes are right continuous, and the intensity processes are left continuous.
If the spectral radius of 
$$  \left\rvert \int_0^{\infty} \bm{h}(t) \D t \right\rvert, $$
is less than one, it is known that $\bm{N}$ has a unique stationary version.
By conditioning the Hawkes process to begin at $-\infty$, 
we assume that the unconditional distribution of the intensity processes at and after time zero will no longer change.
Therefore, we always assume that $\bm{\lambda}$ is in a steady state at time zero.
Hence, the unconditional joint distribution of 
$(\lambda_1(s), \lambda_2(s))$ and $(\lambda_1(t), \lambda_2(t))$ are identical for any time $s, t \geq 0$.

The steady state assumption implies that the intensity processes reach a long-term distribution; hence, the expectations of the intensity processes are constant.
The following formula is well-known.

\begin{proposition}~\label{Prop:E_lambda}
	Under Assumption~\ref{Assum1} and the steady state assumption,
	\begin{equation}
		\E[\bm{\lambda}_t] = (\bm{\beta} - \bm{\alpha})^{-1} \bm{\beta} \bm{\mu} = \left(\bm{\mathrm{I}} - \bm{\beta}^{-1}\bm{\alpha}\right)^{-1}\bm{\mu}. \label{Eq:E_lambda} 
	\end{equation}
\end{proposition}

\begin{proof}
See Appendix~\ref{Sect:proof}.
\end{proof}

To proceed, we use the concept of the quadratic (co)variation of semimartingales in a matrix version.
Consider the vectors of semimartingale processes $\bm{X}$ and $\bm{Y}$;
then the integration by parts of $\bm{X}_t \bm{Y}^{\top}_t$ is expressed in the following differential form:
\begin{align}
	\D (\bm{X}_t \bm{Y}^{\top}_t) &= \bm{X}_{t-} \D \bm{Y}^{\top}_t + \D(\bm{X}_t) \bm{Y}^{\top}_{t-} + \D [\bm{X}\bm{Y}^{\top}]_t,  \label{eq:quad}
\end{align}
where $\D(\bm{X}_t) \bm{Y}^{\top}_{t-}$ implies $(\bm{Y}_{t-} \D \bm{X}_t^{\top})^{\top} $
and $ [\bm{X}\bm{Y}^{\top}]_t $ denotes the matrix whose elements are the quadratic (co)variations for the entries of $\bm{X}_t \bm{Y}_t^{\top}$.

\begin{notation}~\label{Notataion1}
	If $\bm{\mathrm{x}}$ is a vector, then $\mathrm{Dg}(\bm{\mathrm{x}})$ denotes a diagonal matrix whose diagonal entries are composed of the elements of $\bm{\mathrm{x}}$.
	If $\bm{\mathrm{M}}$ is a square matrix, then $\mathrm{Dg}(\bm{\mathrm{M}})$ denotes a diagonal matrix whose diagonal entries are the diagonal parts of $\bm{\mathrm{M}}$.
	In addition, let $\mathcal{T}$ be an operator, such that
	$$ \mathcal{T}(\M) = \M + \M^{\top} $$
	for the square matrix $\M$.
\end{notation}

	The following lemma illustrates how the change over time in the expectation of the product of two point processes is expressed. 
	Based on this general formula, a system of equations for $\E[\bm{\lambda}_t \bm{\lambda}_t^{\top}]$ and $\E[\bm{\lambda}_t \bm{N}_t^{\top}]$ is established in Propositions~\ref{Prop:E_ll} and~\ref{Prop:E_lN}. Once $\E[\bm{\lambda}_t \bm{\lambda}_t^{\top}]$ and $\E[\bm{\lambda}_t \bm{N}_t^{\top}]$ have been computed, we can calculate the variance of the price change distribution over an interval.

\begin{lemma}\label{Lemma:quad}
	Consider $2\times 1$ vector processes $\bm{X}$ and $\bm{Y}$ such that
	\begin{align*}
		\D \bm{X}_t = \bm{a}_t \D t + \bm{f}_x(t) \D \bm{N}_t, \quad \D \bm{Y}_t = \bm{b}_t \D t + \bm{f}_y(t) \D \bm{N}_t
	\end{align*}
	where $\bm{a}$ and $\bm{b}$ are $2\times 1$ vector processes, and $\bm{f}_x$ and $\bm{f}_y$ are $2 \times 2$ matrix processes.
	Then
	$$ \D [\bm{X}\bm{Y}^{\top}]_t = \bm{f}_x(t) \mathrm{Dg} (\D \bm{N}_t) \bm{f}_y^{\top}(t).$$
	Thus, by Eq.~\eqref{eq:quad},
	\begin{align*}
		\D (\bm{X}_t \bm{Y}^{\top}_t) &= \bm{X}_{t-} \D \bm{Y}^{\top}_t + \D(\bm{X}_t) \bm{Y}^{\top}_{t-} + \bm{f}_x(t) \mathrm{Dg} (\D \bm{N}_t) \bm{f}_y^{\top}(t).
	\end{align*}
	In addition, when $\E[\bm{X}_t \bm{Y}^{\top}_t]$ exists,
	\begin{align*}
		\frac{\D\E[\bm{X}_t \bm{Y}^{\top}_t]}{\D t} = \E[ \bm{X}_{t-} (\bm{b}_t^{\top} + \bm{\lambda}_t^{\top} \bm{f}_y^{\top}(t))] + \E[(\bm{a}_t + \bm{f}_x(t)  \bm{\lambda}_t ) \bm{Y}^{\top}_{t-} ] 
		+ \bm{f}_x(t) \mathrm{Dg} (\E[\bm{\lambda}_t]) \bm{f}_y^{\top}(t).
	\end{align*}
	
\end{lemma}

\begin{proposition}~\label{Prop:E_ll}
	Under Assumption~\ref{Assum1} and the steady state assumption,
	we obtain the following Sylvester equation for $\E[\bm{\lambda}_t \bm{\lambda}_t^{\top}]$:
	\begin{equation}
		\E[\bm{\lambda}_t \bm{\lambda}_t^{\top}] (\bm{\alpha}- \bm{\beta} )^{\top}
		+ (\bm{\alpha}- \bm{\beta} ) \E[\bm{\lambda}_t \bm{\lambda}_t^{\top}]
		+ \E[\bm{\lambda}_t]   (\bm{\beta} \bm{\mu})^{\top}
		+ \bm{\beta} \bm{\mu} \E[\bm{\lambda}_t^{\top}]
		+ \bm{\alpha}\mathrm{Dg}(\E[\bm{\lambda}_t])\bm{\alpha}^{\top} = \bm{0} \label{Eq:E_ll}
	\end{equation}
	or using the Notation~\ref{Notataion1},
	$$
	\mathcal{T}\left((\bm{\alpha}- \bm{\beta} ) \E[\bm{\lambda}_t \bm{\lambda}_t^{\top}] + \bm{\beta} \bm{\mu} \E[\bm{\lambda}_t^{\top}] \right) + \bm{\alpha}\mathrm{Dg}(\E[\bm{\lambda}_t])\bm{\alpha}^{\top} = \bm{0}
	$$
	where $\bm{0}$ denotes the $2 \times 2$ zero matrix and $\E[\bm{\lambda}_t]$ satisfies Eq.~\eqref{Eq:E_lambda}.	
\end{proposition}

\begin{proof}
See Appendix~\ref{Sect:proof}.	
\end{proof}

Solving the Sylvester equation in Eq.~\eqref{Eq:E_ll} is equivalent to solving
\[
\{ \bm{\mathrm{I}} \otimes (\bm{\alpha}- \bm{\beta} ) + (\bm{\alpha}- \bm{\beta} )\otimes \bm{\mathrm{I}}\}  \vc \left( \E [\bm{\lambda}_t \bm{\lambda}_t^{\top}] \right) 
= - \vc  \left( \E[\bm{\lambda}_t]   (\bm{\beta} \bm{\mu})^{\top}
+ \bm{\beta} \bm{\mu} \E[\bm{\lambda}_t^{\top}]
+ \bm{\alpha}\mathrm{Dg}(\E[\bm{\lambda}_t])\bm{\alpha}^{\top} \right) 
\]
where $\otimes$ denotes the Kronecker product, and $\vc$ is the vectorization operator.
Assuming the steady state condition, the first and second moments of $\lambda$ do not depend on time.
However, $\E[\bm{\lambda}_t \bm{N}_t^{\top}]$ depends on time.

\begin{proposition}~\label{Prop:E_lN}
	Under Assumption~\ref{Assum1},
	for $\E[\bm{\lambda}_t \bm{N}_t^{\top}]$, 
	we have the following first order $2\times 2$ system of differential equations.
	\begin{align*}
		\frac{\D \E[\bm{\lambda}_t \bm{N}_t^{\top}]}{\D t} = (\bm{\alpha} - \bm{\beta} ) \E[\bm{\lambda}_t \bm{N}_t^{\top}] + \bm{\beta} \bm{\mu} \E[\bm{\lambda}_t^{\top}] t + \E[\bm{\lambda}_t \bm{\lambda}_t^{\top}] + \bm{\alpha}\mathrm{Dg}(\E[\bm{\lambda}_t]).
	\end{align*}
	Let $\V = \begin{bmatrix}\bm{\mathrm{v}}_1 & \bm{\mathrm{v}}_2 \end{bmatrix}$ be the eigenvector matrix of $\bm{\alpha} - \bm{\beta}$ 
	and $\xi_1$ and $\xi_2$ are corresponding eigenvalues.
	The solution of the system is as follows:
	\begin{align*}
		\E[\bm{\lambda}_t \bm{N}_t^{\top}] =
		\V \CC \circ \begin{bmatrix} \ee^{\xi_1 t} & \ee^{\xi_1 t} \\ \ee^{\xi_2 t} & \ee^{\xi_2 t} \end{bmatrix} + \Aa t + \BB
	\end{align*}
	where
	\begin{align}
		\Aa &= - (\bm{\alpha} - \bm{\beta} )^{-1} \bm{\beta}\bm{\mu} \E [\bm{\lambda}_t^{\top}] = \E[\bm{\lambda}_t] \E[\bm{\lambda}_t]^{\top} \label{Eq:basicA}\\
		\BB &= (\bm{\alpha} - \bm{\beta} )^{-1} \{\Aa - \E[\bm{\lambda}_t \bm{\lambda}_t^{\top} ] - \bm{\alpha}\mathrm{Dg}(\E[\bm{\lambda}_t]) \} \label{Eq:basicB}
	\end{align}
	and
	$\CC$ is a constant coefficient matrix that satisfies the initial condition:
	$$ \E[\bm{\lambda}_0 \bm{N}_0^{\top}] = \bm{0}.$$
	As $t \rightarrow \infty$, only the particular solution is significant, that is,
	\begin{align*}
		\E[\bm{\lambda}_t \bm{N}_t^{\top}] \approx
		\Aa t + \BB.
	\end{align*}
\end{proposition}

\begin{proof}
	See Appendix~\ref{Sect:proof}.
\end{proof}

Note that even if we assume that $\bm{\lambda}$ is in a steady state at time 0, $\E[\bm{\lambda}_t \bm{N}_t^{\top}]$ has an exponentially converging term.
However, the homogeneous solution of the exponential is rather insignificant with a relatively large $t$ under the condition that the spectral radius is less than 1.
Hence, $\E[\bm{\lambda}_t \bm{N}_t^{\top}]$ can be approximated as an affine form of time $t$.
With this approximation, the variance of the net number of up and down movements between zero and $t$ 
is presented below.
To compute the variance of price change, the formula should be adjusted based on the minimum tick size.

\begin{theorem}\label{Thm:simple_vol}
	We have
	\begin{align}
		\Var(N_1(t) - N_2(t)) 
		&\approx \U^{\top} \left( 2 \BB  +  \Dg(\E[\bm{\lambda}_t])  \right)\U t \label{Eq:var1}
	\end{align}
	where $\BB$ is defined by Eq.~\eqref{Eq:basicB} and $\bm{\mathrm{u}} = \begin{bmatrix} 1 & -1 \end{bmatrix}^{\top}.$
	
\end{theorem}

\begin{proof}
See Appendix~\ref{Sect:proof}.
\end{proof}

\begin{remark}\label{Remark:simple}

The variance formula presented in Eq.~\eqref{Eq:var1} is simplified under the symmetric self and mutually excited Hawkes model \citep{Bacry2014,lee2017modeling}. 
For this specific model setup, the parameters are defined as:
$$ \bm{\mu} = \begin{bmatrix} \mu \\ \mu \end{bmatrix}, \quad 
\bm{\alpha} = \begin{bmatrix} \alpha_1 & \alpha_2 \\ \alpha_2 & \alpha_1 \end{bmatrix}, \quad
\bm{\beta} = \begin{bmatrix} \beta & 0 \\ 0 & \beta \end{bmatrix}.
$$
Assuming symmetry, the mean intensities for the up and down processes are identical, leading to
\begin{equation}
\E[\lambda_t] := \E[\lambda_1(t)] = \E[\lambda_2(t)] = \frac{\mu \beta}{\beta - \alpha_1 - \alpha_2}. \label{Eq:E_lambda_simple}
\end{equation}
In addition, $\E[\bm{\lambda}_t \bm{\lambda}_t^{\top}]$ is reduced to 
\begin{equation}
\E[\bm{\lambda}_t \bm{\lambda}_t^{\top}] = (\bm{\beta} - \bm{\alpha})^{-1} \left( \frac{1}{2}\bm{\alpha} \Dg (\E [\bm{\lambda}_t]) \bm{\alpha} + \bm{\beta} \bm{\mu} \E [\bm{\lambda}_t^{\top}] 
\right) \label{Eq:Ell_reduce}
\end{equation}
and
\begin{align*}
	\BB &= (\bm{\alpha} - \bm{\beta} )^{-1} \{\Aa - \E[\bm{\lambda}_t \bm{\lambda}_t^{\top} ] - \bm{\alpha}\mathrm{Dg}(\E[\bm{\lambda}_t]) \} \\
	&= (\bm{\alpha} - \bm{\beta} )^{-1} \left\{(\bm{\beta} - \bm{\alpha})^{-1} \bm{\beta} \bm{\mu} \E [\bm{\lambda}_t^{\top}] - (\bm{\beta} - \bm{\alpha})^{-1} \left( \frac{1}{2}\bm{\alpha} \Dg (\E [\bm{\lambda}_t]) \bm{\alpha} + \bm{\beta} \bm{\mu} \E [\bm{\lambda}_t^{\top}] \right) - \bm{\alpha}\mathrm{Dg}(\E[\bm{\lambda}_t])  \right\} \\
	&= \E[\lambda_t] \left\{ \frac{1}{2}(\bm{\alpha} - \bm{\beta} )^{-2} \bm{\alpha}^2 - (\bm{\alpha} - \bm{\beta} )^{-1} \bm{\alpha}    \right\} \\
	&= \E[\lambda_t] \left( \frac{1}{2 (\beta - \alpha_1 - \alpha_2)^2 (\beta - \alpha_1 + \alpha_2)^2}
	\begin{bmatrix} (\alpha_1 - \beta)^2 + \alpha_2^2 & -2 \alpha_2 (\alpha_1 - \beta) \\ -2 \alpha_2 (\alpha_1 - \beta) & (\alpha_1 - \beta)^2 + \alpha_2^2 \end{bmatrix}
	\begin{bmatrix} \alpha_1^2 + \alpha_2^2 & 2\alpha_1\alpha2 \\
	2 \alpha_1 \alpha_2 & \alpha_1^2 + \alpha_2^2 \end{bmatrix}  \right. \\
	 & \left. \phantom{\E[\lambda_t]abc} - \frac{1}{(\beta - \alpha_1 - \alpha_2) (\beta - \alpha_1 + \alpha_2)} 
	 \begin{bmatrix} \alpha_1 - \beta & -\alpha_2 \\ -\alpha_2 & \alpha_1 - \beta \end{bmatrix}
	 \begin{bmatrix} \alpha_1 & \alpha_2 \\ \alpha_2 & \alpha_1 \end{bmatrix} \right).
\end{align*}
Given that
\[
\U^{\top} \BB \U  = \E[\lambda_t] \left( \frac{(\alpha_1 - \alpha_2)^2}{(\beta - \alpha_1 + \alpha_2)^2} + \frac{2(\alpha_1 - \alpha_2)}{\beta - \alpha_1 + \alpha_2} \right),
\]
we obtain
\begin{align*}
	\Var(N_1(t) - N_2(t)) 
	&\approx \U^{\top} \left( 2 \BB  +  \Dg(\E[\bm{\lambda}_t])  \right)\U t = 2 \E[\lambda_t] \left( \frac{\alpha_1 - \alpha_2}{\beta - \alpha_1 + \alpha_2} + 1\right)^2 t \\
	&= \frac{2\mu \beta^3 t}{(\beta - \alpha_1 - \alpha_2)(\beta - \alpha_1 + \alpha_2)^2}.
\end{align*}

\end{remark}

\section{Marked model}~\label{Sec:mark}

In this section, we introduce a marked Hawkes model that incorporates both the time and size of the asset price movement.
This section can be regarded as a general version of \cite{lee2017marked}.
As in the previous section, using an unmarked model, we present the variance formula of the marked Hawkes model.
The marked model differs from the unmarked model in two ways.
First, the effect of the jump size on future intensity should be considered.
Second, the possibility of dependence between jump size and other underlying processes, such as $\N$ and $\bm \lambda$, should be considered.

For $i = 1, 2$, let 
$$E_i = \mathbb{N} \times \{i\}$$ 
be the spaces of the mark (jump) sizes for up and down price movements.
Moreover, let
$$ E = E_1 \cup E_2 = \mathbb{N} \times \{1, 2\}.$$
In this study, the jump size space is $\mathbb{N}$, which is natural numbers, because it is a multiple of the minimum tick size defined in the stock market.

Let $\mathcal{E}_i$ be the $\sigma$-algebra defined on $E_i$.
We have a sequence of $(E_i, \mathcal{E}_i)$-valued random variables $\{Z_{i,n}\}$ in addition to the sequence of random times $\{\tau_{i, n}\}$ for each $i$.
The sequence of random variables $Z_{i,n}$ represents the movement size in the mid-price process, and $\tau_{i,n}$ represents the time of the event.
The random measure $M_i$ is defined in the product of time and jump size space, $\mathbb{\bar R} \times E_i$, such that
$$M_i(\D u \times \D z_i) = \sum_{n} \delta_{\tau_{i,n}, Z_{i,n}}(\D u \times \D z_i)$$
with the Dirac measure $\delta$, which is defined as
$$ \delta_{\tau_{i,n}, Z_{i,n}} (I \times A_i) =
\left\{\begin{array}{lr}
	1, \text{ if } \tau_{i,n} \in I \text{ and } Z_{i,n} \in A_i,  \\
	0, \text{ otherwise,}
\end{array}\right.
$$
for $A_i \subset E_i$ and any time interval $I$.
Note that $z_i$ is a variable associated with the measure $M_i$.
This measure is also called a marked point process with a mark space $(E_i, \mathcal{E}_i)$.

Consider a vector representation of the random measures
$$
\bm{M}(\D u \times \D z) = \begin{bmatrix} M_1(\D u \times \D z_1) \\ M_2(\D u \times \D z_2) \end{bmatrix}
$$
and a vector of c\`adl\`ag counting processes defined by
$$ \bm{N}_t = \int_{(0,t] \times E} \Dg(z) \bm{M}(\D u \times \D z)$$
where the integration is applied element-wise and
$$\Dg(z) = \begin{bmatrix} z_1 & 0 \\ 0 & z_2 \end{bmatrix}. $$
That is, each element of $\bm{N}_t$ counts the number of events, with the weights of jump sizes at that time for up and down mid-price movements.

\begin{assumption}[Marked model]~\label{Assum:g}
The intensity process is defined as follows.
\begin{equation}
	\bm{\lambda}_t = \bm{\mu} + \int_{(-\infty,t] \times E} \bm{h}(t-u, z) \bm{M}(\D u \times \D z)
\end{equation}
where $\bm{\mu} = [\mu_1, \mu_2]^{\top}$ denotes a positive constant vector. 
\begin{equation}
\bm{h}(t, z) = (\bm{\alpha}  + \bm{g}(z)) \circ
\begin{bmatrix}
	\ee^{-\beta_1 t} & \ee^{-\beta_1 t} \\
	\ee^{-\beta_2 t} & \ee^{-\beta_2 t}
\end{bmatrix} \label{Eq:h}
\end{equation}
and $\bm{\alpha}$ is a positive constant $2\times 2$ matrix such that
$$ \bm{\alpha} =  \begin{bmatrix} \alpha_{11} & \alpha_{12} \\ \alpha_{21} & \alpha_{22} \end{bmatrix}.$$
Unlike the unmarked model, the marked model has a function $\bm{g}$ that represents the future impact of a jump size.	
We assume a linear impact function for $\bm{g}$ such that
\begin{equation}
	\bm{g}(z) =  \bm{\eta} \circ (\Z - 1)\label{Eq:g}
\end{equation}
where $\bm{\eta}$ is a positive $2\times 2$ parameter matrix and $\bm{\mathrm{Z}}$ is a matrix such that
$$\bm{\eta} = \begin{bmatrix} \eta_{11} & \eta_{12} \\ \eta_{21} & \eta_{22} \end{bmatrix}, \quad \bm{\mathrm{Z}} = \begin{bmatrix} z_1 & z_2 \\ z_1 & z_2 \end{bmatrix}.$$
\end{assumption}

Under the aforementioned assumption, each row of the exponential kernel matrix in Eq.~\eqref{Eq:h} retains the Markov property by using a uniform decay rate, $\beta_i$. 
Although there may be several candidates for $\bm{g}$, we choose a linear impact function for simplicity.
By using $\ZZ - 1$, taking minus one for every element of $\ZZ$,  the linear impact function $\bm{g}$ is responsible only for the future effect of movements larger than the minimum tick size and, the future impact is proportional to $\bm{\eta}$.

Consider a matrix of predictable stochastic functions $\bm{f}$ such that
$$ \bm{f}(s,z) = \begin{bmatrix} f_{11}(s, z_1) & f_{12}(s, z_2) \\ f_{21}(s, z_1) & f_{22}(s, z_2) \end{bmatrix}.$$
Then, the matrix of processes is defined as
\begin{equation}
	\int_{(0,t] \times E} \bm{f}(s,z) \bm{M}(\D s \times \D z) - \int_{(0,t] \times E}  \bm{f}(s,z) \bm{\nu} (\D s, \D z), \label{Eq:comp}
\end{equation}
which is a matrix of martingales, where $\bm{\nu}$ is a vector of compensator measures for $\bm{M}$. 
In other words, every element of the matrix in Eq.~\eqref{Eq:comp} is a martingale with respect to $\mathcal F$.
In addition, we assume that the conditional probability distribution of the mark size can be separated, given that an event occurs.
Then 
$$ \int_{(0,t] \times E}  \bm{f}(s,z) \bm{\nu} (\D s, \D z)  = \int_0^t \int_E \bm{f}(s,z) \circ \bm{K}(s, \D z)  \bm{\lambda}_s  \D s $$
where $\bm{K}$ is a matrix of random measures, such that
$$ \bm{K}(t, \D z) = \begin{bmatrix} k_1 (t, \D z_1) &  k_2 (t, \D z_2) \\ k_1 (t, \D z_1) &  k_2 (t, \D z_2) \end{bmatrix} $$
and for any time $t$,
$$ \int_E k_i (t, \D z_i)  = 1 \text{ for }i = 1, 2$$
represents the conditional distribution of the mark at time $t$.

The following notations are used:
\begin{equation}
	\mathbb{K}_t [\bm{f}(t,z)] = \int_E \bm{f}(t,z) \circ  \bm{K}(t, \D z).\label{Eq:K}
\end{equation}
For example, 
$$ \mathbb{K}_t [\Z] = \begin{bmatrix} \int_E z_1 k_1 (t, \D z_1)  & \int_E z_2 k_2 (t, \D z_2)  \\ 
	\int_E z_1 k_1 (t, \D z_1) & \int_E z_2  k_2 (t, \D z_2)  \end{bmatrix}.$$
Hence, each column represents the conditional expectation of the mark size at time $t$, given that a corresponding event occurs at time $t$.
Similarly,
$$ \mathbb{K}_t [\ZZ^{\circ 2}] = \begin{bmatrix} \int_E z_1^2 k_1 (t, \D z_1) & \int_E z_2^2 k_2 (t, \D z_2) \\ 
	\int_E z_1^2 k_1 (t, \D z_1) & \int_E z_2^2  k_2 (t, \D z_2) \end{bmatrix}$$
represents the conditional second moments of the mark sizes, where ${\enspace}^{\circ 2}$ denotes the Hadamard power.

The vector of the intensity processes is represented by
\begin{align}
	\bm{\lambda}_t 
	= \bm{\lambda}_0 + \int_0^t \bm{\beta}(\bm{\mu} -  \bm{\lambda}_s) \D s + \int_{(0,t] \times E}  (\bm{\alpha} + \bm{g}(z)) \bm{M}(\D s \times \D z) \label{eq:marklambda}
\end{align}
or using a compensator and the notation in Eq.~\eqref{Eq:K},
\begin{align*}
	\bm{\lambda}_t 
	= \bm{\lambda}_0 &+ \int_0^t \left\{ \bm{\beta}\bm{\mu} + \left( \bm{\alpha} + \mathbb{K}_s[\bm{g}(z)]  - \bm{\beta}) \bm{\lambda}_s \right) \right\}\D s \\
	&+  \int_{(0,t] \times E}  (\bm{\alpha} + \bm{g}(z)) \bm{M}(\D s \times \D z) - 
	\int_0^t  \left( \bm{\alpha}  + \mathbb{K}_s[\bm{g}(z)] \right) \bm{\lambda}_s \D s.
\end{align*}

One of the difficulties in applying the marked model is performing appropriate modeling and quantitative analysis if the distribution of marks is dependent on the underlying processes.
For this dependency structure, we use the following conditional expectations, 
without making any parametric assumptions, to calculate future volatility, using the following definition.

\begin{definition}\label{def:matrix}
	For a $2\times 2$ matrix process $\bm{H}(t)$, 
	\begin{align}
		\overline \ZZ_{\bm{H}} (t) &= \E \left[ \mathbb K_t[\bm{\mathrm{Z}}] \circ \bm{H}_t \right]   \oslash \E[\bm{H}_t],  \label{Eq:ZH}\\
		\overline \ZZ^{(2)}_{\bm{H}} (t) &= \E \left[ \mathbb K_t[\bm{\mathrm{Z}}^{\circ 2}] \circ \bm{H}_t \right]   \oslash \E[\bm{H}_t],  \label{Eq:ZH2}
	\end{align}
	where $\oslash$ denotes the Hadamard division and ${\enspace}^{\circ 2}$ denotes the Hadamard power.	
	These imply covariance adjusted expectations of the mark or squared mark.
\end{definition}

The estimators of Eqs.~\eqref{Eq:ZH}~and~\eqref{Eq:ZH2} can be defined using the sample means of the observed jumps,
as the bars of $\ZZ$ imply.
In addition, if $\bm{H}_t$ and $\bm{K}(t, \D z)$ are independent, then,
\begin{align*}
	\overline \ZZ_{\bm{H}} (t) &= \E \left[ \mathbb K_t[\bm{\mathrm{Z}}] \right]  \\
	\overline \ZZ^{(2)}_{\bm{H}} (t) &= \E \left[ \mathbb K_t[\bm{\mathrm{Z}}^{\circ 2}] \right].
\end{align*}
When 
$$\bm{H}_t = \bm{1}\bm{\lambda}_t^{\top} = \begin{bmatrix}\lambda_1(t) & \lambda_2(t) \\ \lambda_1(t) & \lambda_2(t) \end{bmatrix},$$ 
the subscript of $\overline \ZZ$ is omitted for simplicity.
We assume that $\overline \ZZ (t)$ and $\overline \ZZ_{\bm{\lambda\lambda}^{\top}} (t)$ do not depend on time.
The elements of the matrices can be represented as, for example, 
$$ [\overline \ZZ]_{i,j} = \frac{\E \left[ \lambda_j(t)  \int_E z_j k_j(t, \D z_j)  \right]}{\E \left[ \lambda_j(t) \right]},$$
$$ [\overline \ZZ_{\bm{\lambda}\bm{\lambda}^{\top}}]_{i,j} = \frac{\E \left[ \lambda_i(t)\lambda_j(t)  \int_E z_j k_j(t, \D z_j)  \right]}{\E \left[ \lambda_i(t)\lambda_j(t) \right]}.$$

Whether the mark is dependent or independent of the underlying processes, $\overline \ZZ$s can be computed from the data in a semi-parametric manner using conditional sample means.
For example, for $i=1$, we assume that we observe arrival times and corresponding mark sizes
$$ \{ (\tau_{1,1}, Z_{1,1}), (\tau_{1,2}, Z_{1,2}),  \cdots, (\tau_{1,N}, Z_{1,N}) \}. $$
Then, to estimate $[\overline \ZZ]_{i,j}$, we compute
$$ \frac{ \sum_{n=1}^{N} \hat \lambda_1(\tau_{1,n}) Z_{1,n}}{\sum_{n=1}^{N} \hat \lambda_1(\tau_{1,n})}  $$
where $ \hat \lambda_1$ denotes the fitted intensity.

As described in the previous section, we assume the steady state condition.
We obtain a similar formula for the expected intensities under the marked model in Proposition~\ref{Prop:E_lambda}.

\begin{proposition}~\label{Prop:E_lambda2}

	Under Assumption~\ref{Assum:g} and assuming the steady state condition, we have:

	\begin{equation}
		\E[\bm{\lambda}_t] = ( \bm{\beta} - \bm{\alpha} + \bm{\eta} -  \bm{\eta} \circ \overline \ZZ)^{-1}\bm{\beta}\bm{\mu}.\label{Eq:E_lambda2}
	\end{equation}
\end{proposition}

\begin{proof}
See Appendix~\ref{Sect:proof}.
\end{proof}

After simple but tedious calculations, we have the following lemmas.
The results involve calculating the expectations of stochastic processes represented by random measures. 
These serve as intermediate steps to obtain $\E[\bm{\lambda}_t \bm{\lambda}_t^{\top}]$.
The main purpose of Lemma~\ref{Lemma:EK} is to separate terms related to jumps from the expectation of the stochastic process.
\begin{lemma}~\label{Lemma:EK}
	Under Assumption~\ref{Assum:g}, using the definitions provided in Eqs.~\eqref{Eq:K},~\eqref{Eq:ZH}, and~\eqref{Eq:ZH2}, we obtain
	\begin{align*}
		&\E\left[\mathbb{K}_t [\bm{g}(z)] \bm{X}_t \right] =  \bm\eta \circ (\overline \ZZ_{\bm{1}\bm{X}^{\top}}(t) - 1 )\E[\bm{X}_t]  \\
		&\E\left[\mathbb{K}_t [\bm{g}(z)]  \bm{H}_t \right] = \bm\eta \left((\overline \ZZ_{\bm{H}} (t) - 1)\circ \E[\bm{H}_t]\right) \\
		&\E\left[\bm{H}_t \mathbb{K}_t [\bm{g}(z)]^{\top} \right] = \left((\overline \ZZ_{\bm{H}} (t) - 1)\circ \E[\bm{H}_t]\right) \bm\eta^{\top}
	\end{align*}
 	for a $2\times 1$ vector process $\bm{X}(t)$ and a $2\times 2$ matrix process $\bm{H}(t)$.
\end{lemma}

The following lemma is derived in a similar context, and aims to represent the expectations of stochastic processes, which are expressed as integrals over the random measure $\bm{M}$ by $\E [\bm{\lambda}_t]$
\begin{lemma}~\label{Lemma:useful}
	Under Assumption~\ref{Assum:g} and the steady state assumption, we have
	\begin{align}
		&\E \left[ \int_{(0,t] \times E} \bm{g}(z) \bm{M}(\D s \times \D z) \right] = \bm{\eta} \circ (\overline \ZZ - 1) \E [\bm{\lambda}_t] t, \label{eq:e1}\\
		&\E \left[ \int_{(0,t] \times E} (\bm{\alpha} + \bm{g}(z)) \Dg (\bm{M}(\D s \times \D z)) (\bm{\alpha} + \bm{g}(z))^{\top} \right] = \bm{\mathrm{G}}t, \label{eq:e2}\\
		&\E \left[ \int_{(0,t]\times E}  (\bm{\alpha} + \bm{g}(z)) \Dg(\bm{M}(\D s \times \D z)) \Dg(z) \right] 
		=  \left((\bm{\alpha} - \bm{\eta}) \circ \overline \ZZ + \bm{\eta} \circ \overline \ZZ^{(2)} \right)\Dg(\E[\bm{\lambda}_t]) t,\label{eq:e3}\\
		&\E \left[ \int_{(0,t]\times E} \Dg(z) \Dg(\bm{M}(\D s \times \D z)) \Dg(z) \right] 
		=   \overline \ZZ^{(2)} \circ  \Dg(\E[\bm{\lambda}_t]) t, \label{eq:e4}
	\end{align}
	where
	\begin{equation}
		\bm{\mathrm{G}}= (\bm{\alpha} - \bm{\eta} + \bm{\eta} \circ \overline \ZZ) \Dg(\E[\bm{\lambda}_t]) (\bm{\alpha} - \bm{\eta})^{\top} + 
		(\bm{\alpha} - \bm{\eta})\Dg(\E[\bm{\lambda}_t]) (\bm{\eta} \circ \overline \ZZ)^{\top}
		+\left(\bm{\eta} \circ \overline \ZZ^{(2)\circ\frac{1}{2}}\right) \Dg(\E[\bm{\lambda}_t]) \left(\bm{\eta} \circ \overline \ZZ^{(2)\circ\frac{1}{2}}\right)^{\top}.\label{eq:G}
	\end{equation}
\end{lemma}

\begin{proof}
	See Appendix~\ref{Sect:proof}.
\end{proof}

The following lemma extends Lemma~\ref{Lemma:quad} to the marked model and is used in deriving expressions such as $\E[\bm{\lambda}_t \bm{\lambda}_t^{\top}]$ and $\E[\bm{N}_t \bm{\lambda}_t^{\top}]$.

\begin{lemma}~\label{Lemma:XY}
	Consider vector processes $\bm{X}$ and $\bm{Y}$ such that
	\begin{align*}
		\bm{X}_t &= \bm{X}_0 + \int_0^t \bm{a}_s \D s + \int_{(0,t] \times \mathbb E}\bm{f}_x(s,z) \bm{M}(\D s \times \D z),\\
		\bm{Y}_t &= \bm{Y}_0 + \int_0^t \bm{b}_s \D s + \int_{(0,t] \times \mathbb E}\bm{f}_y(s,z) \bm{M}(\D s \times \D z).
	\end{align*}
	Then,
	\begin{align}
		\frac{\D\E[\bm{X}_t \bm{Y}^{\top}_t]}{\D t} ={}& \E\left[ \bm{X}_{t} (\bm{b}_t^{\top} + \bm{\lambda}_t^{\top} \mathbb{K}_t [\bm{f}_y(t,z)]^{\top}) \right] + \E\left[(\bm{a}_t + \mathbb{K}_t [\bm{f}_x(t,z) ]  \bm{\lambda}_t ) \bm{Y}^{\top}_{t} \right] \nonumber\\
		&+  \frac{\D}{\D t} \E \left[ \int_{(0,t] \times E}  \bm{f}_x(s, z)   \mathrm{Dg} (\bm{M}(\D s \times \D z))  \bm{f}_y(s, z)^{\top} \right] \nonumber \\
		={}& \E[ \bm{X}_{t}\bm{b}_t^{\top} ] + \E\left[ \bm{X}_{t} \bm{\lambda}_t^{\top} \mathbb{K}_t [\bm{f}_y(t,z)]^{\top} \right] 
		+ \E[\bm{a}_t \bm{Y}^{\top}_{t} ]+ \E\left[\mathbb{K}_t [\bm{f}_x(t,z)] \bm{\lambda}_t \bm{Y}^{\top}_{t} \right] \nonumber\\
		&+  \frac{\D}{\D t} \E \left[ \int_{(0,t] \times E}  \bm{f}_x(s, z)   \mathrm{Dg} (\bm{M}(\D s \times \D z))  \bm{f}_y(s, z)^{\top} \right] . \label{eq:diff_xy}
	\end{align}
\end{lemma}

\begin{proof}
	See Appendix~\ref{Sect:proof}.
\end{proof}


Using the above lemmas, we obtain the following formula for the second moments of $\lambda$, similar to Proposition~\ref{Prop:E_ll}.
\begin{proposition}~\label{Prop:E_ll2}
	Under Assumption~\ref{Assum:g} and assuming the steady state condition, $\E[\bm{\lambda}_t \bm{\lambda}_t^{\top}]$ satisfies the following equation:
	\begin{equation}
		\mathcal T \left\{ (\bm{\alpha} - \bm{\beta}) \E[\bm{\lambda}_t \bm{\lambda}_t^{\top} ] +   \bm{\eta} \left( (\overline \ZZ_{\bm{\lambda}\bm{\lambda}^{\top}}^{\top} - 1) \circ \E[ \bm{\lambda}_t \bm{\lambda}_t^{\top} ] \right) + \bm{\beta}\bm{\mu}\E[\bm{\lambda}_t^{\top}] \right\} +  \G = 0 \label{eq:ELL}
	\end{equation}
	where $\G$ is defined by Eq.~\eqref{eq:G}.
\end{proposition}

\begin{proof}
	See Appendix~\ref{Sect:proof}.
\end{proof}

To solve Eq.~\eqref{eq:ELL}, we can use the following fact.
For 2$\times$2 matrices $\M, \N$ and $\X$, the following holds:
$$ \vc ( \M (\N \circ \X) ) = \begin{bmatrix}
	\M \Dg(\N_1)  & \bm{0} \\
	\bm{0} & \M \Dg(\N_2) 
\end{bmatrix}
\vc(\X)
$$
where $\N_i$ is the $i$-th column vector of $\N$ and $\vc$ denotes the vectorization of the matrix by column.
If $\X$ is symmetric, then
$$ 
\vc ((\N \circ \X) \M ) = \begin{bmatrix}
	\M_{11} \N_{11} & \M_{21} \N_{12}  & 0 & 0 \\
	0 & 0 & \M_{11} \N_{21} & \M_{21} \N_{22} \\
	\M_{12} \N_{11} & \M_{22} \N_{12}  & 0 & 0 \\
	0 & 0 & \M_{12} \N_{21} & \M_{22} \N_{22} \\
\end{bmatrix}
\vc(\X)
$$
where $\M_{ij}$ and $\N_{ij}$ are the elements of the $i$-th row and $j$-th column of $\N$ and $\M$, respectively.
In addition, it is known that
$$ \vc (\M \X + \X\N) = (\bm{\mathrm{I}} \otimes \M  + \N^{\top} \otimes \bm{\mathrm{I}}) \vc(\X)$$
where $\otimes$ denotes the Kronecker product.
Hence, Eq.~\eqref{eq:ELL} can be rewritten as
$$ \bm{\mathrm{H}} \, \vc \left(\E[\bm{\lambda}_t \bm{\lambda}_t^{\top}]\right) + \vc\left(\mathcal T \left( \bm{\beta}\bm{\mu} \E[\bm{\lambda}_t^{\top}] \right)  + \G \right) = 0$$
where
\begin{equation*}
	\bm{\mathrm{H}}
	=
	\bm{\mathrm{I}} \otimes (\bm{\alpha} - \bm{\beta})  + (\bm{\alpha} - \bm{\beta}) \otimes \bm{\mathrm{I}} 
	+  
	\begin{bmatrix}
		\bm{\eta} \Dg\left(\overline \ZZ_{\bm{\lambda}\bm{\lambda}^{\top}, 1}^{\top} - 1 \right)  & \bm{0} \\
		\bm{0} & \bm{\eta} \Dg\left(\overline \ZZ_{\bm{\lambda}\bm{\lambda}^{\top}, 2}^{\top} - 1\right)
	\end{bmatrix} 
	+  \begin{bmatrix}
		\eta_{11} \N_{11} & \eta_{12} \N_{12}  & 0 & 0 \\
		0 & 0 & \eta_{11} \N_{21} & \eta_{12} \N_{22} \\
		\eta_{21} \N_{11} & \eta_{22} \N_{12}  & 0 & 0 \\
		0 & 0 & \eta_{21} \N_{21} & \eta_{22} \N_{22} \\
	\end{bmatrix}
\end{equation*}
with $\N =  \ZZ_{\bm{\lambda}\bm{\lambda}^{\top}} - 1$.

As in Proposition~\ref{Prop:E_lN}, the following proposition demonstrates that $\E[\bm{N}_t \bm{\lambda}_t^{\top}]$ can be approximated by a linear function of time.

\begin{proposition}~\label{Prop:E_lN2}
	Under Assumption~\ref{Assum:g}, we use the following approximation for the marked model with constant matrices $\Aa$ and $\BB$:
	$$ \E[\bm{N}_t \bm{\lambda}_t^{\top} ] \approx \Aa t + \BB,$$
	then
	$\Aa$ and $\BB$ satisfy
	\begin{align}
		& \Aa (\bm{\alpha} - \bm{\beta})^{\top}+ (\Aa \circ  (\overline \ZZ_{\bm{N} \bm{\lambda}^{\top}} - 1) )\bm{\eta}^{\top} + \Dg(\overline \ZZ)   \E[\bm{\lambda}_t](\bm{\beta}\bm{\mu})^{\top} = 0 \label{eq:A}\\
		& \BB(\bm{\alpha} - \bm{\beta})^{\top} + 
		(\BB \circ  (\overline \ZZ_{\bm{N} \bm{\lambda}^{\top}} -1) )\bm{\eta}^{\top}+  \overline \ZZ^{\top}_{\bm{\lambda}\bm{\lambda}^{\top}}  \circ \E[\bm{\lambda}_t\bm{\lambda}_t^{\top}] \nonumber\\
		& \hspace{5cm} +\Dg(\E[\bm{\lambda}_t]) \left((\bm{\alpha} - \bm{\eta}) \circ \overline \ZZ  + \bm{\eta} \circ \overline \ZZ^{(2)}\right)^{\top} - \Aa = 0. \label{eq:B}
	\end{align}
	
\end{proposition}

\begin{proof}
	See Appendix~\ref{Sect:proof}.
\end{proof}

Solving Eq.~\eqref{eq:A} is equivalent to solving
\begin{equation*}
	\left(
	\bm{\mathrm{I}} \otimes (\bm{\alpha} - \bm{\beta}) 
	+
	\begin{bmatrix}
		\bm{\eta} \Dg(\overline \ZZ^{\top}_{\bm{N}\bm{\lambda}^{\top}, 1} - 1)  & \bm{0}
		\bm{0} & \bm{\eta} \Dg(\overline \ZZ^{\top}_{\bm{N}\bm{\lambda}^{\top}, 2} - 1)
	\end{bmatrix}
	\right)
	\vc (\Aa^{\top}) \\
	+ \vc(\bm{\beta}\bm{\mu}\E[\bm{\lambda}_t]^{\top} \mathrm{Dg}(\overline \ZZ) ) = 0.
\end{equation*}
After computing $\Aa$, we obtain $\BB$ in a similar manner by solving
\begin{multline*}
	\left(
	\bm{\mathrm{I}} \otimes (\bm{\alpha} - \bm{\beta}) 
	+
	\begin{bmatrix}
		\bm{\eta} \Dg(\overline \ZZ^{\top}_{\bm{N}\bm{\lambda}^{\top}, 1} - 1)  & \bm{0} \\
		\bm{0} & \bm{\eta} \Dg(\overline \ZZ^{\top}_{\bm{N}\bm{\lambda}^{\top}, 2} - 1)
	\end{bmatrix}
	\right)
	\vc (\BB^{\top}) \\
	+ \vc \left( \overline \ZZ_{\bm{\lambda}\bm{\lambda}^{\top}}  \circ \E[\bm{\lambda}_t\bm{\lambda}_t^{\top}]  + \left((\bm{\alpha} - \bm{\eta}) \circ \overline \ZZ  + \bm{\eta} \circ \overline \ZZ^{(2)}\right) \Dg(\E[\bm{\lambda}_t])  - \Aa^{\top} \right) = 0.
\end{multline*}

The variance formula now be calculated using the marked model.
\begin{theorem}\label{Thm:var}
	Under Assumption~\ref{Assum:g}, 
	the variance of the difference between the two counting processes is 
	\begin{align*}
		\Var(N_1(t) - N_2(t)) &= \U^{\top} \E[ \bm{N}_t \bm{N}_t^{\top}] \U - \left( \U^{\top} \Dg(\overline \ZZ) \E[\bm{\lambda}_t] t \right)^2
	\end{align*}
	where
	\begin{align*}
		\E[ \bm{N}_t \bm{N}_t^{\top}] \approx{}&  \mathcal{T} \left\{\overline \ZZ_{\bm{N}\bm{\lambda}^{\top}} \circ 
		\left(\frac{1}{2}\Aa t^2 + \BB t \right) \right\}  +  \overline \ZZ^{(2)} \circ \Dg (\E [\bm{\lambda}_t]) t.
	\end{align*}
\end{theorem}

\begin{proof}
	See Appendix~\ref{Sect:proof}.
\end{proof}

However, this formula has one drawback.
When we compute the $\overline \ZZ$s non-parametrically using conditional sample means, the estimated variance can be negative.
To avoid possible negative values in the non-parametric approach, 
the following restricted version which is similar to Theorem~\ref{Thm:simple_vol}, is preferable.
Note that the variance in Theorem~\ref{Thm:var} is represented as a quadratic function of $t$; however, in the following corollary, the variance is expressed as a linear function of $t$.

\begin{corollary}~\label{Cor:var}
	If $\overline \ZZ = \overline \ZZ_{\bm{N}\bm{\lambda}^{\top}}$,
	\begin{equation}
	\Var(N_1(t) - N_2(t)) = \bm{\mathrm{u}}^{\top} \left[ \mathcal{T} \left\{\overline \ZZ \circ \BB  \right\}  +  \overline \ZZ^{(2)}\circ \Dg (\E [\bm{\lambda}_t]) \right] \bm{\mathrm{u}} t \label{Eq:var_cor}
	\end{equation}
	where $\BB$ satisfies
	\begin{align}
		\BB(\bm{\alpha} - \bm{\beta})^{\top} + 
		(\BB \circ  (\overline \ZZ -1) )\bm{\eta}^{\top}+  \overline \ZZ^{\top}_{\bm{\lambda}\bm{\lambda}^{\top}}  \circ \E[\bm{\lambda}_t\bm{\lambda}_t^{\top}] \nonumber
		+\Dg(\E[\bm{\lambda}_t]) \left((\bm{\alpha} - \bm{\eta}) \circ \overline \ZZ  + \bm{\eta} \circ \overline \ZZ^{(2)}\right)^{\top} - \Aa = 0. 
	\end{align}
	and
	\begin{equation}
		\Aa = \Dg(\overline \ZZ) \E[\bm{\lambda}_t] \E[\bm{\lambda}_t]^{\top}. \label{eq:A2}
	\end{equation}
\end{corollary}

\begin{proof}
	See Appendix~\ref{Sect:proof}.
\end{proof}

When we further assume the independence of mark sizes, we can use the following corollary.
In this case,
$ \overline \ZZ$ and $ \overline \ZZ^{(2)}$ are simply first and second moments of mark size distribution, respectively.

\begin{corollary}~\label{Cor:ind}
	If we assume that the mark size and underlying processes are independent, then we have
	$$  \overline \ZZ = \overline \ZZ_{\bm{N} \bm{\lambda}^{\top}} = \overline \ZZ_{\bm{\lambda} \bm{\lambda}^{\top}}.$$
	Therefore,
	\begin{equation}
	\Var(N_1(t) - N_2(t)) = \bm{\mathrm{u}}^{\top} \left[ \mathcal{T} \left\{\overline \ZZ \circ \BB  \right\}  +  \overline \ZZ^{(2)}\circ \Dg (\E [\bm{\lambda}_t]) \right] \bm{\mathrm{u}} \label{Eq:mark_ind_Var} 
	\end{equation}
	and $\BB$ satisfies
	\begin{equation}
		\BB(\bm{\alpha} - \bm{\beta})^{\top} + 
		(\BB \circ  (\overline \ZZ -1) )\bm{\eta}^{\top}+  \overline \ZZ^{\top}  \circ \E[\bm{\lambda}_t\bm{\lambda}_t^{\top}] 
		+\Dg(\E[\bm{\lambda}_t]) \left((\bm{\alpha} - \bm{\eta}) \circ \overline \ZZ  + \bm{\eta} \circ \overline \ZZ^{(2)}\right)^{\top} - \Aa = 0 \label{Eq:B_cor2}
	\end{equation}
	with $\Aa$ in Eq.~\eqref{eq:A2}.
\end{corollary}

\begin{remark}\label{Remark:simple2}
	The formula for the marked Hawkes variance is simplified when an independent mark is added to the symmetric Hawkes model discussed in Remark~\ref{Remark:simple}, and no additional impact is assumed; that is, $\bm{\eta} = \bm{0}$. 
	Additionally, for symmetry, it is assumed that the distributions of the up and down jump sizes are identical.
	In this case, the mean of the intensity processes satisfies Eq.~\eqref{Eq:E_lambda_simple}.
	Furthermore, $\E[\bm{\lambda}_t \bm{\lambda}_t^{\top}]$ in Proposition~\ref{Prop:E_ll2} is simplified and satisfies Eq.~\eqref{Eq:Ell_reduce}.
	Using Eq.~\eqref{Eq:B_cor2}, we obtain:
	\begin{align*}
		\BB &= (\bm{\alpha} - \bm{\beta} )^{-1} \{\Aa - \overline \ZZ \circ \E[\bm{\lambda}_t \bm{\lambda}_t^{\top} ] - (\bm{\alpha} \circ \overline \ZZ) \mathrm{Dg}(\E[\bm{\lambda}_t]) \} \\
		&= \bar Z \E[\lambda_t] \left( \frac{1}{2 (\beta - \alpha_1 - \alpha_2)^2 (\beta - \alpha_1 + \alpha_2)^2}
		\begin{bmatrix} (\alpha_1 - \beta)^2 + \alpha_2^2 & -2 \alpha_2 (\alpha_1 - \beta) \\ -2 \alpha_2 (\alpha_1 - \beta) & (\alpha_1 - \beta)^2 + \alpha_2^2 \end{bmatrix}
		\begin{bmatrix} \alpha_1^2 + \alpha_2^2 & 2\alpha_1\alpha2 \\
			2 \alpha_1 \alpha_2 & \alpha_1^2 + \alpha_2^2 \end{bmatrix}  \right. \\
		& \left. \phantom{\E[\lambda_t]abc} - \frac{1}{(\beta - \alpha_1 - \alpha_2) (\beta - \alpha_1 + \alpha_2)} 
		\begin{bmatrix} \alpha_1 - \beta & -\alpha_2 \\ -\alpha_2 & \alpha_1 - \beta \end{bmatrix}
		\begin{bmatrix} \alpha_1 & \alpha_2 \\ \alpha_2 & \alpha_1 \end{bmatrix} \right)
	\end{align*}
	where $\bar Z$ is the mean of the mark distribution and $\Aa$ equals Eq.~\eqref{eq:A2}.
	Thus,
	\[
	\U^{\top} (\overline \ZZ \circ \BB) \U  = \bar Z^2 \E[\lambda_t]  \left( \frac{(\alpha_1 - \alpha_2)^2}{(\beta - \alpha_1 + \alpha_2)^2} + \frac{2(\alpha_1 - \alpha_2)}{\beta - \alpha_1 + \alpha_2} \right),
	\]
	and
	the variance formula in Eq.~\eqref{Eq:mark_ind_Var} is reduced to
	\begin{align*}
		\Var(N_1(t) - N_2(t)) &= \bm{\mathrm{u}}^{\top} \left[ \mathcal{T} \left\{\overline \ZZ \circ \BB  \right\}  +  \overline \ZZ^{(2)}\circ \Dg (\E [\bm{\lambda}_t]) \right] \bm{\mathrm{u}} t \\
		&= \frac{2\mu\beta \bar Z^2 t}{\beta - \alpha_1 - \alpha_2}\left(\frac{(\alpha_1 - \alpha_2)^2}{(\beta - \alpha_1 + \alpha_2)^2} + \frac{2(\alpha_1 - \alpha_2)}{\beta - \alpha_1 + \alpha_2}  + \frac{\bar Z^{(2)}}{\bar Z^2}\right)
	\end{align*}
	where $\bar Z^{(2)}$ denotes the second moment of the mark distribution.
\end{remark}

\begin{remark}
	One may use a parametric assumption for the mark distribution.
	Then, the estimations of $\overline \ZZ$ and $\overline \ZZ^{(2)}$ depend on the assumed distribution.
	For example, we can assume that the jump size follows a geometric distribution.
	In this case, under the independent assumption, $\overline \ZZ$ can be estimated using the sample means of the jump sizes and
	$$ \bar \ZZ^{(2)} = 2  \bar \ZZ^{\circ 2} - \bar \ZZ .$$
\end{remark}

\section{Empirical study}~\label{Sec:application}

In this section, the proposed marked Hawkes model is estimated using mid-price processes derived from National Best Bid and Offer (NBBO) data for the US stock market 
and Emini S\&P 500 futures prices.
We applied the Hawkes model to numerous symbols, such as AALP, ADBE, AMZN, NVDA, FB, MCD, NKE, XOM, and others. 
Daily and intraday volatilities are computed using the estimates and formulas derived in the previous section.

\subsection{Filtering and robust test}\label{Subsec:filter}

Filtered data were used instead of raw data
for the Hawkes volatility performance test.
The raw data record all stock price movements and contain information on numerous price changes, even within very short time intervals.
For example, the data of AMZN on June 20, 2018, had an average of 3.23 events per second, and in the most extreme case, 60 events occurred per second.
If the observation interval is 0.1 seconds, then the maximum number of events is 36. 
Even within extremely short time intervals, there may be an innumerable number of records.
The large capacity of high-frequency data is a great advantage because it provides a large amount of information; however, there are also disadvantages.
This is because many records only provide information on very short-term changes that do not significantly impact long-term volatility estimates.

Therefore,  caution is required in the direct use of raw data for estimation.
For example, the estimates of $\alpha$ or $\beta$ in the Hawkes model would be very large compared with $\mu$ because of the ultra-high-frequency activities that might be from automated machines.
In addition, raw price dynamics data may not fit well with the single-kernel exponential Hawkes model.
Furthermore, additional kernels may be required, which complicate the model and quantitative analysis.
It would be desirable to use raw data directly to study the nature of ultra-high-frequency movements of stock price changes.
However, as we focus on daily volatility, 
it is useful to utilize the filtered data described below.
This filtering simplifies the data but does not affect the measurement of daily price variations.

First, we extract the mid-price dynamics, 
which is the average of the bid and ask prices from the NBBO.
We observe the mid-price at fixed intervals of $\Delta t = 0.1$ seconds, 
specifically at each time point $\Delta t, 2\Delta t, 3\Delta t, \ldots$.
If this differs from the previously observed mid-price 0.1 seconds before, 
we then find the time of the last change to that price and record the time and changed price.
If the price is the same as the previously observed price, 
move to the next time point, 0.1 seconds later, and repeat the above procedure.
Through filtering, unnecessary movements observed at ultra-high-frequency,
which is generally known as microstructure noise, can be removed.

Of course, there is no clear standard as to how long the filtering time interval should be, but through numerous empirical studies, it has been verified that 0.1 second filtering produces stable results.
Further research may be needed to develop a more accurate filtering method.
In Figure~\ref{Fig:mid}, we compare the mid-price movement in the raw data (top) with the filtered mid-price movement (bottom), which is smoother.

\begin{figure}
	\centering
	\includegraphics[width=0.65\textwidth]{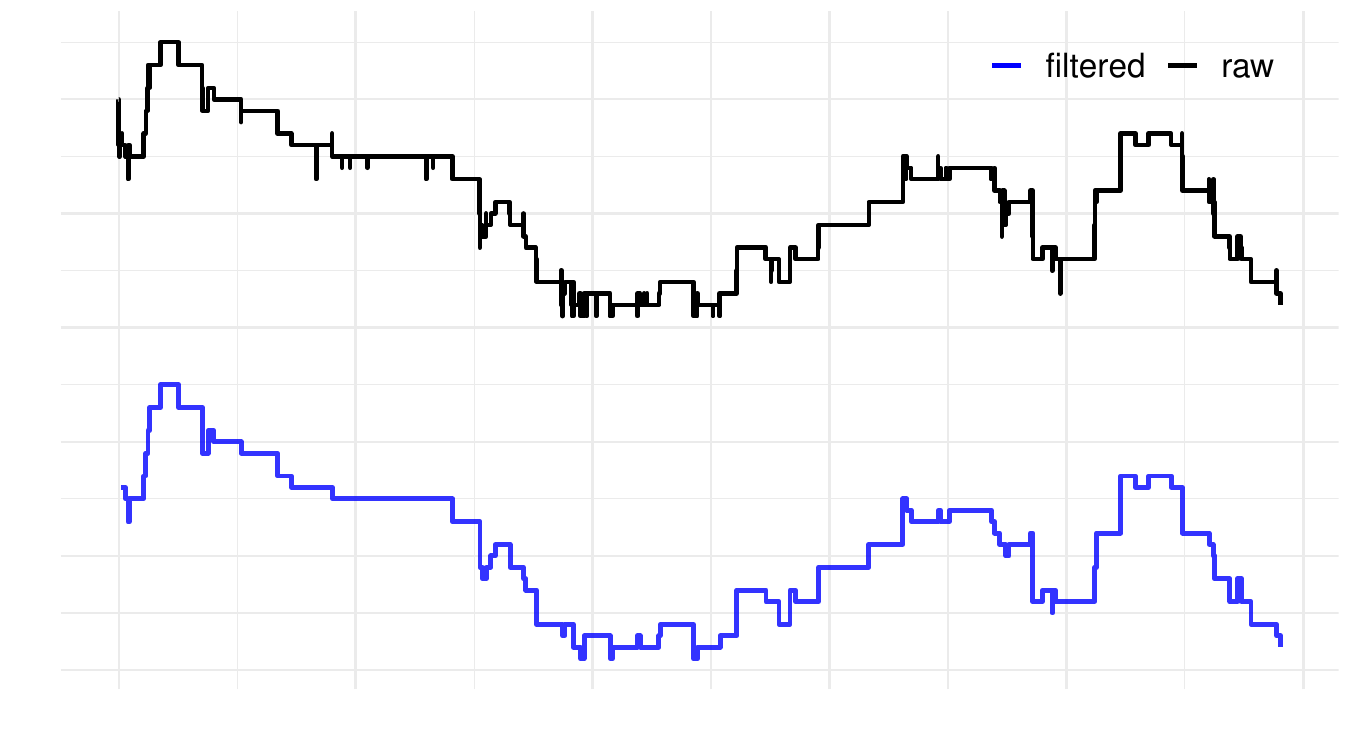}
	\caption{An example of a raw mid price and a filtered version with 0.1 seconds time interval}\label{Fig:mid}
\end{figure}

We conduct a robustness test to assess whether the established marked Hawkes model and filtering methods yield reliable computational outcomes, even when operating under the condition that the model’s assumptions differ from the actual distribution of the tick price process.
Alternatively, we employ the Hawkes process with power-law kernels, 
which have also been used to model high-frequency financial data \citep{Bacry2016,Zhang2016}.
In the two-dimensional power-law model, the kernel for intensity is represented by
\[ \bm{h}(t) = 
\bm{\alpha}  \circ
\begin{bmatrix}
	(c + t)^{-\beta_{11}} & (c + t)^{-\beta_{12} t} \\
	(c + t)^{-\beta_{21} t} & (c + t)^{-\beta_{22} t}
\end{bmatrix}	
\]
with the following presumed parameters:
\[
\bm{\mu} = 
\begin{bmatrix}
	0.41 \\
	0.40 
\end{bmatrix},
\quad
\bm{\alpha} = 
\begin{bmatrix}
230 & 320 \\
300 & 230
\end{bmatrix},
\quad
\begin{bmatrix}
	\beta_{11} & \beta_{12} \\
	\beta_{21} & \beta_{22}
\end{bmatrix}
=
\begin{bmatrix}
	14.7 & 16.5 \\
	15.5 & 15.4
\end{bmatrix},
\quad
c = 1.2.
\]

To assess the robustness of the volatility formula, 
we generate 10,000 paths based on the power-law kernel model, each with a one-hour time horizon. 
After generation, the samples are processed using the previously described filtering method.
We then apply the exponential marked Hawkes model to estimate these parameters.
From the formula specified in Theorem~\ref{Thm:var}, 
we calculate the volatility estimates from the square roots of Eq.~\eqref{Eq:var_cor} for each path to obtain the sampling distribution.
In addition, we compute the difference $N_1(t) - N_2(t)$ directly from each path and utilize these 10,000 values to determine the sample standard deviation of $N_1(t) - N_2(t)$, which serves as the basis for comparison.

In this simulation study, 
the sample standard deviation of $N_1(t) - N_2(t)$ is 167.35, 
and the sample mean of the Hawkes volatilities is 166.53. 
The sampling distribution of the Hawkes volatility is illustrated  in Figure~\ref{Fig:dist}. 
In this case, the Hawkes volatility appears to function nearly as an unbiased estimator.
Although this experiment does not encompass all the potential model assumptions for the tick structure,
it shows that our method maintains a degree of robustness even under model assumptions that do not rely on an exponential kernel or symmetry.

\begin{figure}
	\centering
	\includegraphics[width=0.45\textwidth]{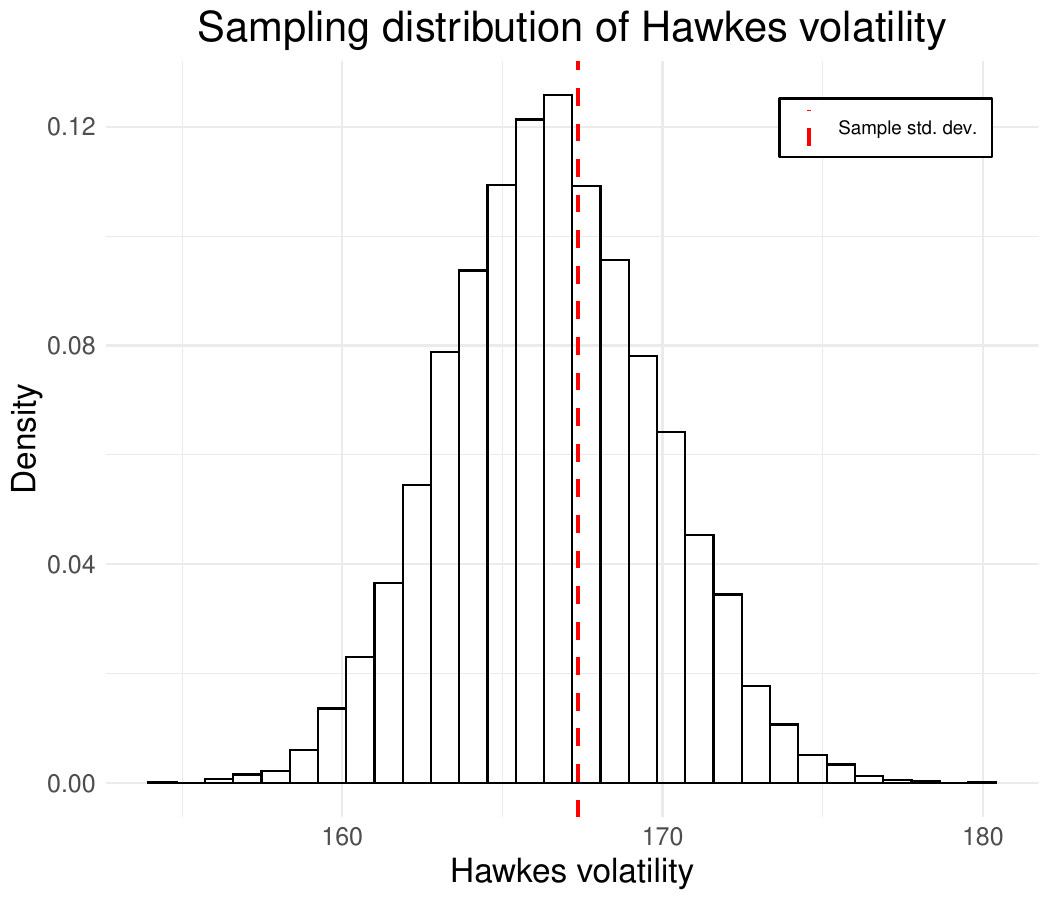}
	\caption{Sampling distribution of Hawkes volatility and sample standard deviation}\label{Fig:dist}
\end{figure}

\subsection{Basic result}

The estimates were computed using the maximum likelihood method, based on the filtered data described in the previous subsection.
The dynamics of the daily estimates of the filtered data from 2016 to 2019 are shown in Figures~\ref{Fig:mu}-\ref{Fig:eta}.
Examples of these estimates are presented in \ref{Sect:estimate}.
Figure~\ref{Fig:mu} shows the dynamics of $\mu_1$ and $\mu_2$ which are the daily estimates from NVDA on the left.
The overall trends of the two estimates are similar, and on some days they are almost the same.
However, there are also many days when the difference between $\mu_1$ and $\mu_2$ is significant, as shown in the right panel of the figure,
where $\mu_1 - \mu_2$ is plotted.
Note that our formula for variances in the previous section does not exclude the case in which $\mu_1 \neq \mu_2$.

\begin{figure}
	\centering
	\includegraphics[width=0.45\textwidth]{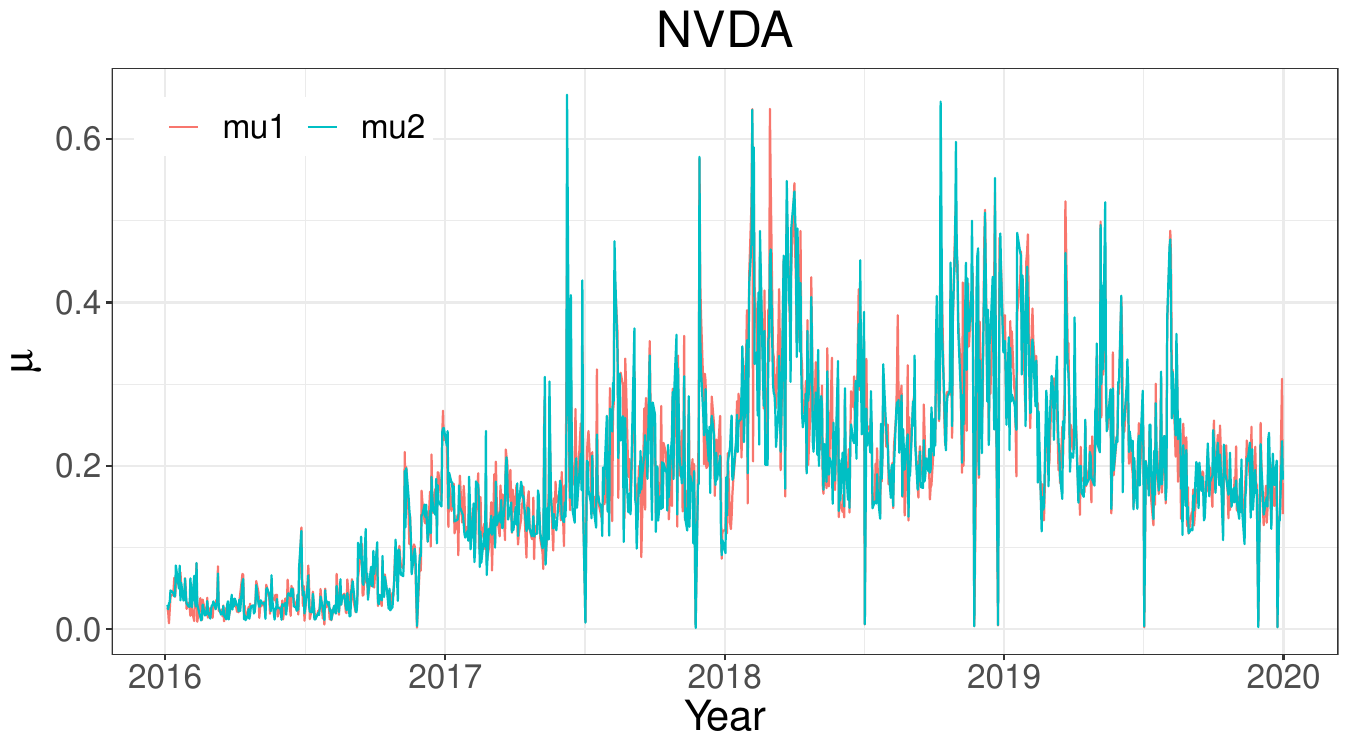} \qquad
	\includegraphics[width=0.45\textwidth]{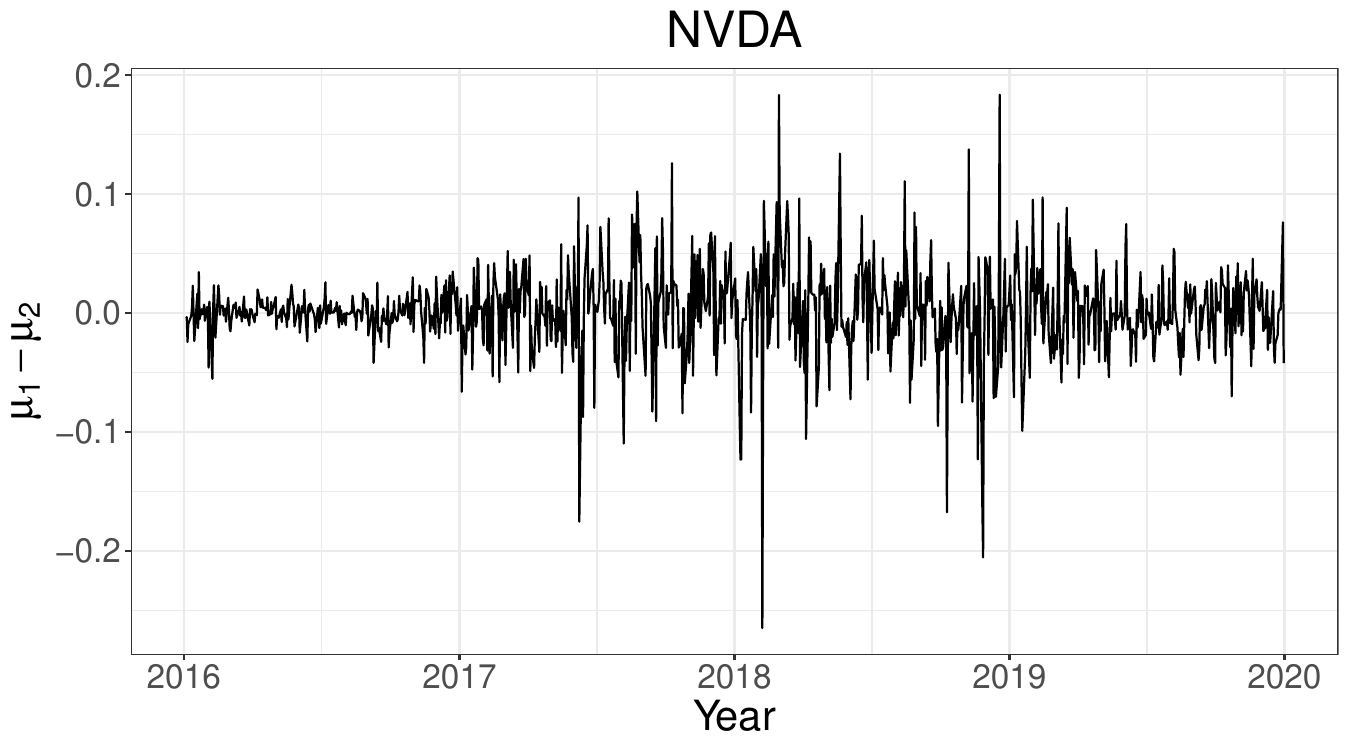} 
	\caption{The dynamics of daily estimated $\mu$ for NVDA from 2016 to 2019 }\label{Fig:mu}
\end{figure}

There are four parameters in the $\bm \alpha$ matrix and the dynamics of the daily estimates are plotted in Figure~\ref{Fig:alpha}. 
On the left, we plot the dynamics of $\alpha_{1,1}$ and $\alpha_{2,2}$, which are also known as the self-excitement parameters.
On the right, we plot the dynamics of $\alpha_{1,2}$ and $\alpha_{2,1}$, which are also known as the mutual excitement or cross-excitement parameters.
In some studies, $\alpha_{1,1} = \alpha_{2,2}$ and $\alpha_{1,2} = \alpha_{2,1}$, that is the symmetric kernel matrix, is assumed to model parsimony.
To some extent, these assumptions are reasonable because of the similar trends in each group.
However, as in the case of $\mu$, we assume a general version of the excitement matrix $\bm \alpha$.
In other words, there are no equality constraints for the elements of the matrix.

\begin{figure}
	\centering
	\includegraphics[width=0.45\textwidth]{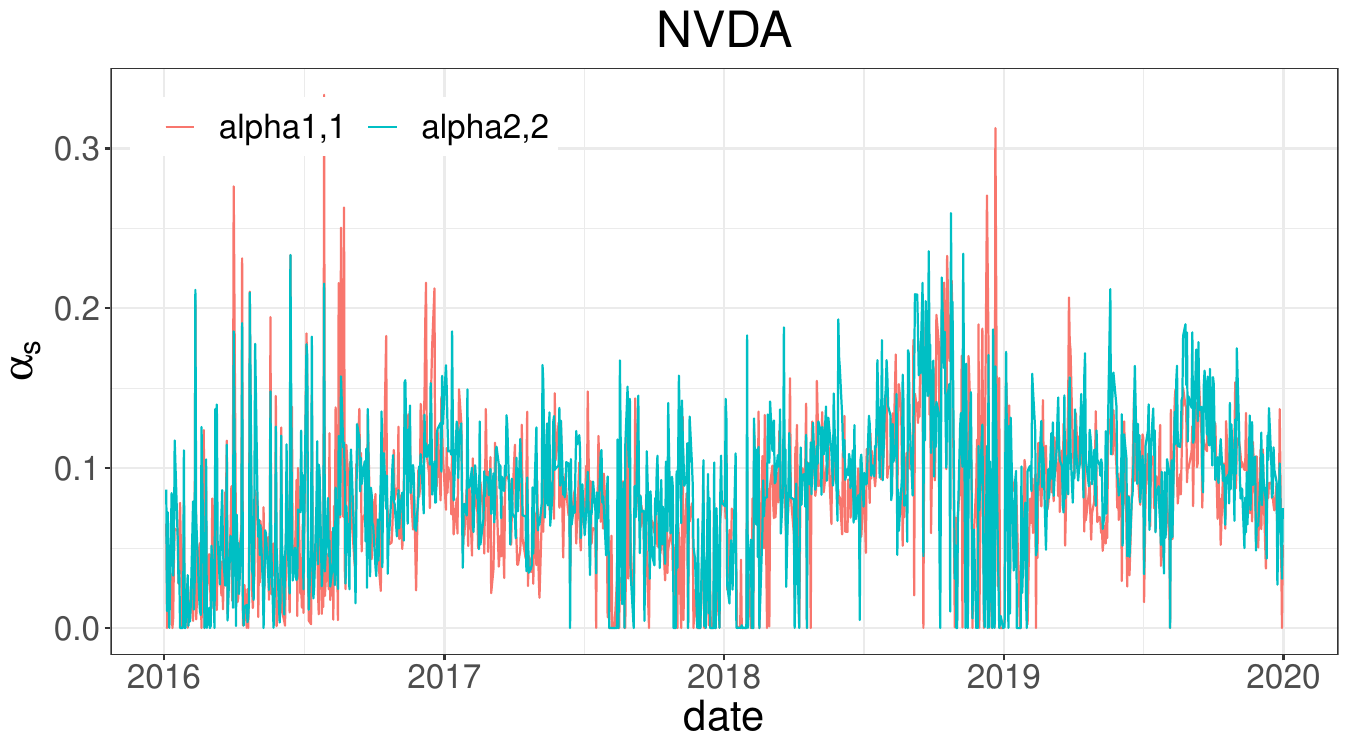} \qquad
	\includegraphics[width=0.45\textwidth]{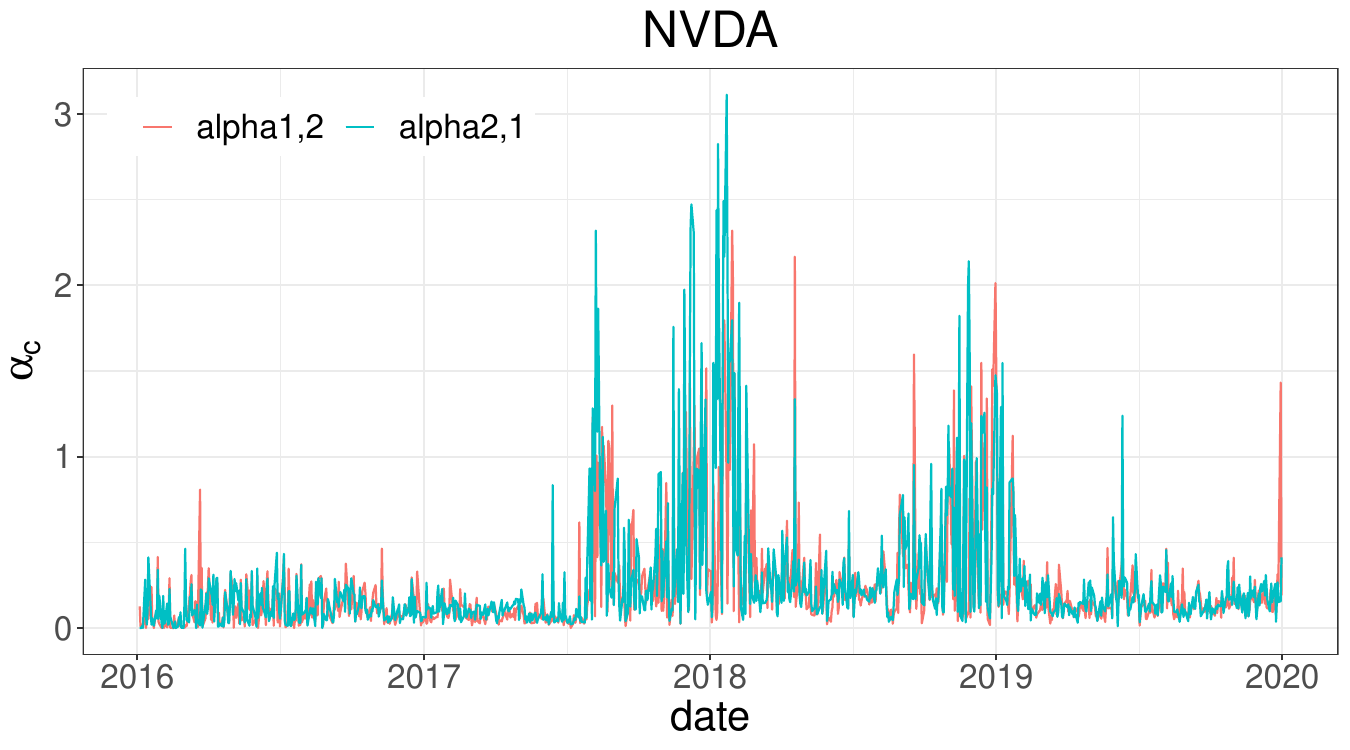} 
	\caption{The dynamics of daily estimated $\alpha$, NVDA from 2016 to 2019}\label{Fig:alpha}
\end{figure}

Figure~\ref{Fig:beta} shows the dynamics of $\beta_1$ and $\beta_2$ on the left side.
Similarly, for $\mu$, $\beta_1$ is similar to $\beta_2$.
However, there are often days on which $\beta_1$ and $\beta_2$ differ significantly.
Therefore, we cannot rule out the possibility that $\beta_1$ and $\beta_2$ differ from each other.

\begin{figure}
	\centering
	\includegraphics[width=0.45\textwidth]{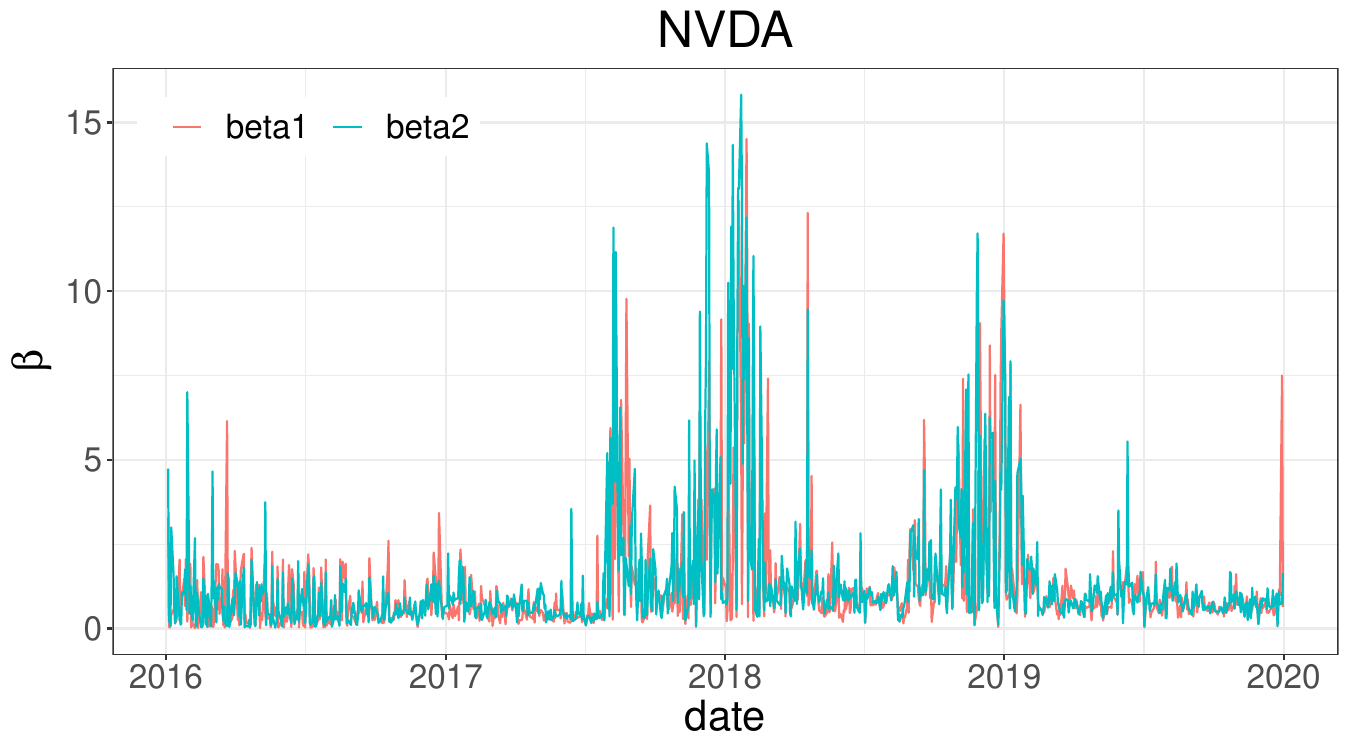} \qquad
	\includegraphics[width=0.45\textwidth]{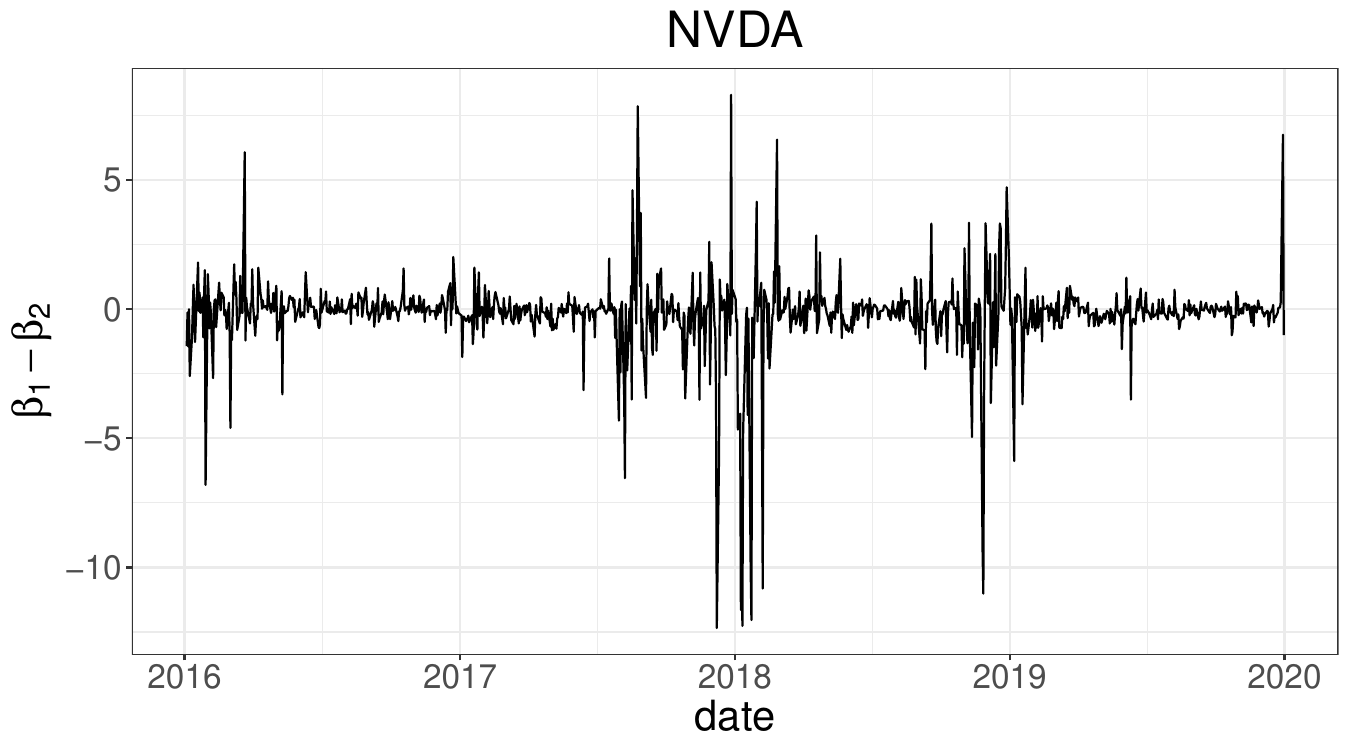} 
	\caption{The dynamics of daily estimated $\beta$ for NVDA from 2016 to 2019}\label{Fig:beta}
\end{figure}

Finally, Figure~\ref{Fig:eta} illustrates the dynamics of the four parameters in $\bm \eta$.
As in the case of $\bm \alpha$, $\eta_{1,1}$ is similar to $\eta_{2,2}$ and $\eta_{1,2}$ is similar to $\eta_{2,1}$.
For simplicity, the assumption of the symmetric $\bm \eta$ matrix is reasonable,
especially when a small amount of data results in a low statistical power.
As the data has sufficient observations, even for the filtered data, 
we do not assert any parameter constraints, and 
estimate the volatility using a general model.
For the intraday estimations to be discussed later, 
the amount of data available for estimations is relatively small, 
and symmetric matrices are assumed for $\bm\alpha$ and $\bm\eta$.

\begin{figure}
	\centering
	\includegraphics[width=0.45\textwidth]{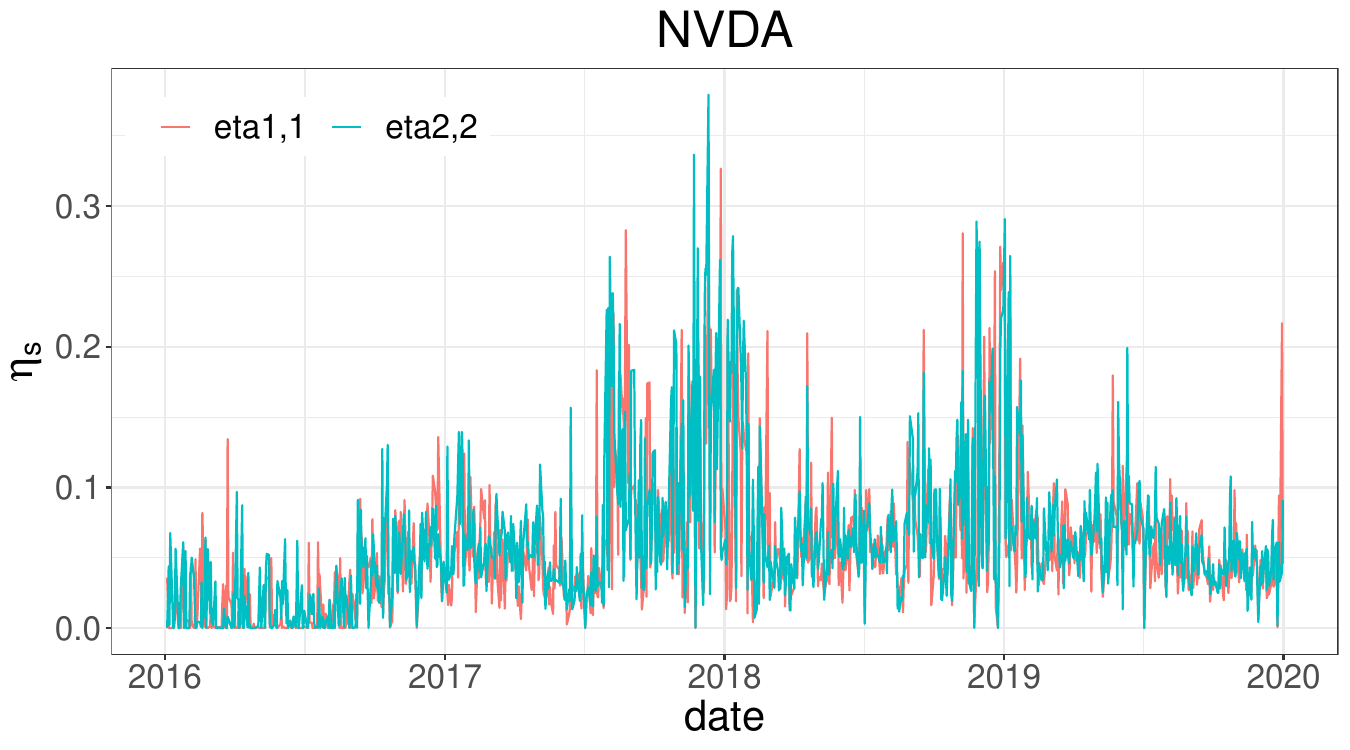} \qquad
	\includegraphics[width=0.45\textwidth]{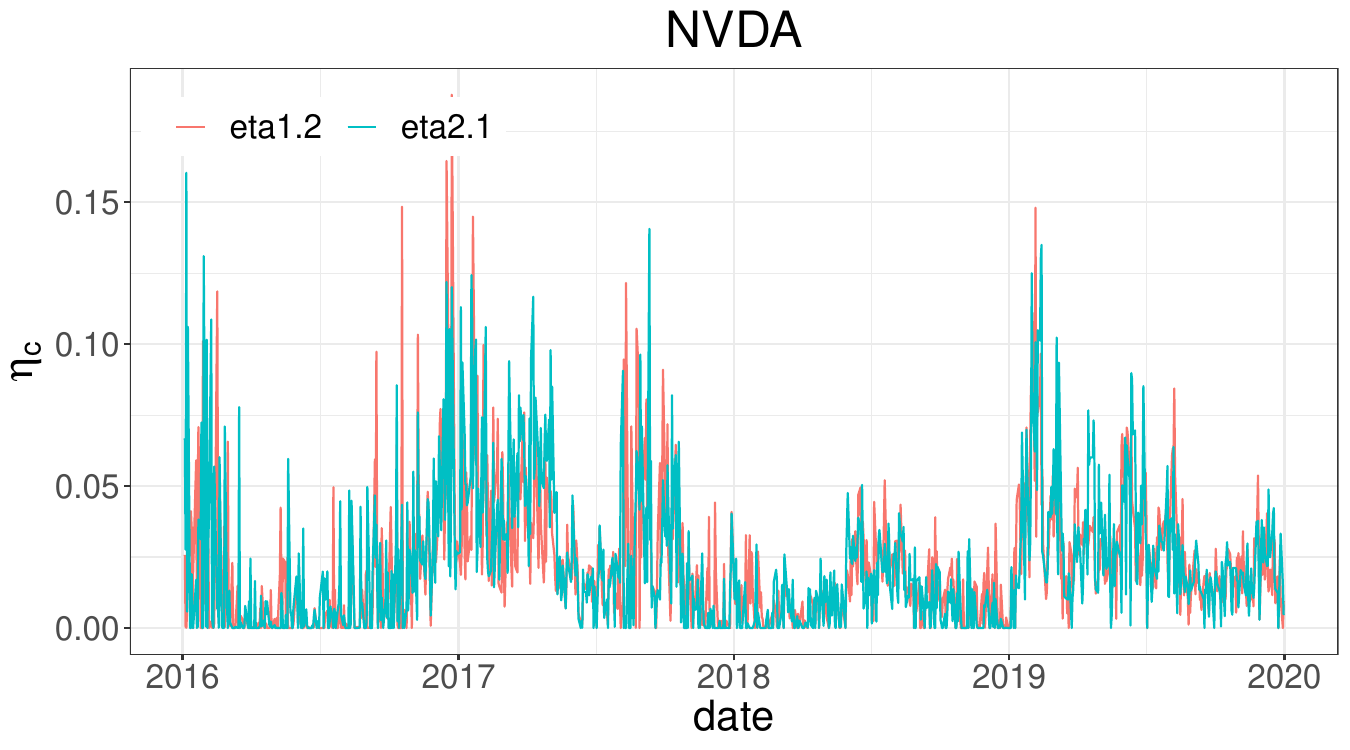} 
	\caption{The dynamics of daily estimated $\eta$, NVDA from 2016 to 2019}\label{Fig:eta}
\end{figure}

We present quantile-quantile (Q-Q) plots of the residuals
to check whether the model fits the filtered data well.
Since the estimations are performed on a daily basis, one Q-Q plot can be generated for each day and stock.
Examples of the Q-Q plots, the residuals versus the unit exponential distribution, are presented in Figure~\ref{Fig:qq}.
Using the estimates obtained from the estimation procedure, 
the fitted intensities, $\hat \lambda_1$ and $\hat \lambda_2$, are computed.
Then, the residuals are computed by
$$\bigcup_{i=1,2} \left\{ \int_{t_{i, j}}^{t_{i, j+1}} \hat \lambda_i (u) \D u \right\},$$
where $t_{i,j}$ denotes the corresponding event times.

\begin{figure}
	\centering
	\includegraphics[width=0.4\textwidth]{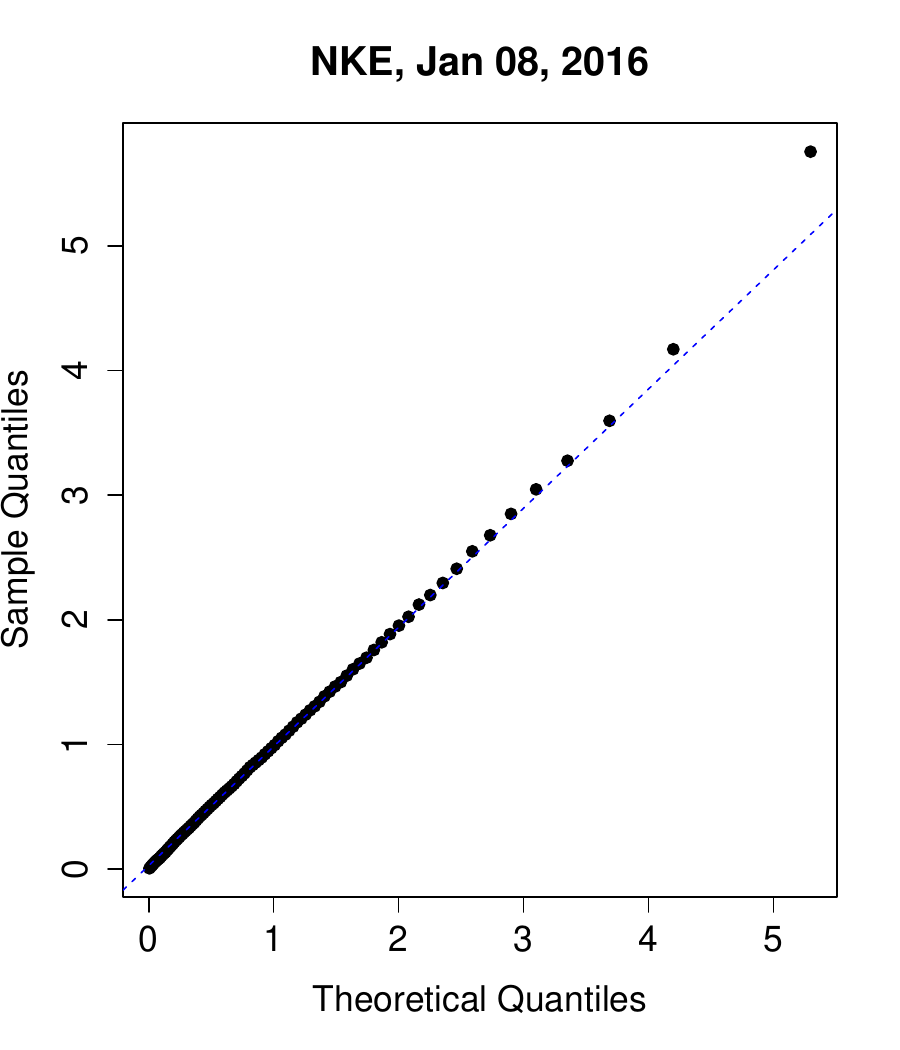} \qquad
	\includegraphics[width=0.4\textwidth]{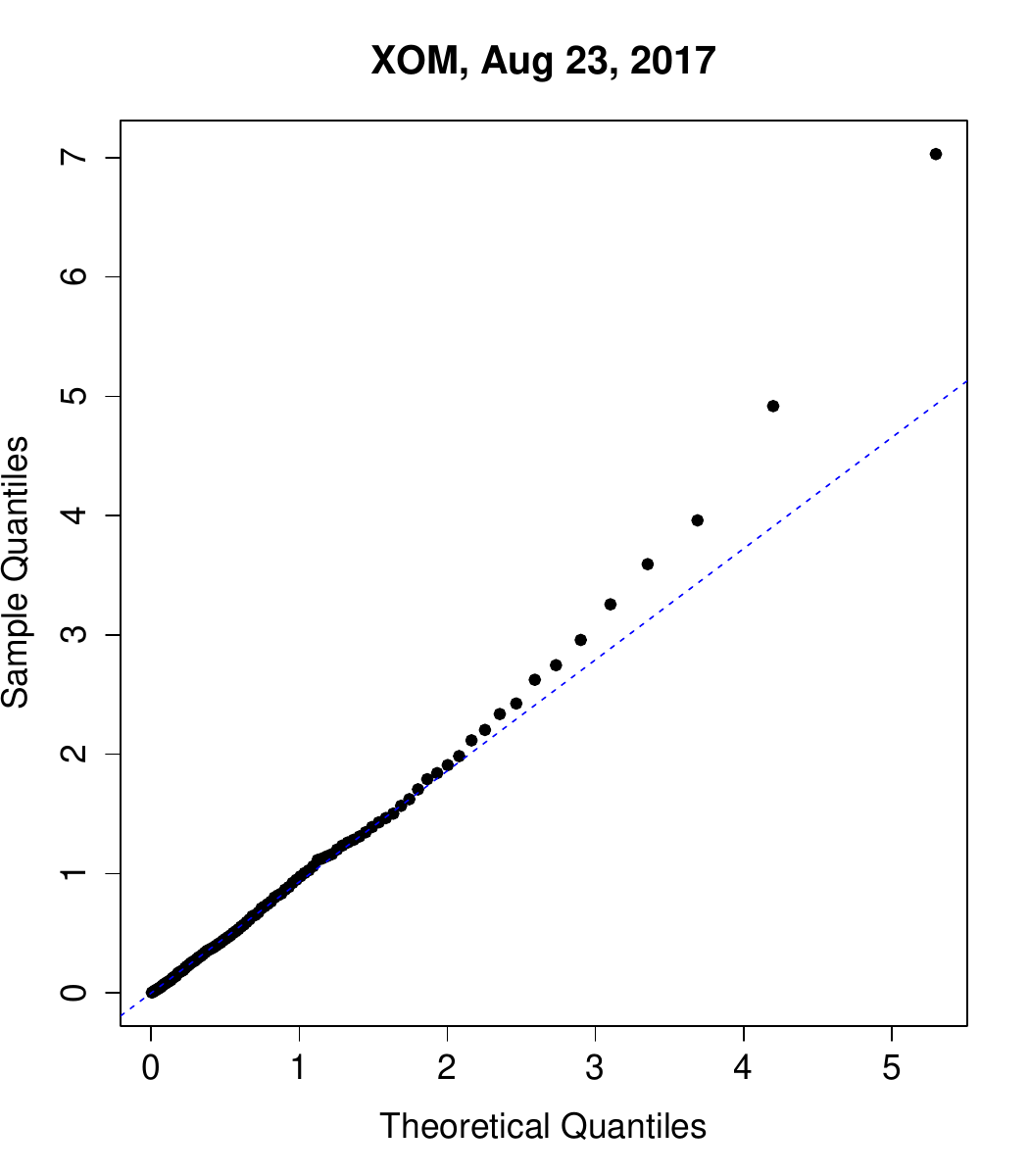} 
	\caption{Q-Q plots of residuals}\label{Fig:qq}
\end{figure}

It is not possible to show the Q-Q plots for all the estimated data.
However, the following two types were observed.
First, the points in the Q-Q plot are very close to the straight line, indicating that the model fits the data well.
In this case, we cannot reject the null hypothesis that the residual distribution follows the unit exponential distribution.
Second, the points in the Q-Q plot are spread slightly above the straight line, and the estimated residuals have slightly thicker tails than the exponential distribution.
In this case, the null hypothesis that the residual follows the unit exponential distribution is rejected at a typical significance level.
However, we believe that this does not significantly affect the estimation of volatility.

\subsection{Volatility dynamics}

The daily volatilities are measured using estimates and formulas in Section~\ref{Sec:mark}.
We present graphs to visually check whether the Hawkes volatility adequately measures stock price variability.
Figure~\ref{fig:abs} shows the absolute daily price changes in NVDA stock from 2016 to 2019.
Figure~\ref{fig:ind} shows the daily Hawkes volatilities computed under the assumption that the marks and intensities are independent, as in Corollary~\ref{Cor:ind}.
In Figure~\ref{fig:dep}, the daily Hawkes volatilities are calculated without excluding the possibility that the marks and intensities are dependent, as in Corollary~\ref{Cor:var}.
The overall trends in the two estimated volatilities are almost identical.
This is a good way to assume independence between the marks and the underlying counting processes and intensities to simplify the procedure.
The realized volatility, calculated following \cite{ABDL} as a benchmark, is plotted in Figure~\ref{fig:real}.
The three types of volatilities show similar dynamics.
The time interval for computing the realized volatility was set to 5 minutes.

\begin{figure}[]
	\begin{subfigure}{0.47\textwidth}
		\centering
		\includegraphics[width=\textwidth]{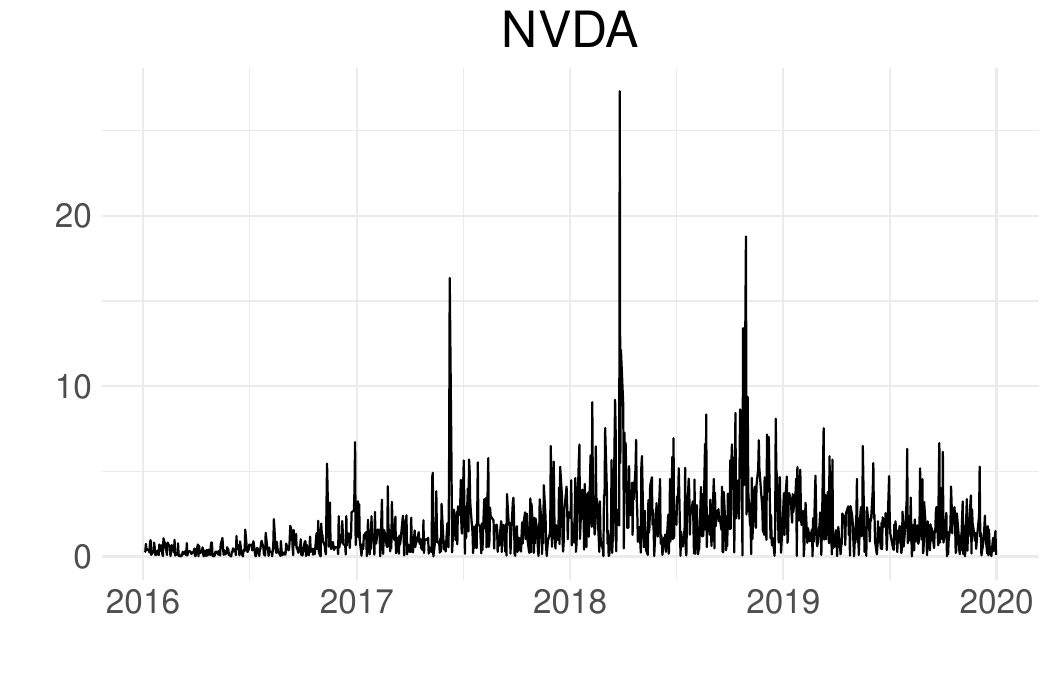}
		\caption{Absolute daily price change}
		\label{fig:abs}
	\end{subfigure}
	\quad
	\begin{subfigure}{0.47\textwidth}
		\centering
		\includegraphics[width=\textwidth]{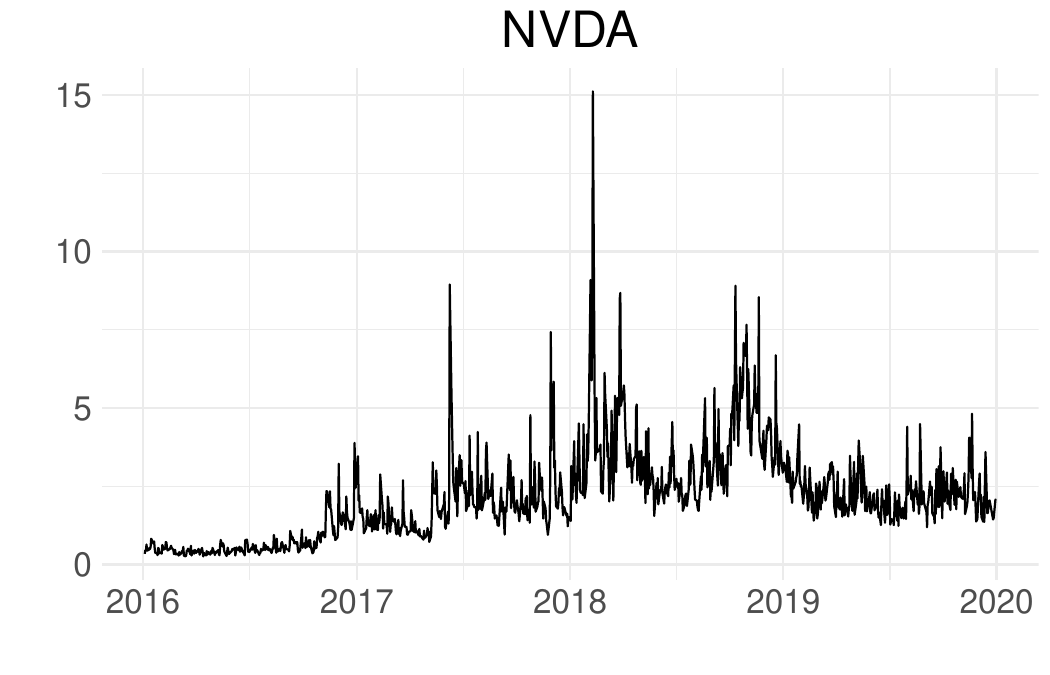}
		\caption{Daily realized volatility}
		\label{fig:real}
	\end{subfigure}
	
	\begin{subfigure}{0.47\textwidth}
		\centering
		\includegraphics[width=\textwidth]{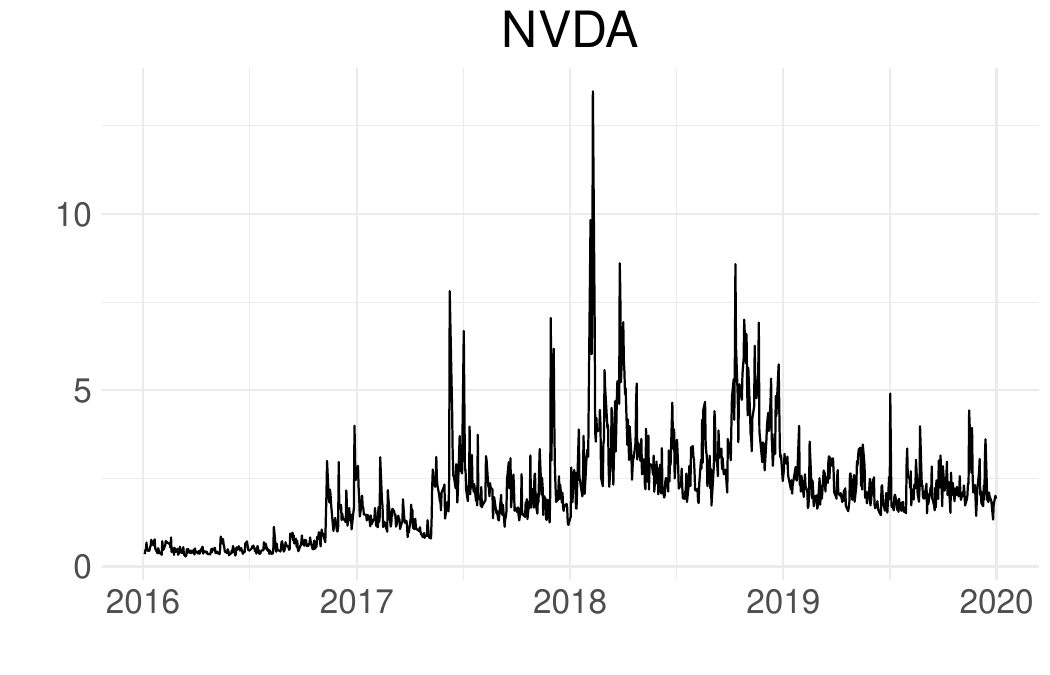}
		\caption{Daily Hawkes volatility under independence}
		\label{fig:ind}
	\end{subfigure}
	\quad
	\begin{subfigure}{0.47\textwidth}
		\centering
		\includegraphics[width=\textwidth]{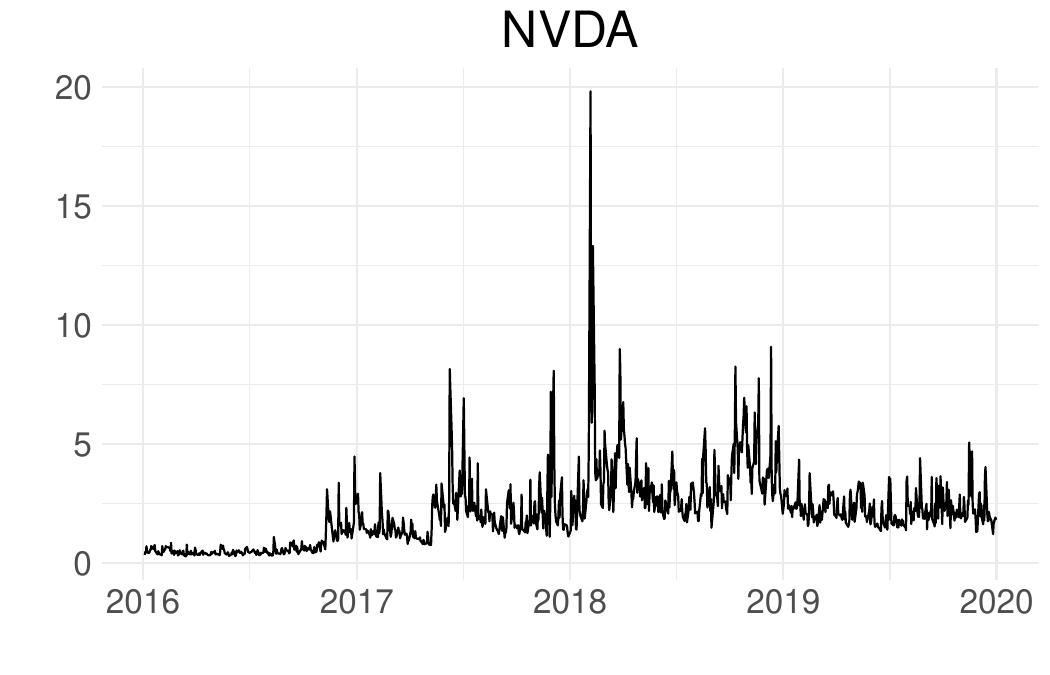}
		\caption{Daily Hawkes volatility under dependence}
		\label{fig:dep}
	\end{subfigure}
	\caption{Comparison of the dynamics of the daily estimated volatilities and the absolute price change}~\label{Fig:compare}
\end{figure}

The backtesting method, which is often used in VaR analysis, is applied to test the performance of the Hawkes volatility.
This provides a graphical guideline for the performance of the volatility measures.
In Figure~\ref{Fig:exceed}, the dynamics of absolute daily price changes from 2016 to 2017 are plotted as a black line.
Twice the daily estimated Hawkes volatility, with the dependence assumption computed by Corollary~\ref{Cor:var}, is plotted as a red line.
Both lines exhibit the same trend.
In 2016, both the red and black lines have small values;
in 2017, both lines exhibit relatively large values.
From a graphical perspective,
the Hawkes volatility seems to adequately capture the variability in price changes.

The days on which the absolute daily price difference exceeds twice the Hawkes standard deviation are represented by blue circles.
Although omitted from the figure, the absolute price difference exceeds twice the Hawkes volatility on 50 days, which is around 5\% of all observed days.
Interestingly, this is consistent with the rule that approximately 95\% of the data are within two standard deviations of the normal distribution.
Furthermore, the distribution of exceedances is spread uniformly over the entire period.
This implies that the Hawkes model captures the volatility trend well.

\begin{figure}
	\centering
	\includegraphics[width=0.8\textwidth]{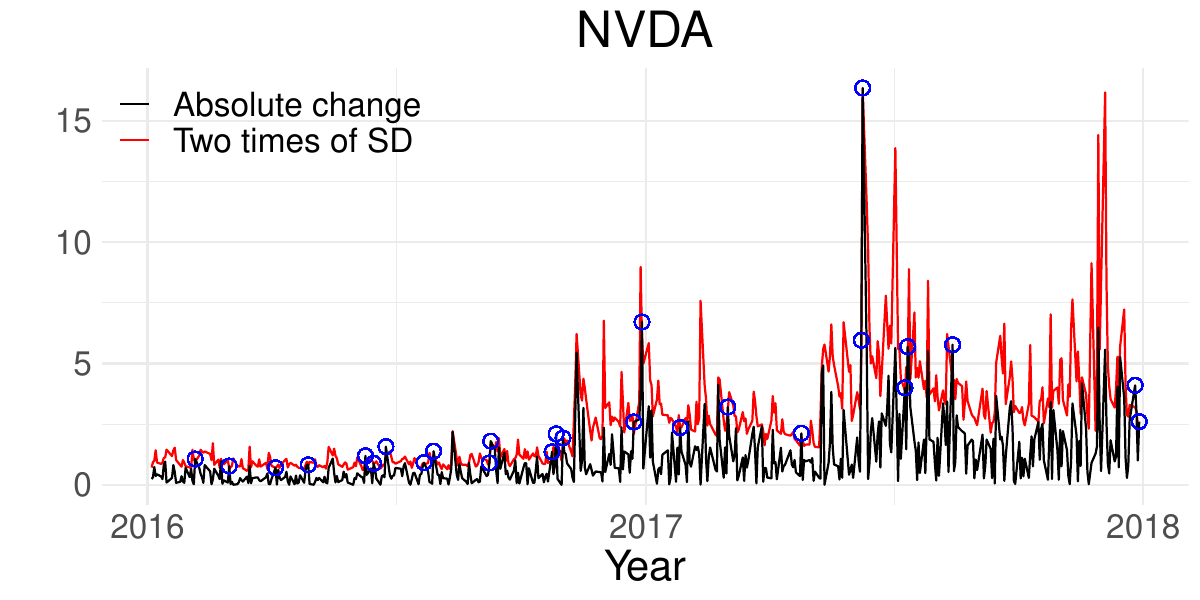} 
	\caption{The exceeded absolute price change compared with two times of measured Hawkes volatility}\label{Fig:exceed}
\end{figure}

We extend this discussion to arbitrary multiples of the standard deviation.
Note that 68\%, 95\%, and 99.7\% of the values lie within one, two, and three standard deviations of the mean zero normal distribution, respectively.
In Figure~\ref{Fig:percent}, the x-axis denotes the number of multiples of the standard deviation.
The y-axis represents the proportion of price changes that do not exceed the corresponding multiples of the standard deviation.
The blue line represents the Hawkes volatility.
The red line represents the probability that the data lie within a multiple of the standard deviation of the standard normal distribution.
Interestingly, the two curves are almost identical.
The Hawkes volatility is based on a model that doese not exclude the possibility of dependence.
However, the volatility under the assumption of independence shows similar results.
These results support the Hawkes volatility as a reliable measure of price variability.

\begin{figure}[]
	\begin{subfigure}{0.47\textwidth}
		\centering
		\includegraphics[width=\textwidth]{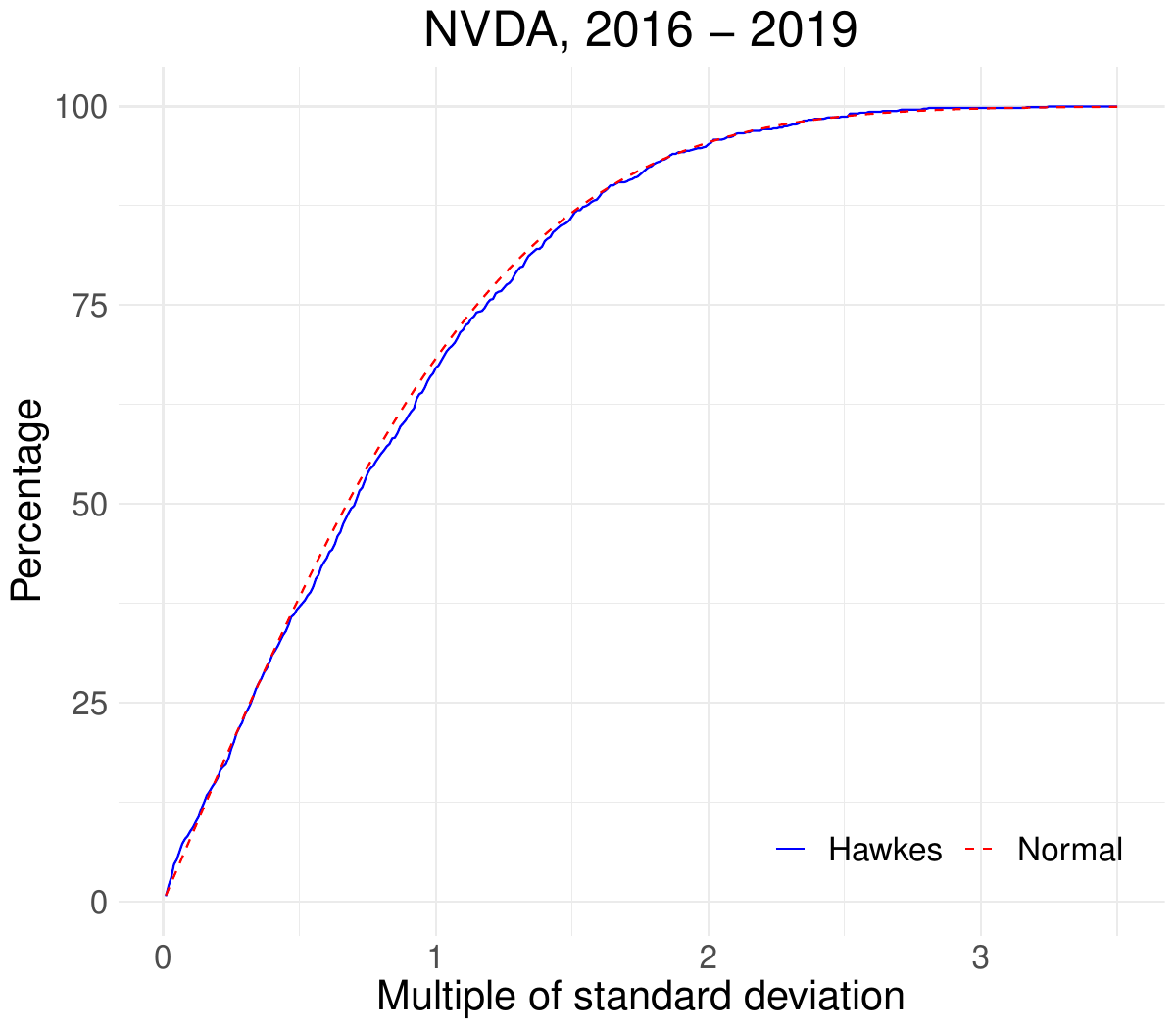}
		\caption{}
		\label{fig:percent_NVDA}
	\end{subfigure}
	\quad
	\begin{subfigure}{0.47\textwidth}
		\centering
		\includegraphics[width=\textwidth]{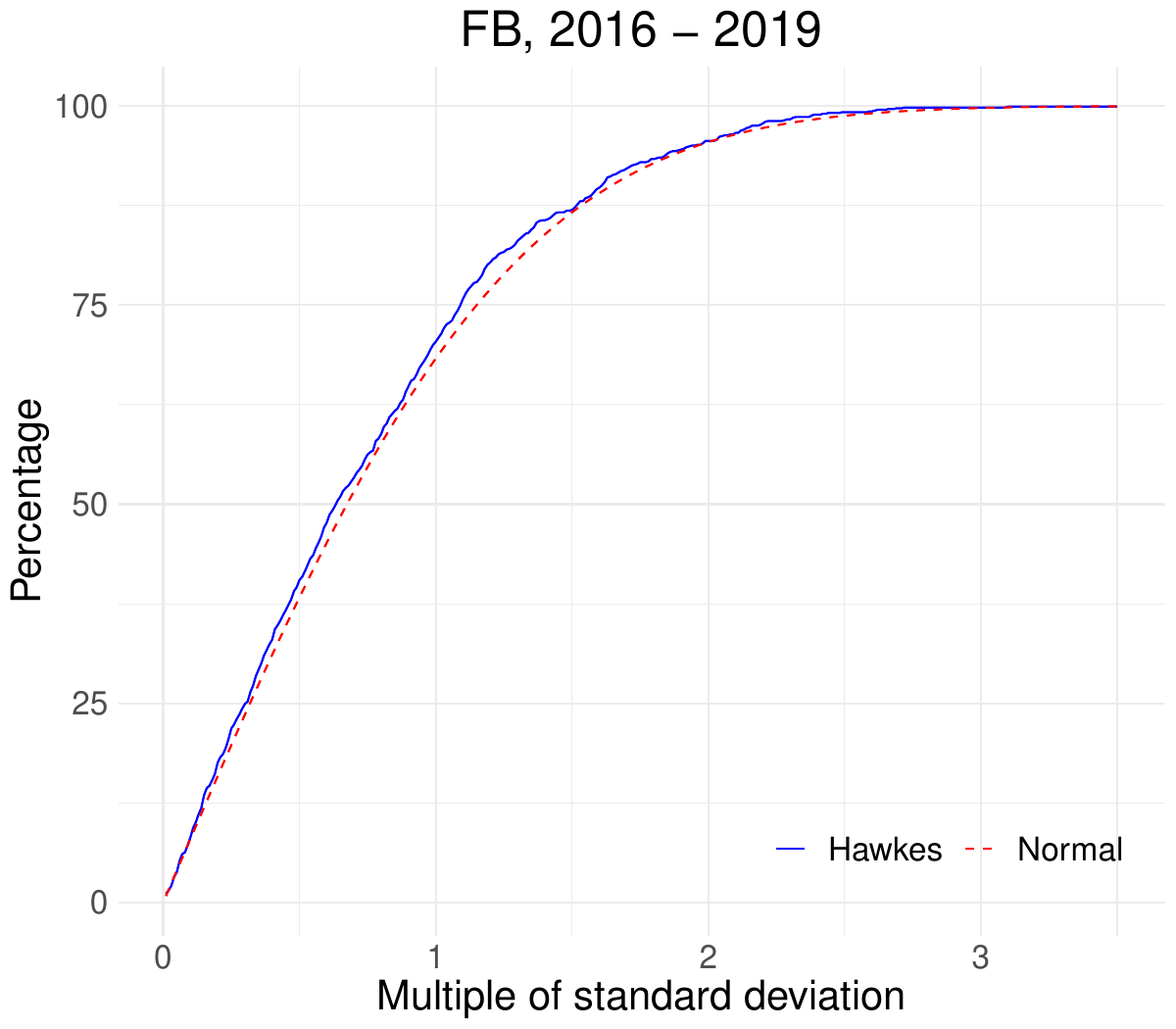}
		\caption{}
		\label{fig:percent_FB}
	\end{subfigure}
	\caption{Percentage of absolute price change that lies within the multiple of standard deviation}\label{Fig:percent}
\end{figure}

\subsection{Intraday estimation}

An important application of Hawkes volatility is the ability to observe time-varying volatility in near real-time.
In other words, we can observe the dynamics of intraday volatility.
This is because of the trading nature of the stock market
that provides sufficient data for estimation, even in a short period.
Under a 30-minute time window, the marked Hawkes model was adequately fitted, and 
volatility was calculated using the estimated parameters.
For small amounts of data, 
we assume that $\bm \alpha$ and $\bm \eta$ are symmetric to achieve model parsimony.
The time interval window for the estimation is reset after 10 seconds.

In Figure~\ref{Fig:intraday1}, we plot examples of intraday volatilities from January 15, 2016, with a timeline from 10:00 AM to 15:00 PM.
In the left panel, it is noteworthy that 
volatility rebounds immediately after noon.
This may indicate that market participants return to the stock market after lunch.
The right panel shows the well-known U-shaped volatility form more precisely.
The real-time Hawkes volatility measure is expected to be useful for traders managing intraday price risk.
This allows them to observe the market volatility of each stock in real time.

\begin{figure}[]
	\begin{subfigure}{0.47\textwidth}
		\centering
		\includegraphics[width=\textwidth]{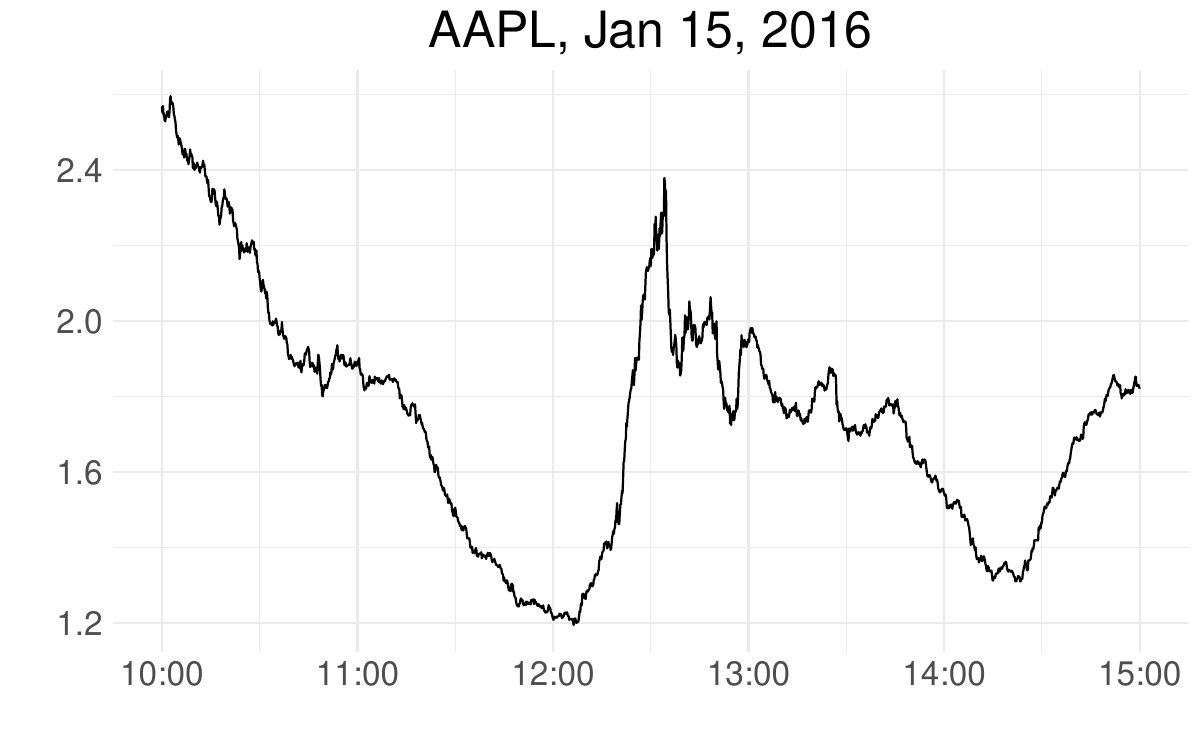}
		\caption{AAPL}
	\end{subfigure}
	\quad
	\begin{subfigure}{0.47\textwidth}
		\centering
		\includegraphics[width=\textwidth]{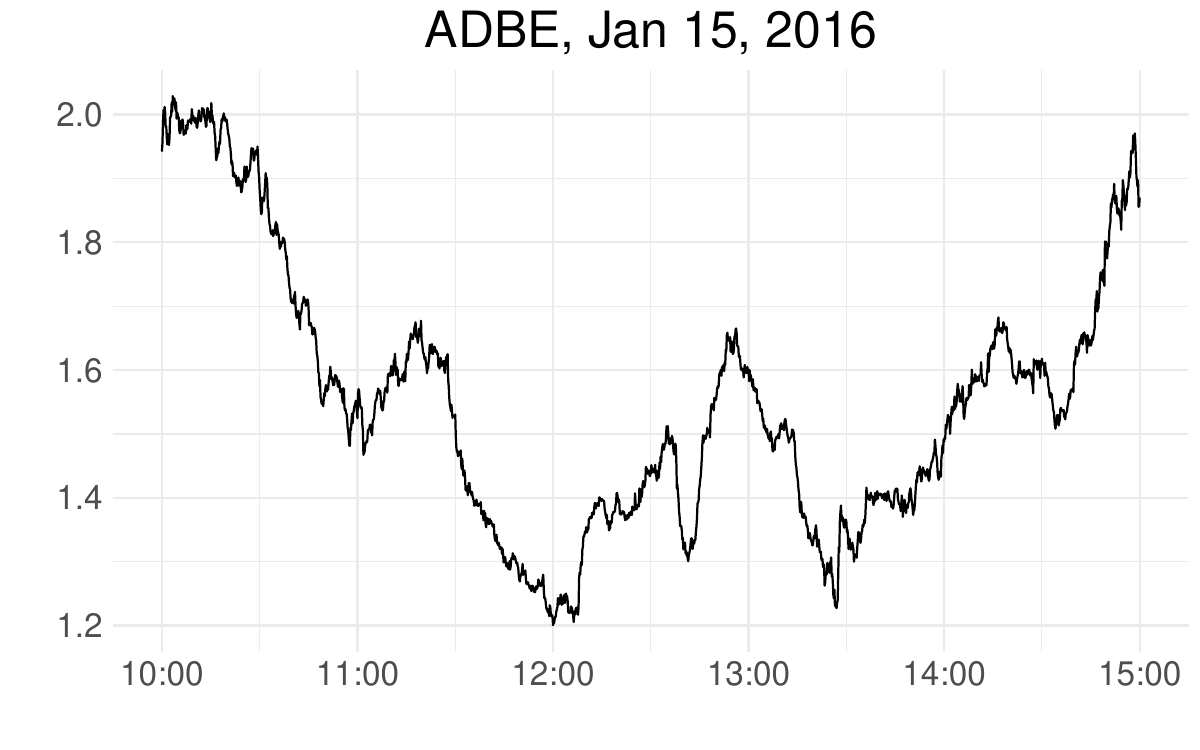}
		\caption{ADBE}
	\end{subfigure}
	\caption{Intraday volatility dynamics in a typical day}\label{Fig:intraday1}
\end{figure}

We apply the simple Hawkes model discussed in Remarks~\ref{Remark:simple}~and~\ref{Remark:simple2} to the unfiltered price processes to estimate intraday volatility. 
Figure~\ref{Fig:intraday_simple} provides further details.
Because no filtering was applied, the raw data retains microstructure noise, 
resulting in volatility estimates that exhibit unstable fluctuations.
Comparing these estimates with those in Figure~\ref{Fig:intraday1}, 
the overall trends in intraday volatilities are similar.

\begin{figure}[]
	\begin{subfigure}{0.47\textwidth}
		\centering
		\includegraphics[width=\textwidth]{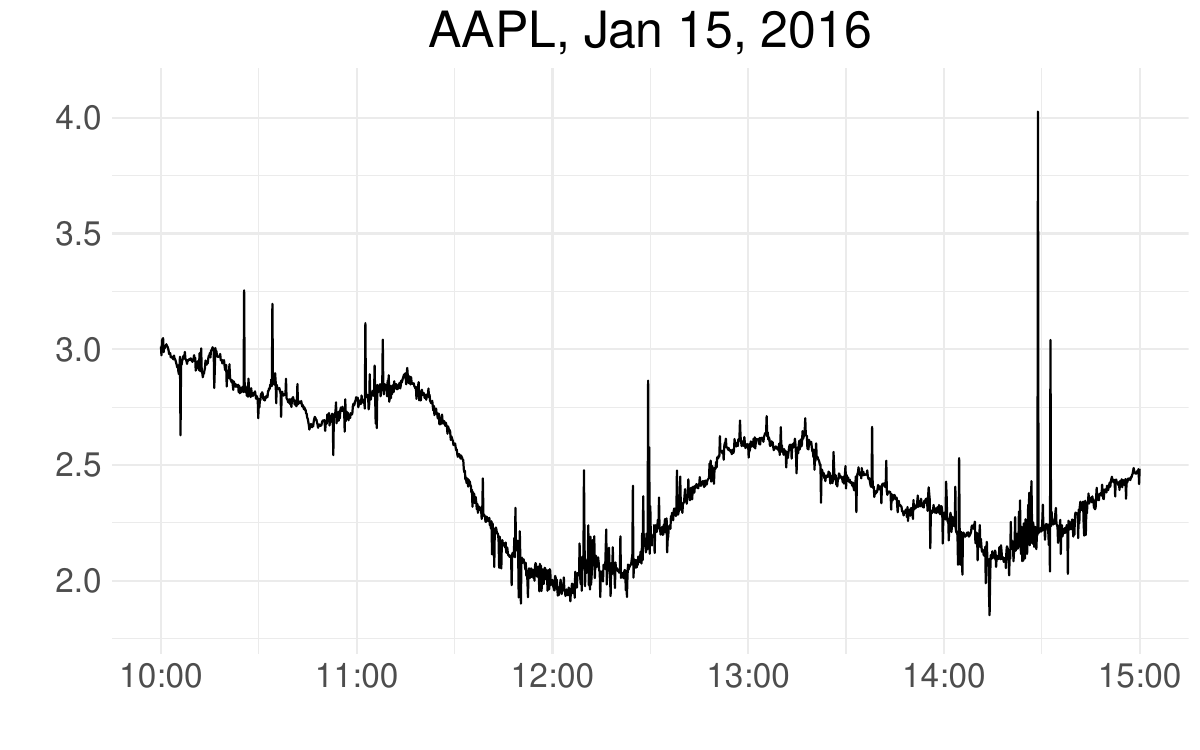}
		\caption{AAPL}
	\end{subfigure}
	\quad
	\begin{subfigure}{0.47\textwidth}
		\centering
		\includegraphics[width=\textwidth]{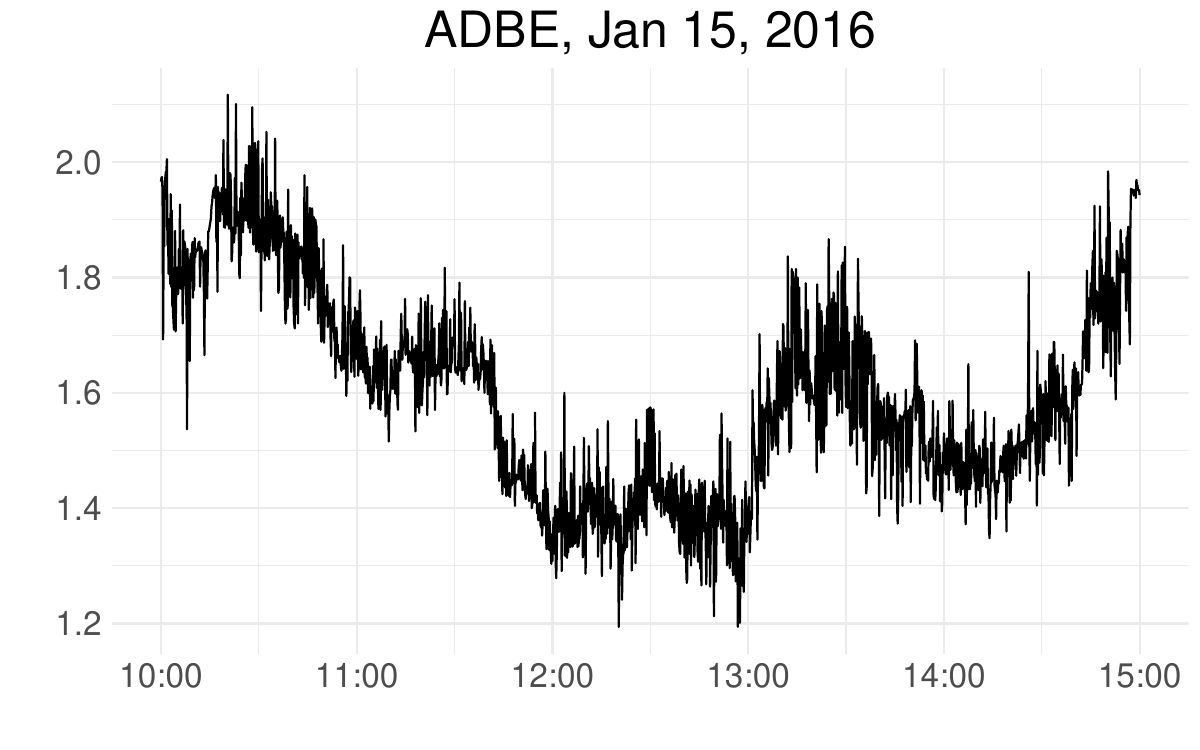}
		\caption{ADBE}
	\end{subfigure}
	\caption{Intraday volatility dynamics in a typical day}\label{Fig:intraday_simple}
\end{figure}

Figure~\ref{Fig:intraday2} shows another example in this regard, observed on February 20, 2020.
On that day, the stock market, which had risen at the start as usual, 
suddenly declined from 10:45 AM for approximately one hour.
Subsequently, it recovered gradually.
Although the exact reason for this decline is unknown, 
the market is suspected to have been highly volatile for some time, 
following the rapid spread of COVID-19.
The Hawkes volatility illustrates the real-time volatility changes observed on this day,
showing that the volatility increased rapidly until approximately 11:30 AM and then gradually stabilized.

\begin{figure}[]
	\begin{subfigure}{0.47\textwidth}
		\centering
		\includegraphics[width=\textwidth]{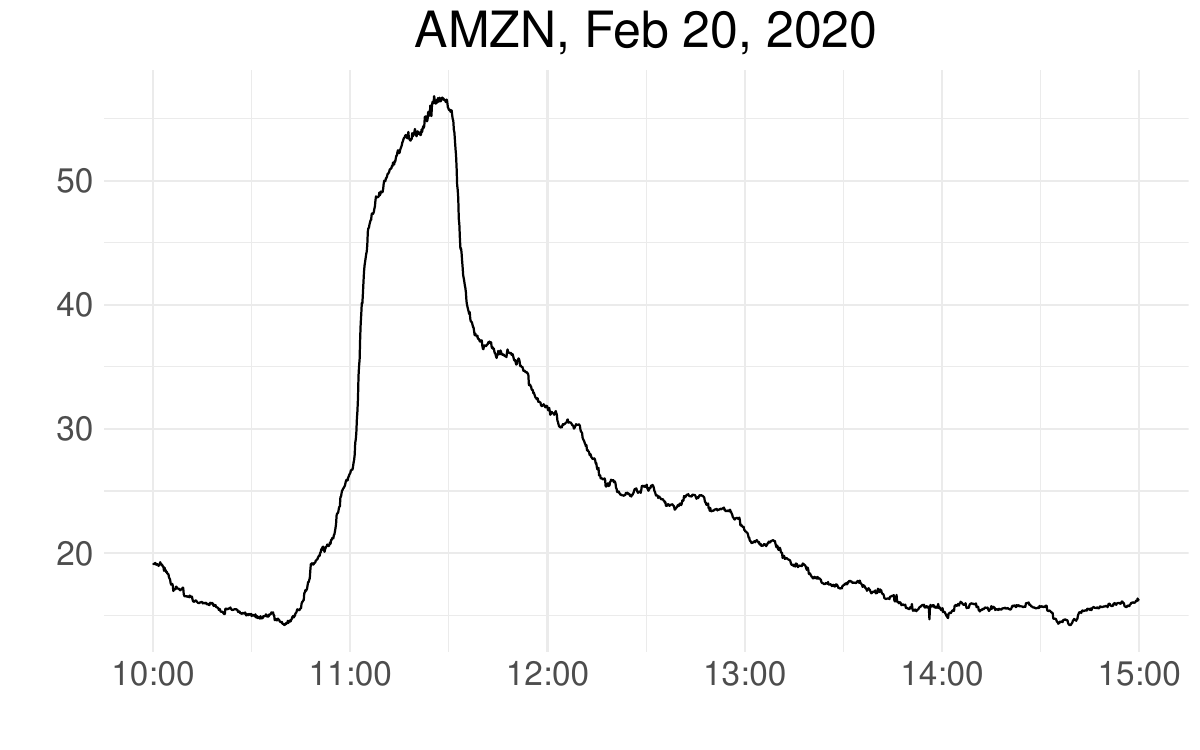}
		\caption{AMZN}
	\end{subfigure}
	\quad
	\begin{subfigure}{0.47\textwidth}
		\centering
		\includegraphics[width=\textwidth]{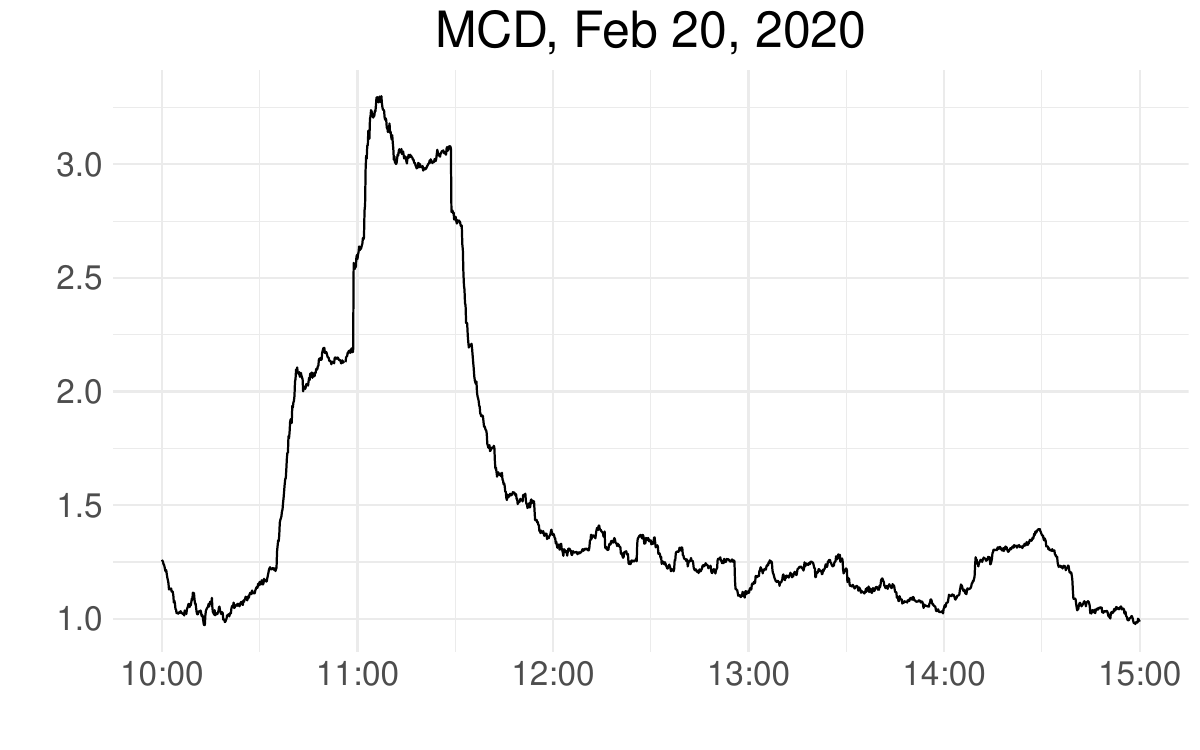}
		\caption{MCD}
		\label{fig:intra_adbe}
	\end{subfigure}
	\caption{Intraday volatility dynamics with sudden change in volatility}\label{Fig:intraday2}
\end{figure}

The typical volatility forecasting used in risk management is performed using a discrete time framework.
Tomorrow's volatility is predicted based on the historical return data up to today.
The introduction of intraday Hawkes volatility and the use of high-frequency data in this study enable real time volatility estimation throughout the day.
We investigate whether intraday volatility, which is updated in real time after the market opens, contains more information about the total real variation of the day's returns.
Does the prediction, including updated intraday volatility,
result in a better volatility forecast than that calculated solely on the return information up to the previous day?
This new information is expected to yield better volatility prediction results.
If so, when does the intraday volatility of the day contain more accurate information than the history up to the previous day?
The Hawkes intraday volatility provides an answer to this question.

The opening and closing times of the stock market are denoted by $t_n$ and $u_n$, respectively, expressed in discrete times for $1 \leq n \leq N$.
Then, the daily return process is 
$$R_n = \frac{S(u_n) - S(t_n)}{S(t_n)}.$$

First, we consider the GJR-GARCH model \citep{glosten1993relation} for traditional volatility prediction in a discrete time setting.
A simple variance model is used to simplify the test.
However, there is no significant difference in the results, 
even when the GARCH model with the mean is used.
In this model, the GARCH volatility on day $n$ is represented by
\begin{equation}
	g^2_n = \omega_g + (\alpha_g + \gamma_g I_{\{R_{n-1} < 0\}} ) \varepsilon_{n-1}^2 + \beta_g g^2_{n-1} \label{Eq:GJR}
\end{equation}
with parameters $\omega_g, \alpha_g, \gamma_g, \beta_g$
and $\varepsilon_{n} = R_n / g_n$.
The estimates of parameters $\hat \omega_g, \hat \alpha_g, \hat \gamma_g, \hat \beta_g$ 
are based on the maximum likelihood estimation with $m$ previous returns $R_{n-m}, \cdots, R_{n-1}$.
We use $m=1,500$ days for the empirical study.
The one step ahead volatility prediction is accomplished by
$$ \hat g^2_{n+1} = \hat \omega_g + (\hat \alpha_g + \hat \gamma_g I_{\{ R_{n} < 0 \}} ) \varepsilon_{n}^2 + \hat \beta_g \hat g^2_{n}.$$
This prediction is assumed to be computed before the opening of the stock market
and can be used for daily risk management from a practical point of view.

After predicting today's volatility in advance using a GARCH-based volatility model and historical data, 
the risk manager watches as the market opens and new tick information begins to flow in.
With this new influx of information, the manager begins to estimate today's Hawkes volatility in real time.
Let
$\hat h_n(T)$ be the estimated Hawkes volatility defined by Corollary~\ref{Cor:var} based on data from the opening time $t_n$ to $t_n + T \leq u_n$ under the symmetric kernel.
We now establish a simple, linearly combined volatility measure:
\begin{equation}
	\sigma_n (T) = \theta_1(T)  g_n + \theta_2(T)  h_n(T) \label{Eq:combined_vol}
\end{equation}
This equation is established to determine whether volatility measurement performance improves when intraday Hawkes volatility up to time $t_n + T$ is added to the traditional GARCH volatility measure.
When applied to empirical studies, the estimates of $g_n$ and $h_n$ are used:
$$ \hat \sigma_n (T) = \theta_1(T)  \hat g_n + \theta_2(T)  \hat h_n(T).$$
Specifically, $\sigma_n (T)$ is the $n$-day volatility, 
a linear combination of the predicted GARCH volatility based on past $m$-return data and 
the Hawkes volatility measured from opening time $t_n$ to $t_n + T$.

Next, we estimate $\theta_1(T)$ and $\theta_2(T)$ using $\hat g_n$ and $\hat \theta_n$, for $1 \leq n \leq N$.
As $n$ varies, the parameters in Eq.~\eqref{Eq:GJR} are re-estimated.
With a fixed $T$, we find $\hat \theta(T)$ such that
\begin{equation}
	\hat \theta_1(T), \hat \theta_2(T) = \argmax_{\theta_1, \theta_2} \ell_T (R_1, \cdots, R_N ; \hat \sigma_1(T), \cdots, \hat \sigma_N(T) ) \label{Eq:llh}
\end{equation}
where $\ell_T$ is the log-likelihood function of $N$ independent mean-zero multivariate normal variables with a vector of standard deviations $\hat \sigma_1(T), \cdots, \hat \sigma_N(T)$.

Using data from AAPL, 2016-2019, the GARCH model is estimated using daily return time series data, and intraday data are used to compute the intraday Hawkes volatility.
As shown on the left side of Figure~\ref{Fig:prediction}, 
the log-likelihood $\ell_T$ in Eq.~\eqref{Eq:llh} monotonically increases with time $T$ after the market opens.
This finding implies that the influx of new information on intraday price changes increases the accuracy of daily volatility predictions.
Intraday changes in $\theta(T)$ over time $T$ are compared on the right-hand side of Figure~\ref{Fig:prediction}.
During the initial market stages, 
$\theta_1$ of GARCH is larger than $\theta_2$. 
That is,
the GARCH volatility predicted before the market opening has more predictive power than the intraday Hawkes volatility.
However, the $\theta_2$ of the Hawkes volatility dominates the role of volatility prediction after approximately 10:00 AM.
Approximately $T = 30$ minutes after the market opens, the information from $t_n$ to $t_n+T$ for predicting daily volatility is very rich, and the role of volatility information from the past days is rapidly reduced.

\begin{figure}[]
	\begin{subfigure}{0.47\linewidth}
		\centering
		\includegraphics[width=\textwidth]{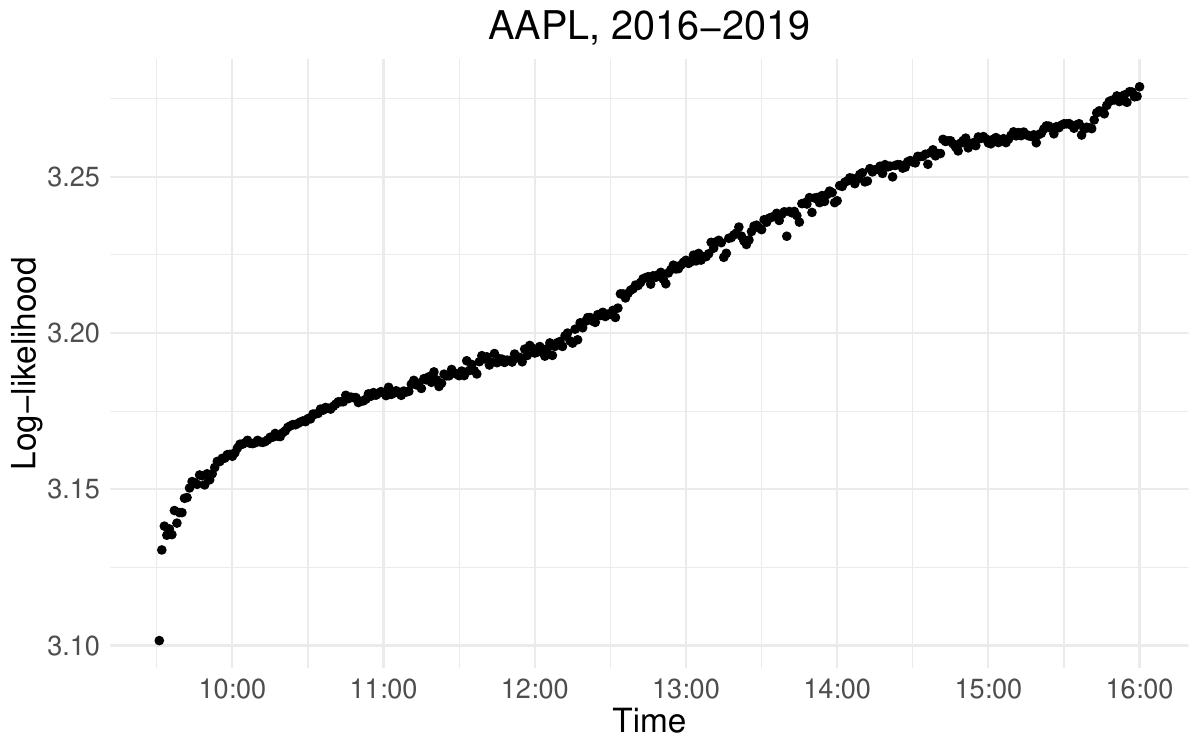}
		\caption{Log-likelihood}
	\end{subfigure}
	\quad
	\begin{subfigure}{0.47\linewidth}
		\centering
		\includegraphics[width=\textwidth]{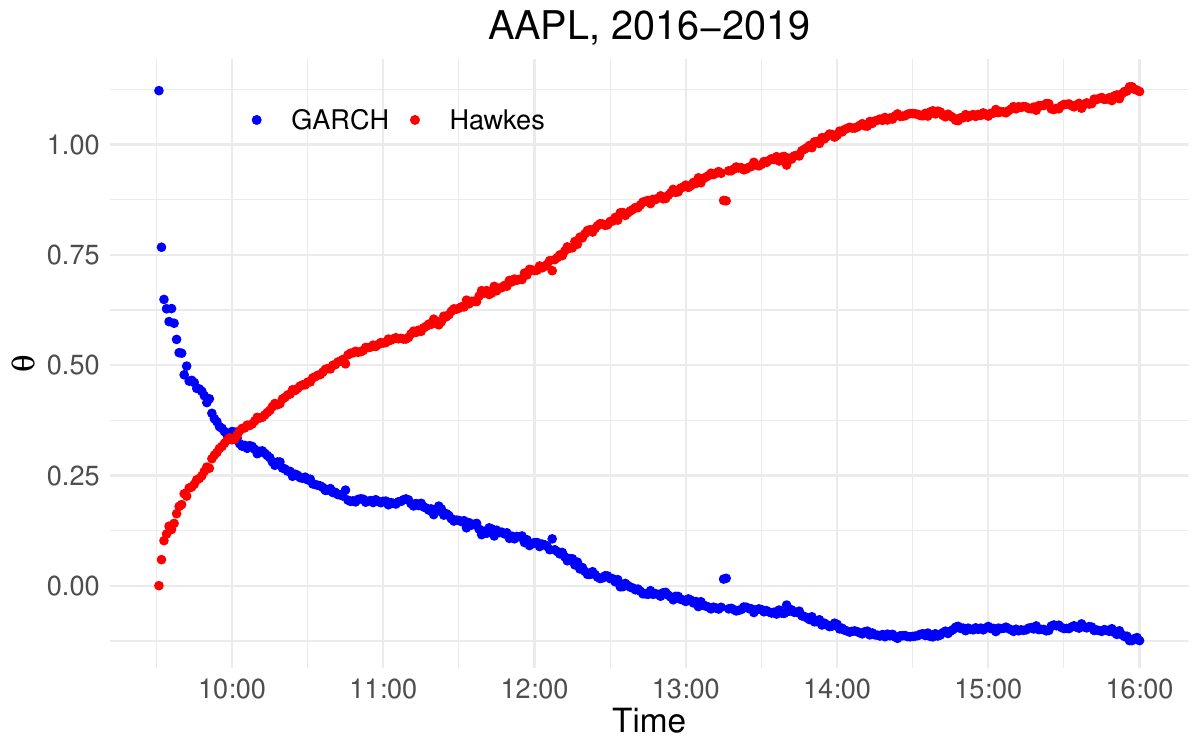}
		\caption{$\theta$}
	\end{subfigure}
	\caption{The predictive power of Hawkes volatility grows stronger over time during a day}\label{Fig:prediction}
\end{figure}

\subsection{Futures and stock volatility}

This subsection examines the relationship between the pre-market volatility of the E-Mini S\&P 500 futures market and stock price volatility during regular market sessions.
The regular stock market opens from 9:30 a.m. to 4 p.m.
However, E-Mini S\&P 500 futures trade from 6:00 p.m. to 5:00 p.m. of the following day.
E-mini S\&P 500 futures are almost continuously traded, even when stocks are not traded on regular exchanges. 
The Hawkes volatility of E-Mini futures can be estimated using data from this period of non-regular market sessions.
This study investigates how much explanatory power the volatility of E-Mini S\&P 500 futures, observed before the market opening (nonregular trading time), 
has for stock market volatility after the market opens.

As in the previous subsection, let $\hat h^s_n(T_1)$ be the estimated intraday Hawkes volatility defined by Corollary~\ref{Cor:var} based on stock prices observed from the opening time, $t_n$, to $t_n + T_1$ of the $n$-th day,
where $t_n + T_1$ is the time between openning and closing of a regular trading session.
In addition,
let $\hat h^f_n(T_2)$ be the estimated Hawkes volatility based on the E-Mini S\&P 500 futures mid-price process from $t_n - T_2$ to opening time $t_n$.
This esitmation is performed prior to the regular stock trading period, 
ensuring that the interval from $t_n - T_2$ does not extend beyond the futures trading halt time of 6:00 p.m.
We then establish a linear model such that
\begin{equation}
	\hat h^s_n(T_1) = \beta_0 + \beta_1 \hat h^f_n(T_2) + \epsilon_n, \label{Eq:reg_fut_stock_1}
\end{equation}
examining the relationship between stock volatility $\hat h^s_n(T_1)$ after market opening and futures volatility $\hat h^f_n(T_2)$ before market opening.
If more than one futures price data point is available on a specific date, the volatility of futures with a closer expiration date is used.

Since the regression analysis in Eq.~\eqref{Eq:reg_fut_stock_1} can be performed for any $0 < T_1 \leq 6.5 \textrm{ hours}$ and $0 < T_2 \leq 16.5 \textrm{ hours}$,
the adjusted R-squared value $R^2(T_1, T_2)$, of the linear model is considered a function of $T_1$ and $T_2$.
By examining the available $T_1$ and $T_2$, we can obtain the overall shape of the R-squared.
The empirical result is presented in Figure~\ref{Fig:Rsq}, 
which displays the R-squared surface as a function of $T_1$ and $T_2$. 
This analysis utilizes the intraday price process of AAPL for stock volatility estimation, 
noting that AAPL constitutes approximately 7\% of the S\&P 500 index.
The data for stock and futures for the estimation of the linear model are for 2019.

\begin{figure}
	\centering
	\includegraphics[width=0.80\textwidth]{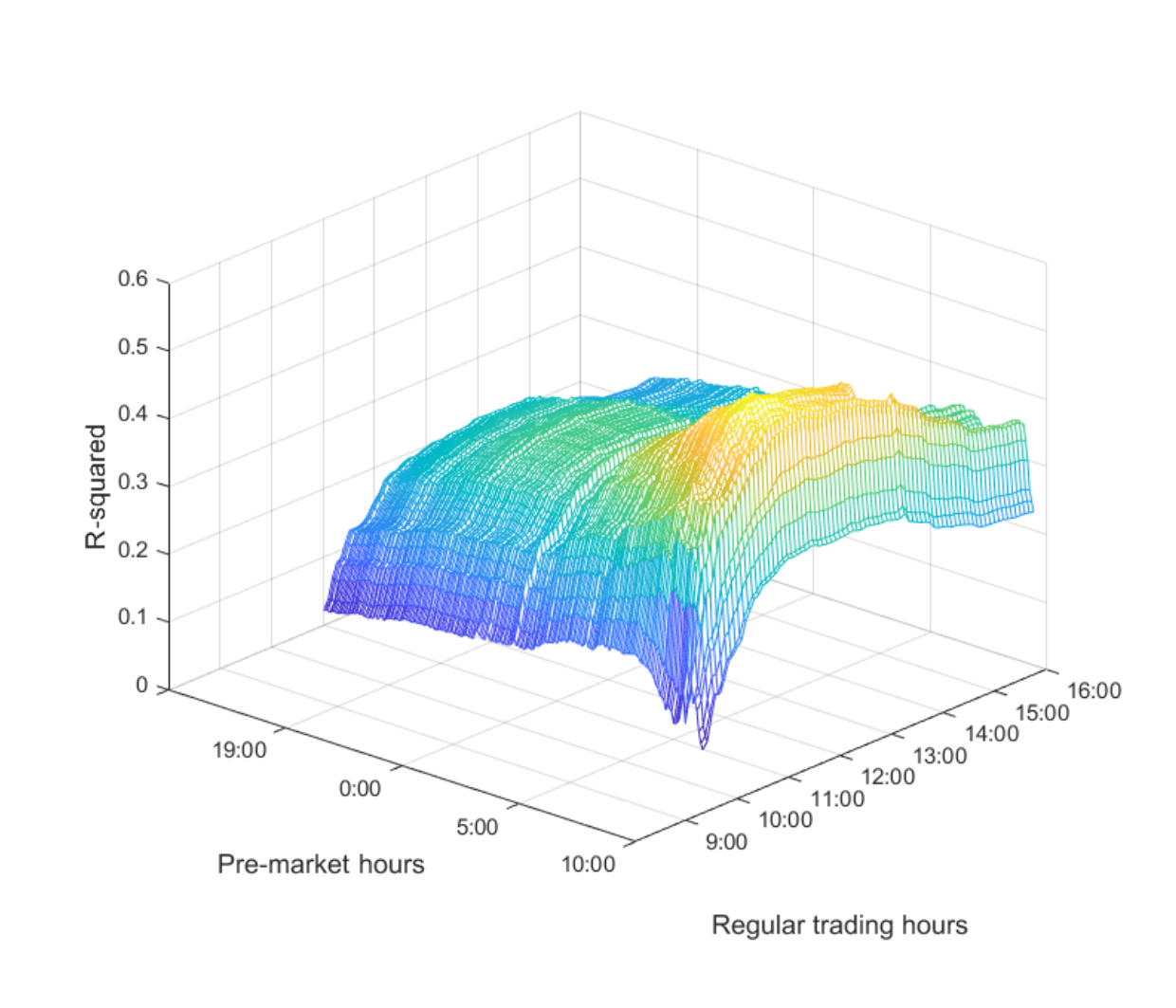}
	\caption{The adjusted R-squared that explains the relationship between the pre-market volatility of E-Mini S\&P 500 futures and stock price volatility during the regular trading time}\label{Fig:Rsq}
\end{figure}

According to the figure, the best pre-market predictor in terms of the R-squared of the Hawkes volatility of AAPL based on E-mini data is E-mini's volatility, calculated using the information accumulated from approximately 7-8 AM to the market opening time of 9:30 AM.
As a target variable, the AAPL intraday volatilities are best predicted based on the period from the market openning to around noon.
More precisely, in this example,
the best predictor is the pre-market E-mini Hawkes volatility estimated from the interval between 7:45 AM and 9:30 AM.
The best-predicted period of the AAPL intraday volatility ranges from 9:30 AM to 11:30 AM, achieving an adjusted R-squared value of 0.52.
In summary, the explanatory power of the pre-market futures volatility for AAPL volatility increases gradually after the market opens, peaks at around 11:30, and gradually decreases over time.

For comparison, we conduct the same analysis using the two-scale realized spot variance (TSRSV) \citep{zu2014estimating}. 
We observe the underlying return processes at 1-second intervals.
From these data, we calculate the subsampled and averaged realized variance using the scale parameter $K$, 
and adjust for noise effects using the conventional realized variance formula to estimate parameters for both futures and stocks.
We test various values of the scale parameters, including 5, 10, 20, 30, and 60, and the results are similar.
Among these, $K=5$ yields the maximum R-squared value of 0.27. 
The best predictor is the TSRSV of futures obtained from from 8:30 AM to 9:30 AM. 
The best-predicted period of AAPL's TSRSV is from 9:30 AM to 10:45 AM.
To save space, we have omit the R-square surface for the TSRSV;
however, the results are similar to those shown in Figure~\ref{Fig:Rsq}.

\subsection{Forecasting}

In the previous section, 
the empirical study shows that market futures volatility explains regular trading-time stock volatility well.
In this section, 
we examine whether pre-market futures volatility can improve stock volatility forecasting more
than forecasting based solely on previous stock market data in terms of the Hawkes volatility.

First, if there is no information on pre-market futures volatility
and the only information about the Hawkes daily volatility of stocks is available,
the following traditional AR(2) model for the daily Hawkes volatility can be considered:
\begin{equation}
	h^s_n = \phi_0 + \phi_1 h^s_{n-1} + \phi_2 h^s_{n-2} + \epsilon_n \label{Eq:AR2}
\end{equation}
where $h^s_n$ is the daily Hawkes volatility on the $n$-th day.
Among the various linear time series models, the AR(2) model is selected using the algorithm \citep{Hyndman}.


More precisely, the daily Hawkes volatility estimation based on tick data for each day preceds as shown in Figure~\ref{Fig:compare}.
Thus, before the $n$-th day, we have $M$ estimated daily Hawkes volatilities $\hat h^s_i$, for $n-M+1 \leq i  \leq n$.
The estimated Hawkes volatilities are used to fit the linear model in Eq.~\eqref{Eq:AR2}
and find the estimates of $\phi_0, \phi_1$ and $\phi_2$.
These estimates are updated daily as new data becomes available.
The one-step ahead forecasting of the Hawkes volatility for day of $n+1$ on the $n$-th day is 
$$ \check h^s_{n+1} = \hat \phi_0 + \hat \phi_1 \hat h^s_{n} + \hat \phi_2 \hat h^s_{n-1}(T).$$ 

Second, we consider the following model, including pre-market futures volatility, as a predictor variable:
\begin{equation}
	h^s_n = \psi_0 + \psi_1 h^s_{n-1} + \psi_2 h^s_{n-2}  + \psi_3 h^f_n(T) + \varepsilon_n \label{Eq:LM}
\end{equation}
where 
$h^f_n(T)$ denotes the pre-market futures Hawkes volatility with time interval $T$, 
measured immediately before the market opens;
that is, from $t_n - T$ to $t_n$.
Compared with Eq.~\eqref{Eq:AR2}, the term $h^s_{n-2}$ is removed.
By introducing $h^f_n(T)$, $h^s_{n-2}$ is no logner significant according to our empirical study.
Similar to before, daily Hawkes volatility estimations for the futures and stock precede.
Subsequently, the model in Eq.~\eqref{Eq:LM} is fitted to determine $\hat \psi_0, \hat \psi_1$ and $\hat \psi_1$  
based on the volatilities estimated before the $n$-th day.
The Hawkes volatility of stock at $(n+1)$-day is forecasted by
$$ \tilde h_{n+1}(T) = \hat \psi_0 + \hat \psi_1 \hat h^s_{n} + \hat \psi_2 \hat h^f_{n+1}(T).$$

In both cases, the root mean square relative errors are defined.
For Eq.~\eqref{Eq:AR2}, we have
$$  \mathrm{RMSRE}_s = \sqrt{ \frac{1}{N} \sum_{n=1}^{N} \left( \frac{\check h^s_{n+1} - \hat h^s_{n+1}}{\check h^s_{n+1}} \right)^2} $$
and for Eq.~\eqref{Eq:LM}, we have
$$  \mathrm{RMSRE}_f(T) = \sqrt{ \frac{1}{N} \sum_{n=1}^{N} \left( \frac{\tilde h^s_{n+1}(T) - \hat h^s_{n+1}}{\tilde h^s_{n+1}(T)} \right)^2} .$$
Note that $\check h^s$ and $\tilde h^s$ are different from $\hat h^s$ such that
$\check h^s_n$ and $\tilde h^s_n$ are computed and forecasted based on past data 
but $\hat h^s_n$ is estimated using the tick data of the $n$-th day.

The empirical results show a smaller error when forecasting is performed using a model that includes the pre-market futures volatility.
Figure~\ref{Fig:predict} illustrates that 
the prediction with futures volatility as a predictor variable improves the RMSRE
using any pre-market time interval $[t_n-T, t_n]$ where $t_n-T$ is represented by the values on the $x$-axis.
In the figure, the dots represent the $\mathrm{RMSRE}_f(T)$
and the solid line represents the $\mathrm{RMSRE}_s$, which does not depend on $T$.
The data used for forecasting spans the years 2018 to 2019, 
and the training for each linear model on the $n$-th day utilizes the previous 400 data points, starting from 2016.
The figure also indicates that pre-market futures volatility, 
measured just before the market opens, further reduces the error.

\begin{figure}
	\centering
	\includegraphics[width=0.75\textwidth]{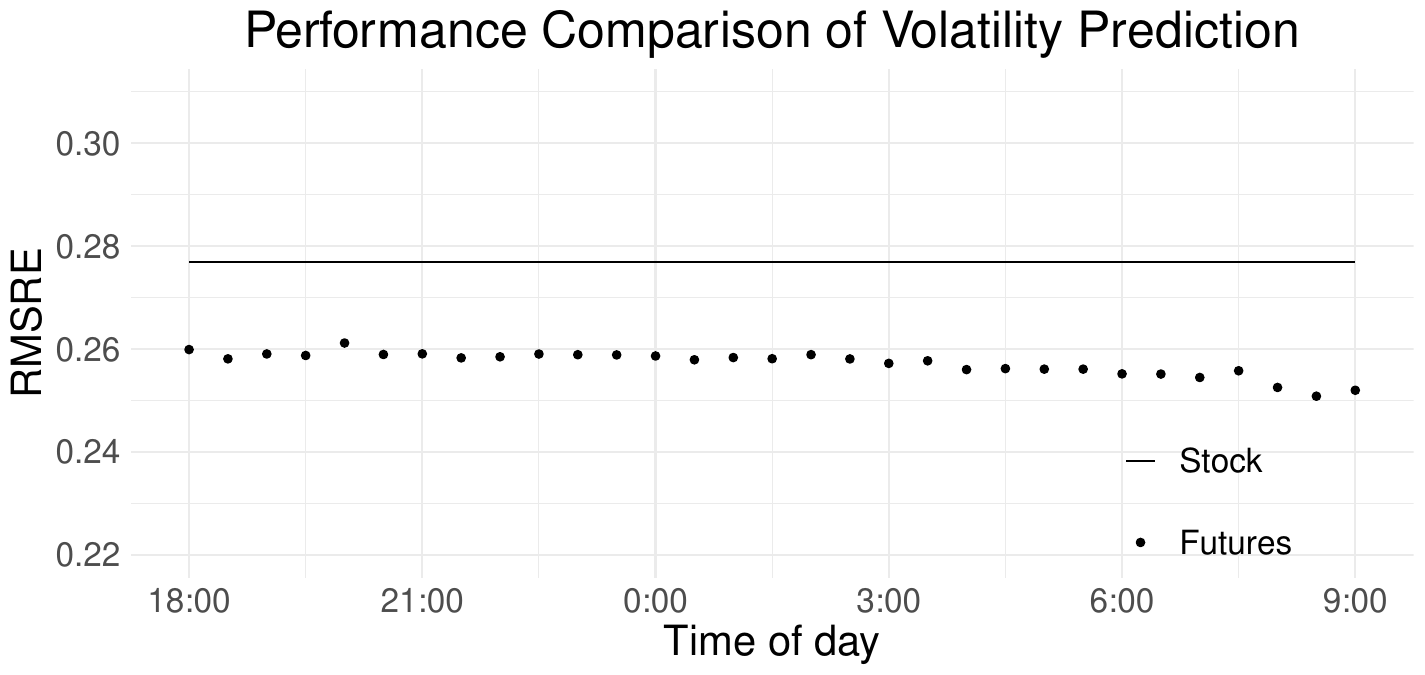}
	\caption{Comparison between $\mathrm{RMSRE}_s$ (solid line) and $\mathrm{RMSRE}_f(T)$ (dotted line)}\label{Fig:predict}
\end{figure}

\section{Conclusion}~\label{Sect:concl}

In this study, the upward and downward movements of price dynamics at the tick level are modeled by self and mutually exciting Hawkes processes.
The variance formulas for the difference between the number of ups and downs are derived for both the unmarked and marked models.
Under the marked Hawkes model, a formula is presented for both dependent and independent mark distributions considering the underlying intensity and counting processes.
The variance formula is applied to the filtered high-frequency data observed every 0.1 second.
The single kernel Hawkes model fits these data well.

By comparing with realized volatility, we confirm that the Hawkes model adequately measures daily volatility.
The biggest advantage of Hawkes volatility is intraday volatility, which allows the precise measurement of real-time instantaneous price variability.
The intraday Hawkes volatility has stronger predictive power for daily volatility, which increases with time.
According to our findings, the importance of the intraday Hawkes volatility approximately 10 AM in the stock market outweighs GARCH volatility, 
a discrete-time volatility prediction model that depends on a time series of daily returns up to the previous day.
The intraday Hawkes  volatility makes it possible to precisely examine the predictive power of volatility in the futures market
before the opening of the stock market, to stock price volatility after market opening.
In addition, using the volatilities of pre-market futures and stocks improves stock market volatility forecasting.


\bibliographystyle{apalike}
\bibliography{wileyNJD-AMA}



\appendix

\section{Proofs}\label{Sect:proof}

\begin{proof}[Proof of Proposition~\ref{Prop:E_lambda}]
	By taking expectations on both sides of Eq.~\eqref{eq:lambda},
	\begin{align*}
		\E[\bm{\lambda}_t] &= \E[\bm{\lambda}_0] + \int_0^t \{ \bm{\beta}\bm{\mu} + (\bm{\alpha} - \bm{\beta}) \E[\bm{\lambda}_s]) \}\D s\\
		&= \E[\bm{\lambda}_t] +  \{ \bm{\beta}\bm{\mu} + (\bm{\alpha} - \bm{\beta}) \E[\bm{\lambda}_t]) \}t,
	\end{align*}
	where we use $\E[\bm{\lambda}_t] = \E[\bm{\lambda}_s] = \E[\bm{\lambda}_0]$
	to obtain Eq.~\eqref{Eq:E_lambda}.
\end{proof}

\begin{proof}[Proof of Proposition~\ref{Prop:E_ll}]
	By Lemma~\ref{Lemma:quad},
	\begin{align*}
		\D(\bm{\lambda}_t \bm{\lambda}_t^{\top}) ={}& \bm{\lambda}_{t} \D \bm{\lambda}_t^{\top} + \D(\bm{\lambda}_t) \bm{\lambda}^{\top}_t + \D [\bm{\lambda}, \bm{\lambda}^{\top}]_t \\
		={}& \left[ \bm{\lambda}_t \{ \bm{\beta} \bm{\mu} + (\bm{\alpha}- \bm{\beta} )\bm{\lambda}_t \}^{\top} +  
		\{\bm{\beta} \bm{\mu} + (\bm{\alpha}- \bm{\beta}) \bm{\lambda}_t \} \bm{\lambda}_t^{\top} \right] \D t
		+ \bm{\alpha}\mathrm{Dg}(\D \bm{N}_t)\bm{\alpha}^{\top}  \\
		&+ \bm{\lambda}_t \{ \bm{\alpha} (\D \bm{N}_t - \bm{\lambda}_t \D t)\}^{\top} +   \bm{\alpha} (\D \bm{N}_t - \bm{\lambda}_t \D t) \bm{\lambda}_t^{\top}
	\end{align*}
	By taking expectation in an integration form,
	\begin{align*}
		\E [\bm{\lambda}_t \bm{\lambda}_t^{\top}] = \E [\bm{\lambda}_0 \bm{\lambda}_0^{\top}] + \int_0^t &\left\{  \E[\bm{\lambda}_s \bm{\lambda}_s^{\top}] (\bm{\alpha}- \bm{\beta} )^{\top}
		+ (\bm{\alpha}- \bm{\beta} ) \E[\bm{\lambda}_s \bm{\lambda}_s^{\top}] \right.\\
		&\left. + \E[\bm{\lambda}_s]   (\bm{\beta} \bm{\mu})^{\top}
		+ \bm{\beta} \bm{\mu} \E[\bm{\lambda}_s^{\top}]
		+ \bm{\alpha}\mathrm{Dg}(\E[\bm{\lambda}_s])\bm{\alpha}^{\top}  \right\} \D s.
	\end{align*}
	By using $\E [\bm{\lambda}_t \bm{\lambda}_t^{\top}]= \E [\bm{\lambda}_s \bm{\lambda}_s^{\top}] = \E [\bm{\lambda}_0 \bm{\lambda}_0^{\top}]$,
	the integrand should be zero, and we have the desired result.
\end{proof}

\begin{proof}[Proof of Proposition~\ref{Prop:E_lN}]
	By Lemma~\ref{Lemma:quad},
	\begin{align*}
		\D(\bm{\lambda}_t \bm{N}_t^{\top}) &= \bm{\lambda}_t \D \bm{N}_t^{\top} + \D(\bm{\lambda}_t) \bm{N}^{\top}_{t-} + \D [\bm{\lambda}\bm{N}^{\top}]_t \\
		&= \bm{\lambda}_t \D \bm{N}_t^{\top} + \{\bm{\beta} (\bm{\mu} - \bm{\lambda}_t) \D t + \bm{\alpha} \D \bm{N}_t \} \bm{N}^{\top}_{t-} +  \bm{\alpha}\mathrm{Dg}(\D\bm{N}_t)
	\end{align*}
	and
	\begin{align*}
		\E[\bm{\lambda}_t \bm{N}^{\top}_t] = \int_0^t \left\{ \E[\bm{\lambda}_s \bm{\lambda}^{\top}_s] + \bm{\beta} \bm{\mu} \E[\bm{\lambda}_s^{\top}] s + 
		(\bm{\alpha} - \bm{\beta} ) \E[\bm{\lambda}_s \bm{N}_s^{\top}] + \bm{\alpha}\mathrm{Dg}(\E[\bm{\lambda}_s]) \right\}\D s
	\end{align*}
	or as a system of linear differential equations,
	\begin{align*}
		\frac{\D \E[\bm{\lambda}_t \bm{N}_t^{\top}]}{\D t} = (\bm{\alpha} - \bm{\beta} ) \E[\bm{\lambda}_t \bm{N}_t^{\top}] + \bm{\beta} \bm{\mu} \E[\bm{\lambda}_t^{\top}] t + \E[\bm{\lambda}_t \bm{\lambda}_t^{\top}] + \bm{\alpha}\mathrm{Dg}(\E[\bm{\lambda}_t]).
	\end{align*}
\end{proof}

\begin{proof}[Proof of Theorem~\ref{Thm:simple_vol}]
	
	Let $f(x_1, x_2) = (x_1 - x_2)^2$.
	The infinitesimal generator on $N_1(t)$ and $N_2(t)$ defined by its action on $f(N_1(t), N_2(t))$ is
	\begin{align*}
		\mathcal{A}f(N_1(t), N_2(t)) ={}& \lim_{h \downarrow 0} \frac{\E[f( N_1(t+h), N_2(t+h)) - f( N_1(t), N_2(t)) \rvert \F_t]}{h}\\
		={}&\lambda_1(t) \{ f(N_1(t) + 1, N_2(t)) - f(N_1(t), N_2(t)\} \\
		&+ \lambda_2(t) \{ f(N_1(t), N_2(t) + 1) - f(N_1(t), N_2(t)) \} \\
		={}& \lambda_1(t) \{ 2N_1(t) - 2N_2(t) + 1) \}  + \lambda_2(t) \{2N_2(t) - 2N_1(t) + 1\} .
	\end{align*}
	By Dynkin's formula,
	\begin{align*}
		&\E \left[ (N_1(t) - N_2(t))^2 \right] = \int_0^t \E[\mathcal{A}f(N_1(s), N_2(s))] \D s \\
		={}& \int_0^t \E \left[ \lambda_1(s) \{ 2N_1(s) - 2N_2(s) + 1) \}  + \lambda_2(s) \{2N_2(s) - 2N_1(s) + 1\} \right] \D s \\
		={}& \int_0^t
		\left\{
		2  \begin{bmatrix} 1 & -1 \end{bmatrix} 
		\E[ \bm{\lambda}_s  \bm{N}_s^{\top}]
		\begin{bmatrix} 1 \\ -1 \end{bmatrix}
		+
		\begin{bmatrix} 1 & 1 \end{bmatrix}
		\E[\bm{\lambda}_s] \right\}
		\D s \\
		={}& \int_0^t
		\left\{
		2  \begin{bmatrix}
			1 & -1
		\end{bmatrix} 
		\left(
		\V \CC \circ \begin{bmatrix} \ee^{\xi_1 s} & \ee^{\xi_1 s} \\ \ee^{\xi_2 s} & \ee^{\xi_2 s} \end{bmatrix}
		+ \Aa s + \BB
		\right)
		\begin{bmatrix}
			1 \\ -1
		\end{bmatrix}
		+
		\begin{bmatrix}
			1 & 1
		\end{bmatrix}
		\E[\bm{\lambda}_s]
		\right\}
		\D s \\
		={}& 
		2  \begin{bmatrix} 1 & -1 \end{bmatrix} 
		\left( \V \DD \circ 
		\begin{bmatrix} \ee^{\xi_1 t} - 1  & \ee^{\xi_1 t} - 1 \\ \ee^{\xi_2 t} - 1 & \ee^{\xi_2 t} - 1 \end{bmatrix}
		+ \frac{1}{2} \Aa t^2 +  \BB t  \right)
		\begin{bmatrix}
			1 \\ -1
		\end{bmatrix}
		+
		\begin{bmatrix}
			1 & 1
		\end{bmatrix}
		\E[\bm{\lambda}_t] t
	\end{align*}
	where
	$$\DD = \CC \circ \begin{bmatrix}\xi_{1}^{-1} & \xi_{1}^{-1} \\ \xi_{2}^{-1} & \xi_{2}^{-1} \end{bmatrix}.$$
	Finally, the variance formula is derived using
	\begin{align*}
		\E [ (N_1(t) - N_2(t)) ]^2 =
		\begin{bmatrix}
			1 & -1
		\end{bmatrix} 
		\E[\bm{\lambda}_s] \E[\bm{\lambda}_s]^{\top} 
		\begin{bmatrix}
			1 \\ -1
		\end{bmatrix} t^2
		=
		\begin{bmatrix}
			1 & -1
		\end{bmatrix} 
		\Aa
		\begin{bmatrix}
			1 \\ -1
		\end{bmatrix} t^2
	\end{align*}	
	and 
	$$ \Var(N_1(t) - N_2(t)) = \E [ (N_1(t) - N_2(t)) ]^2 - (\E[N_1(t) - N_2(t) ])^2.$$
\end{proof}

\begin{proof}[Proof of Proposition~\ref{Prop:E_lambda2}]
	Since
	$$ \E \left[ \mathbb{K}_t [\bm{g}(z)] \bm{\lambda}_t \right] =  \bm{\eta} \circ (\overline \ZZ - 1) \E[\bm{\lambda}_t],$$
	we have
	\begin{align*}
		\E[\bm{\lambda}_t]
		&= \E[\bm{\lambda}_t] +  \{ \bm{\beta}\bm{\mu} + (\bm{\alpha} + \bm{\eta} \circ (\overline \ZZ - 1) - \bm{\beta}) \E[\bm{\lambda}_t]  \}t.
	\end{align*}
	Using $\E[\bm{\lambda}_t] = \E[\bm{\lambda}_s] = \E[\bm{\lambda}_0]$, we obtain Eq.~\eqref{Eq:E_lambda2}.
\end{proof}

\begin{proof}[Proof of Lemma~\ref{Lemma:useful}]
	Since these expressions are the results of simple matrix calculations, only a few parts of the proofs are briefly explained.
	For Eq.~\eqref{eq:e2}, since
	\begin{align*}
		\bm{g}(z) \Dg (\bm{M}(\D s \times \D z)) \bm{g}(z)^{\top} ={}& (\bm{\alpha} - \bm{\eta} + \bm{\eta} \circ \ZZ) \Dg (\bm{M}(\D s \times \D z)) (\bm{\alpha} - \bm{\eta} + \bm{\eta} \circ \ZZ)^{\top} \\
		={}& (\bm{\alpha} - \bm{\eta} + \bm{\eta} \circ \ZZ) \Dg (\bm{M}(\D s \times \D z)) (\bm{\alpha} - \bm{\eta})^{\top} \\
		& + (\bm{\alpha} - \bm{\eta} ) \Dg (\bm{M}(\D s \times \D z)) (\bm{\eta} \circ \ZZ)^{\top} \\
		& + (\bm{\eta} \circ \ZZ) \Dg (\bm{M}(\D s \times \D z)) (\bm{\eta} \circ \ZZ)^{\top}
	\end{align*}
	and
	\begin{align*}
		\begin{aligned}
			&\E \left[ \int_{(0,t] \times E} (\bm{\eta} \circ \ZZ) \Dg (\bm{M}(\D s \times \D z)) (\bm{\eta} \circ \ZZ)^{\top}   \right] \\
			& = \int_{(0,t] \times E} 
			\left[\begin{matrix}
				\eta_{11}^2 \E [z_1^2 M_1(\D s \times \D z_1)] + \eta_{12}^2 \E [z_2^2 M_2(\D s \times \D z_2)] \\
				\eta_{11}\eta_{21} \E [z_1^2 M_1(\D s \times \D z_1)] + \eta_{12}\eta_{22}\E [z_2^2 M_2(\D s \times \D z_2)]
			\end{matrix}\right.\\
			&\qquad\qquad\qquad\qquad\qquad
			\left.\begin{matrix}
				\eta_{11}\eta_{21}\E [z_1^2 M_1(\D s \times \D z_1)] + \eta_{12}\eta_{22} \E [z_2^2 M_2(\D s \times \D z_2)] \\
				\eta_{21}^2 \E [z_1^2 M_1(\D s \times \D z_1)] + \eta_{22}^2 \E [z_2^2 M_2(\D s \times \D z_2)]  
			\end{matrix}\right]\D s\\
			& = \int_{(0,t] \times E} 
			\left[\begin{matrix}
				\eta_{11}^2 Z^{(2)}_{1} \E[\lambda_1(s)]+ \eta_{12}^2 Z^{(2)}_{2} \E[\lambda_2(s)]\\
				\eta_{11}\eta_{21} Z^{(2)}_{1} \E[\lambda_1(s)] + \eta_{12}\eta_{22} Z^{(2)}_{2} \E[\lambda_2(s)]
			\end{matrix}\right.\\
			&\qquad\qquad\qquad\qquad\qquad\qquad\qquad\qquad
			\left.\begin{matrix}
				\eta_{11}\eta_{21} Z^{(2)}_{1} \E[\lambda_1(s)] + \eta_{12}\eta_{22}Z^{(2)}_{2} \E[\lambda_2(s)]\\
				\eta_{21}^2 Z^{(2)}_{1} \E[\lambda_1(s)] + \eta_{22}^2 Z^{(2)}_{2} \E[\lambda_2(s)] 
			\end{matrix}\right]\D s\\
			&= \int_0^t \left(\bm{\eta} \circ \ZZ^{(2)\circ\frac{1}{2}}\right) \Dg(\E[\bm{\lambda}_s]) \left(\bm{\eta} \circ \ZZ^{(2)\circ\frac{1}{2}}\right)^{\top} \D s
		\end{aligned}
	\end{align*}
	we have the desired result.
	For Eq.~\eqref{eq:e3},
	\begin{equation*}
		(\bm{\alpha}  + \bm{g}(z)) \Dg(\bm{M}(\D s \times \D z)) \Dg(z) = (\bm{\alpha} - \bm{\eta}) \Dg(\bm{M}(\D s \times \D z)) \Dg(z)   + (\bm{\eta} \circ \ZZ) \Dg(\bm{M}(\D s \times \D z)) \Dg(z) 
	\end{equation*}
	and
	\begin{align*}
		&\E \left[ \int_{(0,t]\times E}  (\bm{\alpha} - \bm{\eta}) \Dg(\bm{M}(\D s \times \D z)) \Dg(z)  \right] \\
		&= \int_{(0,t]\times E} \E \left[\begin{matrix}
			(\alpha_{11} - \eta_{11}) z_1 M_1(\D s \times \D z_1) & (\alpha_{12} - \eta_{12})z_2 M_2 (\D s \times \D z_2) \\
			(\alpha_{21} - \eta_{21}) z_1 M_1(\D s \times \D z_1) & (\alpha_{22} - \eta_{22})z_2 M_2 (\D s \times \D z_2)
		\end{matrix}\right] \D s\\
		&= (\bm{\alpha} - \bm{\eta}) \circ \overline \ZZ \Dg(\E[\bm{\lambda}_t]) t.
	\end{align*}
	Similarly,
	\begin{align*}
		\E \left[ \int_{(0,t]\times E}  (\bm{\eta} \circ \ZZ) \Dg(\bm{M}(\D s \times \D z)) \Dg(z) \right] &= \E \left[ \int_{(0,t]\times E}  (\bm{\eta} \circ \ZZ^{\circ 2}) \Dg(\bm{M}(\D s \times \D z)) \right] \\
		&= \bm{\eta} \circ \overline  \ZZ^{(2)}  \Dg(\E[\bm{\lambda}_t]) t.
	\end{align*} 
\end{proof}

\begin{proof}[Proof of Lemma~\ref{Lemma:XY}]
	Use
	\begin{align*}
		\bm{X}_t \bm{Y}^{\top}_t &= \bm{X}_0 \bm{Y}^{\top}_0 + \int_0^t \bm{X}_{s-} \D \bm{Y}^{\top}_s + \int_0^t \D(\bm{X}_s) \bm{Y}^{\top}_{s-} + [\bm{X}\bm{Y}^{\top}]_t 
	\end{align*}
	where
	\begin{align*}
		\int_0^t \bm{X}_{s-} \D \bm{Y}^{\top}_u &= \int_0^t \bm{X}_{s} \bm{b}_s^{\top} \D s + \int_{(0,t] \times E} \bm{X}_{s-} \bm{M}^{\top}(\D s \times \D z) \bm{f}_y(s,z)^{\top} \\
		\int_0^t \D(\bm{X}_s) \bm{Y}^{\top}_{s-} &= \int_0^t \bm{a}_s \bm{Y}_{s}^{\top}  \D s + \int_{(0,t] \times E} \bm{f}_x(s,z) \bm{M}(\D s \times \D z) \bm{Y}_{s-}^{\top}\\ 
		[\bm{X}\bm{Y}^{\top}]_t & = \int_{(0,t] \times E} \bm{f}_x(s,z) \mathrm{Dg} (\bm{M}(\D s \times \D z)) \bm{f}_y(s,z)^{\top}
	\end{align*}
	and
	\begin{align*}
		&\E \left[ \int_{(0,t] \times E} \bm{X}_{s-} \bm{M}^{\top}(\D s \times \D z) \bm{f}_y(s,z)^{\top} \right] = \int_0^t  \E\left[ \bm{X}_{s} \bm{\lambda}_s^{\top} \mathbb{K}_s [\bm{f}_y(s,z)]^{\top} \right]  \D s  \\
		&\E \left[ \int_{(0,t] \times E} \bm{f}_x(s,z) \bm{M}(\D s \times \D z) \bm{Y}_{s-}^{\top} \right] =  \int_0^t  \E\left[\mathbb{K}_s [\bm{f}_x(s,z)] \bm{\lambda}_s \bm{Y}^{\top}_{s} \right] \D s.
	\end{align*}
\end{proof}

\begin{proof}[Proof of Proposition~\ref{Prop:E_ll2}]
	Since the intensity matrix process follows Eq.~\eqref{eq:marklambda},
	using Eq.~\eqref{eq:diff_xy} with
	$$\bm{X}_t=\bm{Y}_t=\bm{\lambda}_t, \quad  \bm{a}_t = \bm{b}_t = \bm{\beta}(\bm{\mu} -  \bm{\lambda}_t), \quad \bm{f}_x = \bm{f}_y = \bm{\alpha} + \bm{g}(z) $$
	we have
	\begin{align}
		\frac{\D\E[\bm{\lambda}_t \bm{\lambda}^{\top}_t]}{\D t} ={}& \E[ \bm{\lambda}_{t} (\bm{\mu} -  \bm{\lambda}_t)^{\top} \bm{\beta}^{\top}] + \E[\bm{\lambda}_t \bm{\lambda}_t^{\top} (\bm{\alpha} + \mathbb{K}_t[\bm{g}(z)])^{\top} ] \nonumber\\
		&+ \E[ \bm{\beta}(\bm{\mu} -  \bm{\lambda}_t) \bm{\lambda}_{t}^{\top}]
		+ \E[(\bm{\alpha} + \mathbb{K}_t[\bm{g}(z)])  \bm{\lambda}_t  \bm{\lambda}_t^{\top} ] \nonumber\\
		&+  \frac{\D}{\D t} \E \left[ \int_{(0,t] \times E}  (\bm{\alpha} + \bm{g}(z))   \mathrm{Dg} (\bm{M}(\D s \times \D z))  (\bm{\alpha} + \bm{g}(z))^{\top} \right] = 0. 
	\end{align}
	By Lemma~\ref{Lemma:EK},
	\begin{align*}
		&\E\left[\bm{\lambda}_t \bm{\lambda}_t^{\top}\mathbb{K}_t[\bm{g}(z)]^{\top}\right] = \left(( \overline \ZZ_{\bm{\lambda}\bm{\lambda}^{\top}}^{\top} - 1 )\circ \E[ \bm{\lambda}_t \bm{\lambda}_t^{\top} ] \right) \bm{\eta}^{\top} \\
		&\E\left[\mathbb{K}_t[\bm{g}(z)]\bm{\lambda}_t \bm{\lambda}_t^{\top}\right] =  \bm{\eta} \left(( \overline \ZZ_{\bm{\lambda}\bm{\lambda}^{\top}} - 1 )\circ \E[ \bm{\lambda}_t \bm{\lambda}_t^{\top} ] \right),
	\end{align*}
	and using Eq.~\eqref{eq:e2}, we obtain Eq.~\eqref{eq:ELL} for $\E[\bm{\lambda}_t \bm{\lambda}_t^{\top}]$.
\end{proof}

\begin{proof}[Proof of Proposition~\ref{Prop:E_lN2}]
	Using Eq.~\eqref{eq:diff_xy} with
	$$\bm{X}_t=\bm{N}_t, \quad \bm{Y}_t=\bm{\lambda}_t, \quad  \bm{a}_t = 0, \quad \bm{b}_t = \bm{\beta}(\bm{\mu} -  \bm{\lambda}_t), \quad \bm{f}_x = \mathrm{Dg}(z), \quad   \bm{f}_y = \bm{\alpha} + \bm{g}(z) $$
	we have
	\begin{align*}
		\frac{\D}{\D t} \E [ \bm{N}_t \bm{\lambda}_t^{\top}] ={}& 
		\E[\bm{N}_t (\bm\mu - \bm{\lambda}_t)^{\top} \bm{\beta}^{\top}] +  \E[  \bm{N}_t \bm{\lambda}_t^{\top} (\bm{\alpha} + \mathbb{K}_t[\bm{g}(z)])^{\top} ]
		+ \E [\mathbb{K}_t[\Dg(z)] \bm{\lambda}_t \bm{\lambda}_t^{\top} ]  \\
		&+ \frac{\D}{\D t} \E \left[ \int_{(0,t] \times E)} \Dg(z) \Dg(\bm{M}(\D s \times \D z)) (\bm{\alpha} + \bm{g}(z))^{\top}  \right] \\
		={}& \Dg(\overline \ZZ) \E[\bm{\lambda}_t](\bm{\beta}\bm{\mu})^{\top}  t + \E[ \bm{N}_t \bm{\lambda}_t^{\top}](\bm{\alpha}- \bm{\beta})^{\top} \\
		&+ 
		((\overline \ZZ_{\bm{N}\bm{\lambda}^{\top}} - 1 ) \circ \E[\bm{N}_t \bm{\lambda}_t^{\top} ]) \bm{\eta}^{\top} 
		+ \overline \ZZ^{\top}_{\bm{\lambda}\bm{\lambda}^{\top}}  \circ \E [\bm{\lambda}_t \bm{\lambda}_t^{\top}] \\
		&+ \Dg(\E[\bm{\lambda}_t]) \left((\bm{\alpha} - \bm{\eta}) \circ \overline \ZZ  + \bm{\eta} \circ \overline \ZZ^{(2)}\right)^{\top}.
	\end{align*}
	As in the previous cases, let
	$$ \E[ \bm{N}_t \bm{\lambda}_t^{\top}] \approx \Aa t + \BB,$$
	then $\Aa$ and $\BB$ satisfy Eqs.~\eqref{eq:A}~and~\eqref{eq:B}.
\end{proof}

\begin{proof}[Proof of Theorem~\ref{Thm:var}]
	Since
	\begin{align*}
		\frac{\D}{\D t} \E[ \bm{N}_t \bm{N}_t^{\top}] ={}& \E[ \bm{N}_t \bm{\lambda}_t^{\top} \mathbb{K}_t[ \Dg(z)]] + \E[  \mathbb{K}_t[ \Dg(z)] \bm{\lambda}_t \bm{N}_t^{\top}  ] \\
		&+ \frac{\D}{\D t} \E \left[\int_{(0,t]\times E} \Dg(z) \Dg(\bm{M}(\D s \times \D z)) \Dg(z) \right] \\
		={}& \overline \ZZ_{\bm{N}\bm{\lambda}^{\top}} \circ \E[ \bm{N}_t \bm{\lambda}_t^{\top}] + \overline \ZZ_{\bm{N}\bm{\lambda}^{\top}}^{\top} \circ \E[ \bm{\lambda}_t \bm{N}_t^{\top}] + 
		\overline \ZZ^{(2)} \circ  \Dg (\E [\bm{\lambda}_t])
	\end{align*}
	and
	$$ \E[N_1(t) - N_2(t)]^2   =  \U^{\top} \E[ \bm{N}_t \bm{N}_t^{\top}] \U - \left( \U^{\top} \Dg(\overline \ZZ) \E[\bm{\lambda}_t] t \right)^2, $$
	we have the variance formula.
\end{proof}

\begin{proof}[Proof of Corollary~\ref{Cor:var}]
	If $\overline \ZZ = \overline \ZZ_{\bm{N}\bm{\lambda}^{\top}}$, then $\Aa$ satisfies Eq.~\eqref{eq:A2} which is symmetric,
	and
	\begin{align*}
		\U^{\top} \mathcal{T} \left\{\overline \ZZ_{\bm{N}\bm{\lambda}^{\top}} \circ \left(\frac{1}{2}\Aa t^2 \right) \right\} \U &=  \U^{\top} \left\{ \overline \ZZ \circ \left( \Dg(\overline \ZZ) \E[\bm{\lambda}_t] \E[\bm{\lambda}_t]^{\top} \right) \right\} \U = \left( \U^{\top} \Dg(\overline \ZZ) \E[\bm{\lambda}_t] t \right)^2.
	\end{align*}
\end{proof}

\newpage

\afterpage{%
	\clearpage
	\thispagestyle{empty}
	\begin{landscape}
		\section{Estimation example}\label{Sect:estimate}

		Examples of the maximum likelihood estimates for the Hawkes model are presented in Table~\ref{Table:estimates}.

		\centering 
		\captionof{table}{Examples of the maximum likelihood estimates of the Hawkes model based on the NBBO of NVDA with 0.1-second filtering; the first row of each date presents the estimates and the second row presents the standard error}\label{Table:estimates}
		\begin{tabular}{cccccccccccccc}
			\hline
			date & $\mu_1$  & $\mu_2$ &  $\alpha_{1.1}$ & $\alpha_{1.2}$ & $\alpha_{2.1}$ & $\alpha_{2.2}$ & $\beta_1$ & $\beta_2$ & $\eta_{1.1}$ & $\eta_{1.2}$ & $\eta_{2.1}$ & $\eta_{2.2}$ & llh \\
			\hline
			20191001 & 0.2017 & 0.2437 & 0.1447 & 0.0894 & 0.1248 & 0.157  & 0.5994 & 0.7947 & 0.0271 & 0.0137 & 0.0164 & 0.0455 & -30189 \\
			& 0.0079 & 0.0076 & 0.0123 & 0.0099 & 0.0126 & 0.0124 & 0.0443 & 0.0472 & 0.0038 & 0.0032 & 0.0039 & 0.0050 &          \\
			20191002 & 0.2556 & 0.2302 & 0.1141 & 0.1293 & 0.0504 & 0.1525 & 0.7721 & 0.5664 & 0.0349 & 0.0279 & 0.0199 & 0.0308 & -30983 \\
			& 0.0137 & 0.0085 & 0.0220 & 0.0198 & 0.0102 & 0.0122 & 0.1535 & 0.0442 & 0.0058 & 0.0061 & 0.0034 & 0.0043 &          \\
			20191004 & 0.1691 & 0.1671 & 0.1281 & 0.1009 & 0.1148 & 0.0927 & 0.5300 & 0.5033 & 0.0314 & 0.013  & 0.0156 & 0.0330 & -27744 \\
			& 0.0067 & 0.0071 & 0.0101 & 0.0106 & 0.0112 & 0.0095 & 0.0347 & 0.0371 & 0.0043 & 0.0037 & 0.0036 & 0.0043 &          \\
			20191007 & 0.2378 & 0.2284 & 0.0683 & 0.0716 & 0.0841 & 0.1068 & 0.4127 & 0.5021 & 0.0266 & 0.0183 & 0.0131 & 0.0250 & -31509 \\
			& 0.0157 & 0.0121 & 0.0107 & 0.0145 & 0.0145 & 0.0119 & 0.0783 & 0.0700 & 0.0050 & 0.0034 & 0.0033 & 0.0045 &          \\
			20191008 & 0.2096 & 0.2091 & 0.1028 & 0.1213 & 0.1157 & 0.0908 & 0.4897 & 0.4806 & 0.0251 & 0.0109 & 0.0113 & 0.0311 & -31609   \\
			& 0.0091 & 0.0102 & 0.0099 & 0.0117 & 0.0131 & 0.0103 & 0.0389 & 0.049  & 0.0035 & 0.0028 & 0.0028 & 0.0039 &          \\
			20191009 & 0.2149 & 0.1764 & 0.0738 & 0.1757 & 0.1264 & 0.0979 & 0.6455 & 0.5485 & 0.0485 & 0.0087 & 0.0005 & 0.0453 & -28724   \\
			& 0.0077 & 0.0077 & 0.0102 & 0.0146 & 0.0115 & 0.0102 & 0.0480 & 0.0432 & 0.0058 & 0.0044 & 0.0038 & 0.0052 &          \\
			20191010 & 0.2494 & 0.2138 & 0.0811 & 0.0873 & 0.1121 & 0.1100 & 0.5253 & 0.5887 & 0.0413 & 0.0114 & 0.0233 & 0.0282 & -30643 \\
			& 0.0093 & 0.0084 & 0.0101 & 0.0107 & 0.0118 & 0.0116 & 0.0494 & 0.0493 & 0.0049 & 0.0034 & 0.0041 & 0.0043 &          \\
			20191011 & 0.2183 & 0.1939 & 0.0521 & 0.1570 & 0.1368 & 0.1182 & 0.5975 & 0.6337 & 0.0445 & 0.0093 & 0.0210  & 0.0426 & -29412   \\
			& 0.0079 & 0.0086 & 0.0097 & 0.0160 & 0.0160 & 0.0124 & 0.0569 & 0.0646 & 0.0051 & 0.0036 & 0.0042 & 0.0054 &          \\
			20191014 & 0.1302 & 0.1091 & 0.0954 & 0.1534 & 0.1260 & 0.0921 & 0.5462 & 0.4448 & 0.0233 & 0.0132 & 0.0042 & 0.0383 & -23533 \\
			& 0.0068 & 0.0057 & 0.0120 & 0.0190 & 0.0128 & 0.0100 & 0.0689 & 0.0392 & 0.0053 & 0.0050 & 0.0040 & 0.0053 &          \\
			20191015 & 0.1905 & 0.2261 & 0.1029 & 0.0601 & 0.1292 & 0.1020 & 0.3521 & 0.6533 & 0.0197 & 0.0102 & 0.0227 & 0.0392 & -31026 \\
			& 0.0110 & 0.0092 & 0.0100 & 0.0092 & 0.0146 & 0.0111 & 0.0407 & 0.0622 & 0.0034 & 0.0027 & 0.0039 & 0.0051 &          \\
			20191016 & 0.1905 & 0.2127 & 0.0739 & 0.1984 & 0.2159 & 0.0914 & 0.7615 & 0.9208 & 0.0558 & 0.0203 & 0.0251 & 0.0716 & -28294 \\
			& 0.0026 & 0.0063 & 0.0097 & 0.0096 & 0.0150  & 0.0122 & 0.0132 & 0.0490  & 0.0037 & 0.0044 & 0.0053 & 0.0062 &          \\
			20191017 & 0.1887 & 0.1894 & 0.0621 & 0.1223 & 0.1679 & 0.0646 & 0.4861 & 0.6641 & 0.0385 & 0.0217 & 0.0302 & 0.0491 & -28803 \\
			& 0.0078 & 0.0087 & 0.0089 & 0.0129 & 0.0198 & 0.0109 & 0.0401 & 0.0784 & 0.0047 & 0.0037 & 0.0052 & 0.0058 &          \\
			20191018 & 0.2252 & 0.1851 & 0.1294 & 0.1365 & 0.0879 & 0.0932 & 0.7626 & 0.4563 & 0.0484 & 0.0279 & 0.0220 & 0.0298 & -29896 \\
			& 0.0084 & 0.0088 & 0.0131 & 0.0170 & 0.0109 & 0.0095 & 0.0736 & 0.0419 & 0.0064 & 0.0045 & 0.0036 & 0.0041 &          \\
			\hline
		\end{tabular}
		
	\end{landscape}
	\clearpage
}



\end{document}